\documentclass[a4paper,UKenglish,cleveref, autoref, thm-restate]{lipics-v2021}

\pdfoutput=1

\usepackage[utf8]{inputenc}
\usepackage{cite}
\usepackage{amssymb,amsfonts,amsmath}
\usepackage{algorithmic}
\usepackage{graphicx}
\usepackage{textcomp}
\usepackage{xcolor}
\usepackage{mathtools}
\usepackage{calc}
\usepackage{physics}
\usepackage{hyperref}
\usepackage{cleveref}
\usepackage{multicol}
\usepackage{stmaryrd}
\usepackage{enumerate}
\usepackage[normalem]{ulem}
\usepackage{subcaption}
\usepackage{supertabular}
\usepackage{tabularx}
\usepackage{xspace}
\usepackage{nicefrac}

\usepackage{tikz}
\usepackage{pgfplots}
\usetikzlibrary{decorations.markings,decorations.pathmorphing, decorations.text,positioning}
\usetikzlibrary{backgrounds,folding}
\usetikzlibrary{shapes.geometric, shapes.misc}
\usetikzlibrary{arrows.meta}
\pgfdeclarelayer{edgelayer}
\pgfdeclarelayer{nodelayer}
\pgfsetlayers{background,edgelayer,nodelayer,main}
\tikzstyle{every picture}=[baseline=-0.25em]
\tikzset{every path/.style={draw=black!80, line width=0.6pt}}
\tikzstyle{none}=[inner sep=0mm]
\tikzstyle{black dot}=[inner sep=0.5mm,minimum width=0pt,minimum height=0pt,fill=black,draw=black,shape=circle]
\tikzstyle{dot}=[black dot]
\tikzstyle{white dot}=[inner sep=0.5mm,minimum width=0pt,minimum height=0pt,fill=white,draw=black,shape=circle]
\tikzstyle{white phase dot}=[minimum size=3.2mm, font={\footnotesize\boldmath}, shape=rectangle, rounded corners=1.3mm, inner sep=0.8mm, outer sep=-2mm, scale=0.8, draw=black, fill=white]
\tikzstyle{ket}=[minimum size=2mm, font={\footnotesize\boldmath}, shape=rectangle, rounded corners=1.2mm, inner sep=0.6mm, outer sep=-2mm, scale=0.8, draw=black, fill=black, text=white, minimum width=3mm, minimum height=3mm]
\tikzstyle{box}=[rectangle,fill=white,draw=black, font=\scriptsize, inner sep=2pt]
\tikzstyle{box-no-outline}=[rectangle, inner sep=2pt]
\tikzstyle{W-1-n}=[draw, trapezium, trapezium left angle=70, trapezium right angle=70, fill=black, inner sep=2pt]

\pgfdeclarelayer{edgelayer}
\pgfdeclarelayer{nodelayer}
\pgfsetlayers{background,edgelayer,nodelayer,main}
\tikzstyle{every loop}=[]

\def\scaling{1}

\def\fig{}

\newcommand{\eq}[2][~]{
	#1
	\underset{\substack{#2}}{=}
	#1
}

\newcommand{\interp}[1]{\left\llbracket #1 \right\rrbracket}

\newcommand{\cat}[1]{\mathbf{#1}}
\newcommand{\zw}{\ensuremath{\operatorname{ZW}_{\!d}}\xspace}
\newcommand{\zwf}{\ensuremath{\operatorname{ZW}_{\!\!f}}\xspace}

\newcommand{\obfhilb}[1]{\llparenthesis #1\rrparenthesis}

\allowdisplaybreaks
\nolinenumbers
	\title{Minimality in Finite-Dimensional ZW-Calculi}

\author{Marc de Visme}{Université Paris-Saclay, Inria, CNRS, ENS Paris-Saclay, Laboratoire Méthodes Formelles, 91190, Gif-sur-Yvette, France.}{marc.de-visme@inria.fr}{https://orcid.org/0009-0004-7227-7540}{}

\author{Renaud Vilmart}{Université Paris-Saclay, Inria, CNRS, ENS Paris-Saclay, Laboratoire Méthodes Formelles, 91190, Gif-sur-Yvette, France.}{renaud.vilmart@inria.fr}{https://orcid.org/0000-0002-8828-4671}{}

\authorrunning{M. de Visme \& R. Vilmart}

\Copyright{Marc de Visme and Renaud Vilmart}


\funding{This work is supported by the PEPR integrated project EPiQ ANR-22-PETQ-0007 part of Plan France 2030, by the projects ANR-22-PNCQ-0002 and ANR-22-CE47-0012.}

\acknowledgements{The authors would like to thank Antoine Guilmin-Crépon for discussions about the minimality of the present equational theories. The diagrams of the present paper were drawn using the TikZit tool \cite{tikzit}.}

\DOIPrefix{}

\hideLIPIcs

\makeatletter
\newcommand\xlabel[2][]{\phantomsection\def\@currentlabelname{\ensuremath{#1}}\label{#2}}
\makeatother

\begin{document}

\maketitle

\begin{abstract}
The ZW-calculus is a graphical language capable of representing 2-dimensional quantum systems (qubit) through its diagrams, and manipulating them through its equational theory. We extend the formalism to accommodate finite dimensional Hilbert spaces beyond qubit systems. 

First we define a qu\textbf{d}it version of the language, where all systems have the same arbitrary finite dimension $d$, and show that the provided equational theory is both complete -- i.e.~semantical equivalence is entirely captured by the equations -- and minimal -- i.e.~none of the equations are consequences of the others. We then extend the graphical language further to allow for mixed-dimensional systems. We again show the completeness and minimality of the provided equational theory.
\end{abstract}

\section{Introduction}

Graphical languages for quantum computations are a product of the categorical quantum mechanics program \cite{Abramsky2004categorical,Dodo} devoted to studying the foundations of quantum mechanics through the prism of category theory. These graphical languages come in different flavours, depending on which generators are used to build the diagrams (graphical representations of the quantum operators), and critically, displaying different kinds of interactions between said generators. The ZX-calculus describes the interaction between two complementary bases \cite{Coecke2011interacting}, the ZW-calculus, the interaction between the two ``spiders'' derived from the ``GHZ'' and ``W'' states, the only two fully entangled tripartite states up to SLOCC-equivalence \cite{Coecke2010compositional}, and the ZH-calculus the interaction between the same GHZ-state inferred spider and a spider obtained by generalising the Hadamard gate \cite{Backens2019ZH}.

The equations that describe these interactions form ``equational theories'' that define syntactic equivalence classes of diagrams, that are also semantically equivalent. When the syntactic equivalence matches perfectly the semantic one (i.e.~when two diagrams represent the same quantum operator iff they can be turned into one another), we say that the equational theory is \emph{complete}. Complete equational theories have been found for the aforementioned graphical languages, betimes for restrictions of them \cite{Backens2021ZHcompleteness,%
Backens2014ZX,Carette2019completeness,Carette2023compositionality,%
Hadzihasanovic2015diagrammatic,Hadzihasanovic2018complete,Jeandel2020completeness,%
Jeandel2018complete,Vilmart2019nearminimal}.

As is customary for the computer science part of the quantum computation community, the focus was largely set onto \emph{qubit} systems, i.e.~systems where the base quantum system is 2-dimensional, yet this is enough to get applications in optimisation \cite{Backens2020ThereAB,Kissinger2020reducing}, quantum error correction \cite{deBeaudrap2020zxlattice,Huang2023graphical,Townsend2023floquetifying}, verification \cite{Duncan2014verifying,MSc.Hillebrand}, simulation \cite{Kissinger2022Simulating,Kissinger2022classical,Koch2023speedy}... However, physics allows for qudit systems (where the base quantum system is $d$-dimensional with $d>2$) and even infinite dimensional systems. Several attempts have hence been made to go beyond the qubit case \cite{Booth2022complete,Hadzihasanovic2017algebra,Wang2018qutrit}, but it was only recently that a complete equational theory was found for $d$-dimensional (i.e.~qudit) systems \cite{Poor2023completeness} and later for finite dimensional systems (so-called ``qufinite'', i.e.~for the category $\cat{FdHilb}$) \cite{Wang2023completeness}. The results were obtained by generalising both the ZX and the ZW calculi and mixing them together. The W-node in particular allows for a neat intuitive (and unique) normal form for the diagrams. Satisfying the necessary conditions for every diagram to be normalisable then yields a complete equational theory. However, we believe that the ones obtained in \cite{Poor2023completeness,Wang2023completeness} are far from being \emph{minimal}, due in particular to the presence of the generators from both the ZX and the ZW calculi.

From a foundational perspective, it can be enlightening to know if an equation is a defining property of quantum systems (and hence necessary), or on the contrary if it is derivable from more fundamental properties (see e.g.~\cite{Collins2020Hopf,Duncan2016interacting}). The redundancy in the equational theory may also cause issues when trying to explore the space of equivalent diagrams, or to transport the completeness result to other diagrammatic languages for qudit systems (every equation has to be proven in the new language, hence the fewer the better); or when trying to generalise further, e.g.~to the $\cat{FdHilb}$ setting.

We argue here that the ZW-calculus is enough to get a natural normal form (akin to that of \cite{Poor2023completeness}) even in the qudit version, and provide an elegant equational theory that we show to be \emph{complete}, resorting to the normal form instead of transporting the completeness result from \cite{Poor2023completeness}, for the reason described above. We also show that the equational theory is minimal, meaning that none of the equations can be derived from the others, hence avoiding the aforementioned redundancy in the presentation.

We then adapt diagrams and the equational theory of the graphical language to accommodate all finite dimensional Hilbert spaces ($\cat{FdHilb}$), in a way that requires no additional generator and only one new equation. Here again we prove the completeness and the minimality of the equational theory, by leveraging that of the qudit setting.

The paper is split into two parts, \Cref{sec:qudit} and \Cref{sec:fdhilb}, devoted respectively to the $\cat{Qudit}_d$ version, and to the $\cat{FdHilb}$ version. In the $\cat{Qudit}_d$ version, diagrams and their interpretation are introduced in \Cref{sec:qudit-diags} and the equational theory is introduced and discussed in \Cref{sec:qudit-ET}. We then show its minimality in \Cref{sec:qudit-minimality} and its completeness in \Cref{sec:qudit-completeness}. In the $\cat{FdHilb}$ version, diagrams and their interpretation are introduced in \Cref{sec:fdhilb-diags}, and the equational theory is introduced and shown to be complete in \Cref{sec:fdhilb-ET}. All missing proofs are provided in the appendix.

\subsection*{The Dirac Notation}

\label{sec:dirac-notation}

All the upcoming diagrams can be given an interpretation as a linear map, in the appropriate category. In quantum information, it is usual to express such linear maps using the so-called Dirac notation. The current section hence serves as a gentle introduction to this notation.

Let $d\geq2$. In the $d$-dimensional Hilbert space $\mathbb C^d$, the canonical basis
\[\Big\lbrace
\begin{pmatrix}1&0&\cdots&0\end{pmatrix}^\intercal,~
\begin{pmatrix}0&1&\cdots&0\end{pmatrix}^\intercal,~\ldots,~
\begin{pmatrix}0&0&\cdots&1\end{pmatrix}^\intercal
\Big\rbrace\]
is usually denoted $\left\lbrace\ket0,\ket1,...,\ket{d-1}\right\rbrace$ (with $(.)^\intercal$ being the transpose). All 1-qudit systems have states that live in $\mathbb C^d$ and that can hence be represented by linear combinations of the elements of this basis: $\ket\psi = a_0\ket0+a_1\ket1+...+a_{d-1}\ket{d-1}$ (notice that the ``ket'' notation $\ket{.}$ is used for states in general, not only basis elements).

To combine systems, we use the tensor product (Kronecker product): $(.\otimes.)$ which is a fairly standard operation on linear maps. In particular, the overall state obtained by composing two 1-qudit systems in respective states $\ket\psi$ and $\ket\varphi$ is simply $\ket{\psi}\otimes\ket{\varphi}$. Notice that $\left\lbrace \ket i\otimes \ket j\right\rbrace_{0\leq i <d, 0\leq j < d'}$ forms a basis of $\mathbb C^d\otimes \mathbb C^{d'} \simeq \mathbb C^{d\times d'}$. It is customary to write $\ket{\psi,\varphi}$ to abbreviate $\ket{\psi}\otimes\ket{\varphi}$.

The ``bra'' notation $\bra.$ is used to represent the dagger (the conjugate transpose) of a state, i.e.~$\bra\psi = \ket\psi^\dag=\overline{\ket\psi}^\intercal$. The choice of the ``bra-ket'' notation is such that composing a bra with ket forms the bracket, the usual inner product in $\mathbb C^d$: 
$\bra\psi\circ\ket\varphi = \braket{\psi}{\varphi}$.
Linear combinations of ``ket-bras'' $\ket{i}\bra{j}$ of the canonical basis can be used to represent any linear map of the correct dimensions, e.g.~the 1-qudit identity: 
$\textit{id} = \sum\limits_{k=0}^{d-1} \ketbra k$.

\section{ZW-Calculus for Qudit Systems}

In this section, we introduce a graphical language for quantum systems that all have the same fixed dimension $d$ ($d\geq 2$): qudit systems.

\label{sec:qudit}

\subsection{Diagrams of $\cat{ZW}_d$ and their Interpretation}

\label{sec:qudit-diags}

First, we need to introduce the mathematical objects at the heart of the graphical language -- the diagrams -- and what they represent.

\subsubsection{The Diagrams}

As is traditional for graphical langu\-ages for finite-dimensional quantum systems, we work with a $\dagger$-compact prop \cite{Lack2004composing,PhD.Zanasi,selinger2010survey}. Categorically speaking, this is a symmetric, compact closed monoidal category generated by a single object, endowed with a contravariant endofunctor that behaves well with the symmetry and the compact structure. The following explains some of these concepts in more detail.

Let us denote $\cat{ZW}_d$ the $\dagger$-compact prop generated by:\\[0.5em]
\begin{minipage}{0.55\columnwidth}
\begin{itemize}
\setlength{\itemsep}{0.5em}
\item the Z-spiders $
\scalebox{\scaling}{
	\input{./figures/Z-spider.tikz}%
}:n\to m$, $\begin{aligned}[t]&\text{for }r\in\mathbb C\\[-1ex]&\text{and }n,m\geq0 \end{aligned}$
\item the state $\ket1$ $
\scalebox{\scaling}{
	\begin{tikzpicture}
	\begin{pgfonlayer}{nodelayer}
		\node [style=ket] (0) at (0, 0.125) {$1$};
		\node [style=none] (1) at (0, -0.25) {};
	\end{pgfonlayer}
	\begin{pgfonlayer}{edgelayer}
		\draw (1.center) to (0);
	\end{pgfonlayer}
\end{tikzpicture}
}:0\to1$
\end{itemize}
\end{minipage}
\begin{minipage}{0.44\columnwidth}
\begin{itemize}
\setlength{\itemsep}{0.5em}
\item the W-nodes $
\scalebox{\scaling}{
	\input{./figures/W-n.tikz}%
}:1\to n$, for $n\geq0$
\item the global scalars $r : 0\to 0$, for $r\in\mathbb C$
\end{itemize}
\end{minipage}\\
\begin{itemize}
\setlength{\itemsep}{0.5em}
\item the swap $
\scalebox{\scaling}{
	\input{./figures/swap.tikz}%
}:2\to2$ representing the symmetry of the prop
\item the cup $
\scalebox{\scaling}{
	\begin{tikzpicture}
	\begin{pgfonlayer}{nodelayer}
		\node [style=none] (0) at (0.5, 0.125) {};
		\node [style=none] (1) at (1, 0.125) {};
		\node [style=none] (2) at (0.75, -0.125) {};
	\end{pgfonlayer}
	\begin{pgfonlayer}{edgelayer}
		\draw [bend right=90, looseness=1.75] (0.center) to (1.center);
	\end{pgfonlayer}
\end{tikzpicture}
}:2\to0$ and cap $
\scalebox{\scaling}{
	\begin{tikzpicture}
	\begin{pgfonlayer}{nodelayer}
		\node [style=none] (0) at (-0.5, -0.125) {};
		\node [style=none] (1) at (0, -0.125) {};
		\node [style=none] (2) at (-0.25, 0.125) {};
	\end{pgfonlayer}
	\begin{pgfonlayer}{edgelayer}
		\draw [bend left=90, looseness=1.75] (0.center) to (1.center);
	\end{pgfonlayer}
\end{tikzpicture}
}:0\to2$ representing the compact structure
\item and the identity $
\scalebox{\scaling}{
	\begin{tikzpicture}
	\begin{pgfonlayer}{nodelayer}
		\node [style=none] (0) at (0, 0.25) {};
		\node [style=none] (1) at (0, -0.25) {};
	\end{pgfonlayer}
	\begin{pgfonlayer}{edgelayer}
		\draw (0.center) to (1.center);
	\end{pgfonlayer}
\end{tikzpicture}
}:1\to1$.
\end{itemize}
\smallskip
All these generators can be composed sequentially and in parallel, as follows:
\begin{align*}

\scalebox{\scaling}{
	\input{./figures/compo.tikz}%
}~:=~
\scalebox{\scaling}{
	\input{./figures/D2.tikz}%
}\circ
\scalebox{\scaling}{
	\input{./figures/D1.tikz}%
}
\qquad
\text{and}
\qquad

\scalebox{\scaling}{
	\input{./figures/D1.tikz}%
}~
\scalebox{\scaling}{
	\input{./figures/D2.tikz}%
}~:=~
\scalebox{\scaling}{
	\input{./figures/D1.tikz}%
}\otimes
\scalebox{\scaling}{
	\input{./figures/D2.tikz}%
}
\end{align*}
The symmetry and the compact structure satisfy the following identities:
\begin{align*}
\def\fig{swap-involution}
\scalebox{\scaling}{
	\begin{tikzpicture}
	\begin{pgfonlayer}{nodelayer}
		\node [style=white phase dot] (0)  at (0.45, 0.875) {$\star$};
		\node [style=W-1-n, shape border rotate=180] (1)  at (-0.3, -0.375) {};
		\node [style=none] (2)  at (-0.3, -1.125) {};
		\node [style=none, font={\scriptsize}] (3)  at (0.1, 0.275) {...};
		\node [style=none, font={\scriptsize}] (4)  at (-0.05, 0.875) {...};
		\node [style=white phase dot] (5)  at (-0.5, 0.875) {$\star$};
		\node [style=none, font={\scriptsize}] (6)  at (-0.55, 0.25) {...};
		\node [style=none, font={\scriptsize}, color=gray] (7)  at (-0.75, 1.0) {$1$};
		\node [style=none, font={\scriptsize}, color=gray] (8)  at (0.75, 1.0) {$1$};
		\node [style=none, font={\scriptsize}, color=gray] (9)  at (-0.425, -0.575) {$\delta$};
		\node [style=W-1-n, shape border rotate=0] (10)  at (-0.3, -0.875) {};
		\node [style=none, font={\scriptsize}, color=gray] (11)  at (-0.3, -1.25) {$a_i$};
		\node [style=none] (12)  at (-0.5, 1.25) {};
		\node [style=none] (13)  at (0.45, 1.25) {};
	\end{pgfonlayer}
	\begin{pgfonlayer}{edgelayer}
		\draw [bend left=15] (0) to (1);
		\draw [bend right=15] (0) to (1);
		\draw (2.center) to (1);
		\draw (5) to (1);
		\draw [bend right, looseness=1.25] (5) to (1);
		\draw (12.center) to (5);
		\draw (13.center) to (0);
	\end{pgfonlayer}
\end{tikzpicture}%
}
\eq[]{}\scalebox{\scaling}{
	\begin{tikzpicture}
	\begin{pgfonlayer}{nodelayer}
		\node [style=W-1-n, shape border rotate=180] (16)  at (-0.325, -0.25) {};
		\node [style=none] (17)  at (-0.075, -1.25) {};
		\node [style=none, font={\scriptsize}] (19)  at (-0.075, 1.0) {...};
		\node [style=white phase dot] (20)  at (-0.525, 1.0) {$\star$};
		\node [style=none, font={\scriptsize}] (21)  at (-0.575, 0.375) {...};
		\node [style=none, font={\scriptsize}, color=gray] (22)  at (-0.775, 1.125) {$1$};
		\node [style=none, font={\scriptsize}, color=gray] (24)  at (-0.45, -0.45) {$\delta$};
		\node [style=W-1-n, shape border rotate=0] (25)  at (-0.075, -1.0) {};
		\node [style=none, font={\scriptsize}, color=gray] (26)  at (-0.075, -1.375) {$a_i$};
		\node [style=none] (27)  at (-0.525, 1.375) {};
		\node [style=white phase dot] (29)  at (0.475, 1.0) {$\star$};
		\node [style=W-1-n, shape border rotate=180] (30)  at (0.225, -0.25) {};
		\node [style=none, font={\scriptsize}] (31)  at (0.3, 0.4) {...};
		\node [style=none, font={\scriptsize}, color=gray] (33)  at (0.775, 1.125) {$1$};
		\node [style=none] (34)  at (0.475, 1.375) {};
		\node [style=W-1-n, shape border rotate=180] (35)  at (-0.075, -0.65) {};
		\node [style=none, font={\scriptsize}, color=gray] (36)  at (0.15, -0.525) {$\delta$};
		\node [style=none, font={\scriptsize}, color=gray] (37)  at (0.1, -0.825) {$\delta$};
	\end{pgfonlayer}
	\begin{pgfonlayer}{edgelayer}
		\draw (16) to (35);
		\draw (20) to (16);
		\draw [bend right, looseness=1.25] (20) to (16);
		\draw (27.center) to (20);
		\draw [bend left=15] (29) to (30);
		\draw [bend right] (29) to (30);
		\draw (34.center) to (29);
		\draw (35) to (30);
		\draw (35) to (17.center);
	\end{pgfonlayer}
\end{tikzpicture}%
}
\hspace*{3.5em}
\def\fig{symmetry-naturality}
\scalebox{\scaling}{
}
\eq[]{}\scalebox{\scaling}{
}
\hspace*{3.5em}
\def\fig{swapped-cap}
\scalebox{\scaling}{
}
\eq[]{}\scalebox{\scaling}{
}
\hspace*{3.5em}
\def\fig{snake}
\scalebox{\scaling}{
}
\eq[]{}\scalebox{\scaling}{
}
\eq[]{}\scalebox{\scaling}{
	\begin{tikzpicture}
	\begin{pgfonlayer}{nodelayer}
		\node [style=W-1-n, shape border rotate=180] (39)  at (-0.325, -0.25) {};
		\node [style=none] (40)  at (-0.075, -1.25) {};
		\node [style=none, font={\scriptsize}] (41)  at (-0.075, 1.0) {...};
		\node [style=white phase dot] (42)  at (-0.525, 1.0) {$\star$};
		\node [style=none, font={\scriptsize}] (43)  at (-0.575, 0.375) {...};
		\node [style=none, font={\scriptsize}, color=gray] (44)  at (-0.775, 1.125) {$1$};
		\node [style=none, font={\scriptsize}, color=gray] (45)  at (-0.525, -0.45) {$\delta$};
		\node [style=none, font={\scriptsize}, color=gray] (47)  at (-0.075, -1.375) {$a_i$};
		\node [style=none] (48)  at (-0.525, 1.375) {};
		\node [style=white phase dot] (49)  at (0.475, 1.0) {$\star$};
		\node [style=W-1-n, shape border rotate=180] (50)  at (0.225, -0.25) {};
		\node [style=none, font={\scriptsize}] (51)  at (0.3, 0.4) {...};
		\node [style=none, font={\scriptsize}, color=gray] (52)  at (0.775, 1.125) {$1$};
		\node [style=none] (53)  at (0.475, 1.375) {};
		\node [style=W-1-n, shape border rotate=180] (57)  at (-0.075, -0.975) {};
		\node [style=W-1-n, shape border rotate=0] (58)  at (-0.325, -0.625) {};
		\node [style=W-1-n, shape border rotate=0] (59)  at (0.225, -0.625) {};
		\node [style=none, font={\scriptsize}, color=gray] (60)  at (0.4, -0.45) {$\delta$};
		\node [style=none, font={\scriptsize}, color=gray] (61)  at (0.175, -0.875) {$a_i$};
		\node [style=none, font={\scriptsize}, color=gray] (62)  at (-0.325, -0.875) {$a_i$};
	\end{pgfonlayer}
	\begin{pgfonlayer}{edgelayer}
		\draw (39) to (58);
		\draw (42) to (39);
		\draw [bend right, looseness=1.25] (42) to (39);
		\draw (48.center) to (42);
		\draw [bend left=15] (49) to (50);
		\draw [bend right] (49) to (50);
		\draw (53.center) to (49);
		\draw (57) to (40.center);
		\draw (57) to (59);
		\draw (58) to (57);
		\draw (59) to (50);
	\end{pgfonlayer}
\end{tikzpicture}%
}
\end{align*}
This compact structure in particular allows us to define the "upside-down" version of the generators, for instance: 
\def\fig{W-n-1}
$
\scalebox{\scaling}{
}
~~:=~~\scalebox{\scaling}{
}
\quad
\text{and}
\quad
\def\fig{bra-1}
\scalebox{\scaling}{
}
~~:=~~\scalebox{\scaling}{
}
$

The $\dag$ functor is defined inductively as:\\
\[ \begin{array}{rcl}
(D_2\circ D_1)^\dag &=& D_1^\dag\circ D_2^\dag\\
(D_1\otimes D_2)^\dag &=& D_1^\dag\otimes D_2^\dag\\
\left(r\right)^\dag &=& \bar r\\
\end{array}\qquad \left(~
\scalebox{\scaling}{
	\input{./figures/Z-spider.tikz}%
}~\right)^\dag =~ 
\scalebox{\scaling}{
	\input{./figures/dag-Z.tikz}%
} \qquad \begin{array}{p{4cm}}with the other generators being mapped to their upside-down version.\end{array}\]
\noindent
Notice that thank to the identities satisfied by the $\dag$-compact prop, the $\dag$-functor is involutive.

As will be made clearer in what follows, in $\cat{ZW\!}_d$, $d\geq2$ represents the dimension of the "base" quantum system, called qudit. As this $d$ will be fixed in the following, we may forget to specify it. For convenience, we define an empty white node as a parameter-$1$ Z-spider: 
\def\fig{Z-spider-empty}
$
\scalebox{\scaling}{
}
:=\scalebox{\scaling}{
}
$
and give the $0\to1$ W-node a special symbol, akin to that of $\ket1$ (as its interpretation, as we will see later, is merely $\ket0$): 
\def\fig{ket-0-def}
$
\scalebox{\scaling}{
}
~:=~\scalebox{\scaling}{
}
$. 
We generalise the ket symbol inductively as follows (for $2\leq k<d$): 
\def\fig{ket-k-def-induct}
$
\scalebox{\scaling}{
}
~~:=~~\scalebox{\scaling}{
}
$. 
These last symbols can be given an upside-down definition using the compact structure as was done for $\ket1$ and the W-node.

\subsubsection{The Interpretation}

The point of the diagrams of the $\cat{ZW\!}_d$ is to represent quantum operators on multipartite $d$-dimensional systems. The way those are usually specified is thanks to the category $\cat{Qudit}_d$. This forms again a symmetric $\dag$-compact prop, where the base object is $1:=\mathbb C^d$, and morphisms $n\to m$ are linear maps $\mathbb C^{d^n}\to \mathbb C^{d^m}$. The symmetry and the compact structure correspond to their counterparts in $\cat{ZW}_d$, they will be stated out in the following, as part of the interpretation of $\cat{ZW\!}_d$ diagrams. The $\dag$ functor is the usual $\dag$ of linear maps over $\mathbb C$.

We may hence interpret diagrams of the $\cat{ZW\!}_d$-calculus thanks to the functor $\interp{.}:\cat{ZW\!}_d\to \cat{Qudit}_d$ inductively defined as follows:\\
\begin{minipage}{0.42\columnwidth}
\begin{align*}
\interp{D_2\circ D_1}
&= \interp{D_2}\circ\interp{D_1}\\
\interp{D_1\otimes D_2}
&= \interp{D_1}\otimes\interp{D_2}\\
\interp{~
\scalebox{\scaling}{
	\begin{tikzpicture}
	\begin{pgfonlayer}{nodelayer}
		\node [style=none] (0) at (0, 0.25) {};
		\node [style=none] (1) at (0, -0.25) {};
	\end{pgfonlayer}
	\begin{pgfonlayer}{edgelayer}
		\draw (0.center) to (1.center);
	\end{pgfonlayer}
\end{tikzpicture}
}~}
&=\sum_{k} \ketbra k\\
\interp{
\scalebox{\scaling}{
	\input{./figures/swap.tikz}%
}}
&=\sum_{k,\ell} \ketbra{\ell,k}{k,\ell}\\
\interp{
\scalebox{\scaling}{
	\begin{tikzpicture}
	\begin{pgfonlayer}{nodelayer}
		\node [style=none] (0) at (-0.5, -0.125) {};
		\node [style=none] (1) at (0, -0.125) {};
		\node [style=none] (2) at (-0.25, 0.125) {};
	\end{pgfonlayer}
	\begin{pgfonlayer}{edgelayer}
		\draw [bend left=90, looseness=1.75] (0.center) to (1.center);
	\end{pgfonlayer}
\end{tikzpicture}
}}
&=\interp{
\scalebox{\scaling}{
	\begin{tikzpicture}
	\begin{pgfonlayer}{nodelayer}
		\node [style=none] (0) at (0.5, 0.125) {};
		\node [style=none] (1) at (1, 0.125) {};
		\node [style=none] (2) at (0.75, -0.125) {};
	\end{pgfonlayer}
	\begin{pgfonlayer}{edgelayer}
		\draw [bend right=90, looseness=1.75] (0.center) to (1.center);
	\end{pgfonlayer}
\end{tikzpicture}
}}^\dag=\sum_{k} \ket{k,k}
\end{align*}
\end{minipage}
\hspace*{-1em}
\begin{minipage}{0.59\columnwidth}
\begin{align*}
\interp{
\scalebox{\scaling}{
	\input{./figures/Z-spider.tikz}%
}}
&= \sum_{k=0}^{d-1} r^k \sqrt{k!}^{n+m-2}\ketbra{k^m}{k^n}\\
\interp{
\scalebox{\scaling}{
	\input{./figures/W-n.tikz}%
}}
&= \!\sum_{\substack{k\in\{0,...,d-1\}\\i_1+...+i_n=k}}\! \sqrt{\binom{k}{i_1,...,i_n}}\ketbra{i_1,...,i_n}{k}\\
\interp{
\scalebox{\scaling}{
	\begin{tikzpicture}
	\begin{pgfonlayer}{nodelayer}
		\node [style=ket] (0) at (0, 0.125) {$1$};
		\node [style=none] (1) at (0, -0.25) {};
	\end{pgfonlayer}
	\begin{pgfonlayer}{edgelayer}
		\draw (1.center) to (0);
	\end{pgfonlayer}
\end{tikzpicture}
}}
&= \ket 1\\
\interp{r}&=r
\end{align*}
\end{minipage}\\
where $\displaystyle\binom{k}{i_1,...,i_n} = \frac{k!}{i_1!...i_n!}$ is a multinomial coefficient. 
Notice that the interpretation of the $0\to1$ W-node is simply: 
$
\def\fig{ket-0-def}
\interp{\scalebox{\scaling}{
}} = \sum\limits_{\substack{k\in\{0,...,d-1\}\\0=k}} \sqrt{\binom{k}{0}}\ket{k} = \ket0
$, 
and that of the black node symbol $k$ for $k<d$ is $\ket k$ up to renormalisation: 
$\interp{
\scalebox{\scaling}{
	\begin{tikzpicture}
	\begin{pgfonlayer}{nodelayer}
		\node [style=ket] (0) at (-0.25, 0.25) {$k$};
		\node [style=none] (1) at (-0.25, -0.25) {};
	\end{pgfonlayer}
	\begin{pgfonlayer}{edgelayer}
		\draw (1.center) to (0);
	\end{pgfonlayer}
\end{tikzpicture}
}}
=\sqrt{\binom{k}{1,...,1}}\ket k
= \sqrt{k!} \ket k$. 
The presence of $\sqrt{\cdots}$ on the coefficients is not particularly relevant, and is simply an artefact of us maintaining some symmetry between generators and their dagger. Indeed, we want \def\fig{braket-k}$\interp{\scalebox{\scaling}{
}}=k!$, and for that the coefficient $k!$ needs to be split between both nodes, resulting in either an asymmetric presentation or a square root (see \Cref{sec:asymmetric-presentation} for an equivalent semantics, without any $\sqrt{\cdots}$ but asymmetric instead).

Notice also that the interpretation of the Z-spider differs from more usual generalisations of its qubit counterpart, because of the $\sqrt{k!}^{n+m-2}$ which depends on the degree of the spider. While it makes the interpretation of the diagrams slightly more complicated, it allows us -- as will be stated later -- to quite conveniently generalise equations from the qubit ZW-calculus, and hence have a simpler equational theory.
It will be shown in the following (\Cref{cor:universality}), that the above set of generators makes for a universal calculus, i.e.~any linear map of $\cat{Qudit}_d$ can be represented by a $\cat{ZW\!}_d$-diagram.

To gain intuition about the upcoming equations between diagrams, it can be useful to \emph{semantically} decompose a diagram into sums of simpler ones\footnote{Notice that here, such decompositions are merely semantical. The upcoming completeness is only interested in equivalence between single diagrams.}. To do so, it can be convenient to understand $\ket k$ as a bunch of $k$ indistinguishable particles:\\
\begin{minipage}{0.56\columnwidth}
\def\fig{ket-k-on-W}
\begin{align}
\label{eq:ket-distrib-on-W}
\interp{\scalebox{\scaling}{
}}
&=\hspace*{-1em}\sum_{i_1+...+i_n=k}\binom{k}{i_1,...,i_n}\interp{\scalebox{\scaling}{
}}\\
\def\fig{ket-sum}
\label{eq:ket-sum}
\interp{\scalebox{\scaling}{
	\input{./figures/\fig/\fig_12.tikz}%
}}
&=\def\fig{ket-sum}
\begin{cases}
\interp{\scalebox{\scaling}{
	\input{./figures/\fig/\fig_14.tikz}%
}} & \text{ if }k+\ell<d\\
\vec 0 & \text{ if }k+\ell\geq d
\end{cases}
\end{align}
\end{minipage}
\hfill
\begin{minipage}{0.38\columnwidth}
\def\fig{copy-k}
\begin{align}
\label{eq:copy}
\interp{\scalebox{\scaling}{
}}
&=r^k\interp{\scalebox{\scaling}{
	\begin{tikzpicture}
	\begin{pgfonlayer}{nodelayer}
		\node [style=none] (94)  at (-0.325, -1.0) {};
		\node [style=none, font={\scriptsize}] (95)  at (-0.075, 0.75) {...};
		\node [style=white phase dot] (96)  at (-0.525, 0.75) {$\star$};
		\node [style=none, font={\scriptsize}] (97)  at (-0.575, 0.125) {...};
		\node [style=none, font={\scriptsize}, color=gray] (98)  at (-0.775, 0.875) {$1$};
		\node [style=none, font={\scriptsize}, color=gray] (100)  at (-0.325, -1.125) {$a_i$};
		\node [style=none] (101)  at (-0.525, 1.125) {};
		\node [style=white phase dot] (102)  at (0.475, 0.75) {$\star$};
		\node [style=W-1-n, shape border rotate=180] (103)  at (-0.325, -0.5) {};
		\node [style=none, font={\scriptsize}] (104)  at (0.05, 0.15) {...};
		\node [style=none, font={\scriptsize}, color=gray] (105)  at (0.775, 0.875) {$1$};
		\node [style=none] (106)  at (0.475, 1.125) {};
		\node [style=W-1-n, shape border rotate=180] (107)  at (-0.325, -0.475) {};
	\end{pgfonlayer}
	\begin{pgfonlayer}{edgelayer}
		\draw (96) to (103);
		\draw [bend right, looseness=1.25] (96) to (103);
		\draw (101.center) to (96);
		\draw [bend left=15] (102) to (103);
		\draw [bend right] (102) to (103);
		\draw (106.center) to (102);
		\draw (107) to (94.center);
	\end{pgfonlayer}
\end{tikzpicture}%
}}\\
\def\fig{braket-k-l}
\label{eq:braket}
\interp{\scalebox{\scaling}{
}}
&=k!\braket{k}{\ell}\\
\label{eq:id-decomp}
\interp{
\scalebox{\scaling}{
	\begin{tikzpicture}
	\begin{pgfonlayer}{nodelayer}
		\node [style=none] (0) at (0, 0.25) {};
		\node [style=none] (1) at (0, -0.25) {};
	\end{pgfonlayer}
	\begin{pgfonlayer}{edgelayer}
		\draw (0.center) to (1.center);
	\end{pgfonlayer}
\end{tikzpicture}
}} &= \sum_{k=0}^{d-1} \frac1{k!}\interp{
\scalebox{\scaling}{
	\begin{tikzpicture}
	\begin{pgfonlayer}{nodelayer}
		\node [style=ket] (0) at (0, 0.175) {$k$};
		\node [style=none] (1) at (0, 0.425) {};
		\node [style=ket] (2) at (0, -0.175) {$k$};
		\node [style=none] (3) at (0, -0.425) {};
	\end{pgfonlayer}
	\begin{pgfonlayer}{edgelayer}
		\draw (1.center) to (0);
		\draw (3.center) to (2);
	\end{pgfonlayer}
\end{tikzpicture}
}}
\end{align}
\end{minipage}\\[1em]
\Cref{eq:ket-distrib-on-W} explains how the W-node spreads the $k$ ``particles'' that enter it following a multinomial distribution. \Cref{eq:ket-sum} shows that the $2\to1$ W-node takes two bunches of particles $k$ and $\ell$ and regroups them into one, and yields the null state if $k+\ell$ exceeds the ``capacity'' (i.e.~the dimension) of a single wire. 
This will be proven graphically (\Cref{lem:ket-sum}) from the upcoming equational theory (\Cref{fig:equational-theory}). When $k+\ell<d$, the fact that there is no additional scalar is due to the rescaling of the $k$-dots to represent $\sqrt{k!}\ket k$. This rescaling also makes the ``copy'' more natural: The Z-spider $1\to n$ copies any bunch of $k$ particles entering it, yielding global scalar $r^k$ in the process, as is shown by \Cref{eq:copy}. 
This will again be proven graphically (\Cref{lem:copy}) from the equational theory. 
The rescaling, however, forces \Cref{eq:braket}. 
Finally, it can be useful to decompose the identity as a linear combination of products of kets and bras as is done in \Cref{eq:id-decomp}.

\subsection{Equational Theory}

\label{sec:qudit-ET}

With the above interpretation of the $\cat{ZW\!}_d$, different diagrams may yield the same linear map. All axioms of symmetric $\dag$-compact props in particular preserve the interpretation. More generally, we may want to relate together \emph{all} diagrams that have the same semantics. This is done through an equational theory, i.e.~a set of equations that can be applied locally in a diagram without changing the semantics of the whole.

\subsubsection{Equations of the $\cat{ZW\!}_d$-Calculus}

On top of the axioms of symmetric $\dag$-compact props, we assume some conventional equations about the topology of the generators, which should align with the symmetries of the symbols used to depict them. The Z-spider does not distinguish between any of its connections: it is ``flexsymmetric''\cite{Carette2021wielding}, meaning that we can interchange any of its legs without changing the semantics. Graphically, for any permutation of wires $\sigma$: 
\def\fig{Z-flex}
$
\scalebox{\scaling}{
}
\eq{}\scalebox{\scaling}{
}
$. 
On the other hand, the binary $W$-node is only co-com\-mu\-tative,
which, together with the upcoming Equation \nameref{ax:W-assoc}, means that all the outputs of the $n$-ary W-node can be exchanged, i.e.~for any permutation of wires $\sigma$: 
\def\fig{W-n-commutative}
$
\scalebox{\scaling}{
}
\eq{}\scalebox{\scaling}{
}
$. 
With all this in place, we can give the core of the equational theory, in \Cref{fig:equational-theory}. When diagram $D_1$ can be turned into diagram $D_2$ using the rules of \zw, we write $\zw\vdash D_1=D_2$.

\begin{figure*}[!htb]
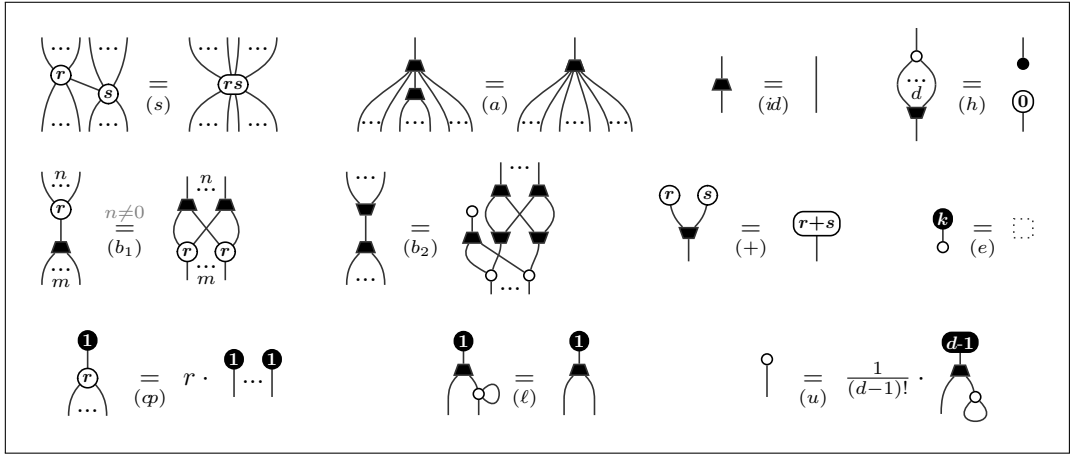

\boxed{
\begin{minipage}{0.984\textwidth}
\medskip
\hspace*{0.0em}
\def\fig{Z-spider-rule}
$\label{ax:Z-spider}\xlabel[(s)]{ax:Z-spider}
\scalebox{\scaling}{
	\input{./figures/\fig/\fig_00.tikz}%
}
\eq[]{(s)}\scalebox{\scaling}{
	\input{./figures/\fig/\fig_01.tikz}%
}$
\hfill
\def\fig{W-assoc}
$\label{ax:W-assoc}\xlabel[(a)]{ax:W-assoc}
\scalebox{\scaling}{
	\input{./figures/\fig/\fig_00.tikz}%
}
\!\!\eq[]{(a)}\!\!\scalebox{\scaling}{
	\input{./figures/\fig/\fig_01.tikz}%
}$
\hfill
\def\fig{W-unit-aux}
$\label{ax:W-unit}\xlabel[(i\hspace*{-1pt}d)]{ax:W-unit}
\scalebox{\scaling}{
	\input{./figures/\fig/\fig_00.tikz}%
}
\eq[~]{(i\hspace*{-1pt}d)}\scalebox{\scaling}{
	\input{./figures/\fig/\fig_01.tikz}%
}$
\hfill
\def\fig{Hopf-aux}
$\label{ax:hopf}\xlabel[(h)]{ax:hopf}
\scalebox{\scaling}{
	\input{./figures/\fig/\fig_00.tikz}%
}
\!\eq[~]{(h)}\scalebox{\scaling}{
	\input{./figures/\fig/\fig_01.tikz}%
}$
\hspace*{0.0em}
\medskip

\hspace*{.0em}
\def\fig{Z-W-bialgebra}
$\label{ax:bialgebra-Z-W}\xlabel[(b_1)]{ax:bialgebra-Z-W}
\scalebox{\scaling}{
	\input{./figures/\fig/\fig_00.tikz}%
}
\overset{{\color{gray}n\neq0}}{\eq[~]{(b_1)}}\scalebox{\scaling}{
	\input{./figures/\fig/\fig_01.tikz}%
}$
\hfill
\def\fig{W-bialgebra}
$\label{ax:bialgebra-W}\xlabel[(b_2)]{ax:bialgebra-W}
\scalebox{\scaling}{
	\input{./figures/\fig/\fig_00.tikz}%
}
\!\!\eq[~]{(b_2)}\!\scalebox{\scaling}{
	\input{./figures/\fig/\fig_01.tikz}%
}$
\hfill
\def\fig{sum}
$\label{ax:sum}\xlabel[(+)]{ax:sum}
\scalebox{\scaling}{
	\input{./figures/\fig/\fig_00.tikz}%
}
\!\!\eq[~]{(+)}\scalebox{\scaling}{
	\input{./figures/\fig/\fig_01.tikz}%
}$
\hfill
$\label{ax:one}\xlabel[(e)]{ax:one}

\scalebox{\scaling}{
	\input{./figures/erasing-ket.tikz}%
}\eq[]{(e)}
\scalebox{\scaling}{
	\input{./figures/empty.tikz}%
}$
\hspace*{.0em}
\bigskip

\hspace*{1em}
\def\fig{copy-1}
$\label{ax:copy}\xlabel[(c\!p)]{ax:copy}
\scalebox{\scaling}{
	\input{./figures/\fig/\fig_04.tikz}%
}
\eq{(c\!p)} r\cdot \scalebox{\scaling}{
	\input{./figures/\fig/\fig_07.tikz}%
}$
\hfill
\def\fig{Z-loop-removal}
$\label{ax:loop}\xlabel[(\ell)]{ax:loop}
\scalebox{\scaling}{
	\input{./figures/\fig/\fig_00.tikz}%
}\hspace*{-0.7em}
\eq{(\ell)}\scalebox{\scaling}{
	\input{./figures/\fig/\fig_01.tikz}%
}$
\hfill
\def\fig{alternative-Z-decomp}
$\label{ax:decomp-Z}\xlabel[(u)]{ax:decomp-Z}
\scalebox{\scaling}{
	\input{./figures/\fig/\fig_00.tikz}%
}
\eq{(u)}\frac1{(d{-}1)!}\cdot\scalebox{\scaling}{
	\input{./figures/\fig/\fig_01.tikz}%
}$
\hspace*{1em}
\medskip
\end{minipage}}
\caption{Equational theory \zw for the qudit $\cat{ZW}$-calculus.}
\label{fig:equational-theory}
\end{figure*}

\begin{remark}
In this framework, we can tensor global scalars together $r\otimes s$, which graphically could be confused with their product $rs$. This is actually unambiguous in the equational theory, as, using Equation \nameref{ax:copy}:
\def\fig{scalar-multiplication}
\begin{align*}
\zw\vdash~r\otimes s
\eq{\nameref{ax:copy}}r\cdot\scalebox{\scaling}{
}
\eq{\nameref{ax:copy}}\scalebox{\scaling}{
}
\eq{\nameref{ax:Z-spider}}\scalebox{\scaling}{
}
\eq{\nameref{ax:copy}}rs
\end{align*}
Moreover, using Equation~\nameref{ax:one}, one can easily show that global scalar $1$ is the empty diagram (\Cref{lem:scalar-1}). 
Scalar multiplication is assumed to be automatically applied, and scalar $1$ is assumed to be automatically removed in the following after Lemma \ref{lem:scalar-1} is proven.
\end{remark}

All rules up to \nameref{ax:one} are fairly standard generalisations of rules of the qubit $\cat{ZW}$-calculus (with \nameref{ax:bialgebra-W} being inspired from \cite{Poor2023completeness} to avoid using a \emph{fermionic swap}\footnote{The fermionic swap, introduced in \cite{Hadzihasanovic2015diagrammatic}, is a generator that in many situations behave like the actual swap. The qubit version of ZW uses the fermionic swap, but this generator loses many properties (for instance its involution, or the fact that it maps $\ket{k,\ell}$ to $\ket{\ell,k}$ up to some coefficient) when going in larger dimensions. By not using it here, we avoid having to axiomatise it.}). 
The non-conventional $\sqrt k ^{n+m-2}$ coefficients in the interpretation of the Z-spider seem to be necessary for Equations \nameref{ax:Z-spider}, \nameref{ax:bialgebra-Z-W} and \nameref{ax:sum} to all work. Notice that this makes the Z-spider \emph{non-special}, meaning that: 
\def\fig{spider-non-special}
$
\scalebox{\scaling}{
}
\neq~~\scalebox{\scaling}{
}
$. 
Equation \nameref{ax:loop} however gives a context in which that inequality becomes an equality. 
Finally, Equation \nameref{ax:decomp-Z} shows how a $0\to1$ Z-spider can be obtained by distributing $d-1$ ``particles'' over two paths, and erasing (or post-selecting) adequately one of the two paths.

\begin{remark}
Thanks to the compact structure and its interaction with the generators of the language, all upside-down version of the equations of \Cref{fig:equational-theory} are derivable, by simply deforming the diagrams to get the actual axiom. For instance, the upside-down version of \nameref{ax:sum} can be derived as follows:
\def\fig{upside-down-sum}
\begin{align*}
\zw\vdash\scalebox{\scaling}{
}
\eq{}\scalebox{\scaling}{
}
\eq{}\scalebox{\scaling}{
}
\eq{\nameref{ax:sum}}\scalebox{\scaling}{
	\begin{tikzpicture}
	\begin{pgfonlayer}{nodelayer}
		\node [style=W-1-n, shape border rotate=180] (64)  at (-0.325, -0.375) {};
		\node [style=none] (65)  at (-0.075, -1.125) {};
		\node [style=none, font={\scriptsize}] (66)  at (-0.075, 0.875) {...};
		\node [style=white phase dot] (67)  at (-0.525, 0.875) {$\star$};
		\node [style=none, font={\scriptsize}] (68)  at (-0.475, 0.25) {...};
		\node [style=none, font={\scriptsize}, color=gray] (69)  at (-0.775, 1.0) {$1$};
		\node [style=none, font={\scriptsize}, color=gray] (71)  at (-0.075, -1.25) {$a_i$};
		\node [style=none] (72)  at (-0.525, 1.25) {};
		\node [style=white phase dot] (73)  at (0.475, 0.875) {$\star$};
		\node [style=W-1-n, shape border rotate=180] (74)  at (0.225, -0.375) {};
		\node [style=none, font={\scriptsize}] (75)  at (0.525, 0.275) {...};
		\node [style=none, font={\scriptsize}, color=gray] (76)  at (0.775, 1.0) {$1$};
		\node [style=none] (77)  at (0.475, 1.25) {};
		\node [style=W-1-n, shape border rotate=180] (78)  at (-0.075, -0.85) {};
		\node [style=none, font={\scriptsize}, color=gray] (82)  at (0.175, -0.75) {$a_i$};
		\node [style=none, font={\scriptsize}, color=gray] (83)  at (-0.325, -0.75) {$a_i$};
		\node [style=W-1-n, shape border rotate=0] (84)  at (-0.725, 0.25) {};
		\node [style=W-1-n, shape border rotate=0] (85)  at (-0.225, 0.25) {};
		\node [style=W-1-n, shape border rotate=0] (86)  at (0.275, 0.25) {};
		\node [style=W-1-n, shape border rotate=0] (87)  at (0.775, 0.25) {};
		\node [style=none, font={\scriptsize}, color=gray] (88)  at (-0.65, -0.1) {$1$};
		\node [style=none, font={\scriptsize}, color=gray] (89)  at (-0.175, -0.1) {$1$};
		\node [style=none, font={\scriptsize}, color=gray] (90)  at (0.15, -0.1) {$1$};
		\node [style=none, font={\scriptsize}, color=gray] (91)  at (0.6, -0.175) {$1$};
	\end{pgfonlayer}
	\begin{pgfonlayer}{edgelayer}
		\draw (64) to (78);
		\draw (67) to (85);
		\draw (72.center) to (67);
		\draw (73) to (87);
		\draw (74) to (78);
		\draw (77.center) to (73);
		\draw (78) to (65.center);
		\draw (84) to (64);
		\draw (84) to (67);
		\draw (85) to (64);
		\draw (86) to (74);
		\draw (86) to (73);
		\draw (87) to (74);
	\end{pgfonlayer}
\end{tikzpicture}%
}
\eq{}\scalebox{\scaling}{
}
\end{align*}
\end{remark}

\begin{proposition}
\label{prop:soundness}
All equations in \zw are sound, i.e.:
\[\zw\vdash D_1=D_2\implies \interp{D_1}=\interp{D_2}\]
\end{proposition}

\begin{proof}
This is a straightforward verification for most of the equations. They can all be proven using the aforementioned identities (Equations~\ref{eq:ket-distrib-on-W} to \ref{eq:id-decomp}) in the semantics of the diagrams, especially the decomposition of the identity (Equation \ref{eq:id-decomp}). Equations~\nameref{ax:bialgebra-W} and \nameref{ax:sum} require respectively the Vandermonde identity and  the binomial formula:
\[
\sum_{k_1+...+k_p=m} \binom{n_1}{k_1}...\binom{n_p}{k_p} = \binom{n_1+...+n_p}{m}
\quad\text{ and }\quad(r+s)^n = \sum_{k=0}^n\binom{n}{k}r^ks^{n-k}.\qedhere\]
\end{proof}

\subsection{Minimality}

\label{sec:qudit-minimality}

Minimality of an equational theory states that every single equation is necessary: none can be derived from the others. Said otherwise, as soon as we remove one of the equations, some equalities (that were previously provable) become unprovable. Minimality is fundamental, as it allows us to pinpoint properties that are necessary to our model, and a contrario those that are consequences of the necessary ones. Notice however that there is usually not a single minimal equational theory, for two reasons: 1) it often happens that one equation can be replaced by an equivalent one, and 2) one could start with a totally different set of equations. Most of the equations we chose are generalisations or adapatations of equations that were already used by other graphical languages, and are usually motivated by either categorical or physical considerations. We only ``filled in the blanks'' with equations \nameref{ax:loop} and \nameref{ax:decomp-Z}.

In trying to prove minimality, it often happens that two equations fail to be proven necessary individually, but that the pair (i.e.~at least one of the two) can be proven necessary. Such cases underline some sort of proximity between the two equations, and the obstacle it poses to minimality can sometimes be circumvented (somewhat artificially) by merging them into a single, potentially slightly less intuitive, equation. This happened once here: we merged equations \def\fig{Hopf}
$\scalebox{\scaling}{
}\eq[]{}\scalebox{\scaling}{
}$ and \def\fig{ket-0}$\scalebox{\scaling}{
}\eq[]{}\scalebox{\scaling}{
}$, that we initially had as axioms, into Equation \nameref{ax:hopf}. This finally provides us with a minimal equational theory for qudit ZW-calculus.

To prove that an equation is necessary, we define a non-standard interpretation which is preserved by all the equations (including the axioms of $\dag$-compact props), except the equation of interest. When such an interpretation is exhibited, we can safely conclude that the equation is necessary, since if it were a consequence of the others, it would also preserve this interpretation. Interpretations like these sometimes simply take the form of a quantity that turns out to be invariant for all equations except the one that is considered. In the realm of quantum graphical theories, such arguments were used for single equations in \cite{Duncan2009graph,%
Schroder2014incomplete,%
Perdrix2016supplementarity,%
Jeandel2018Ycalculus,%
Jeandel2017Cyclotomic}, partial minimality results were obtained for Clifford ZX-calculus \cite{Backens2020towards}, unrestricted ZX-calculus \cite{Vilmart2019nearminimal}, and quite recently, full minimality (with completeness) was obtained for quantum circuits \cite{Clement2023minimal}.

In the following, we show that the equational theory \zw from \Cref{fig:equational-theory} is minimal, i.e.~that none of the equations can be derived from the others. It is to be noted that most of the equations in \Cref{fig:equational-theory} are schemas, that is they are parametrised, and  the equation is assumed for all possible values of the parameters. Our minimality result is ``weak'' in the sense that for each equation schema, we show that \emph{at least} one of the occurrences cannot be derived from the other equations, but we do not pinpoint for which parameters the equation is necessary or not. Nevertheless:

\begin{theorem}
\label{thm:minimality}
All equations in \zw are necessary, hence \zw is minimal.
\end{theorem}

Several arguments in the proof are graphical, meaning that they rely on the understanding that a diagram can also be seen as a special kind of graph with vertices (inputs, outputs, Z-spiders, W-nodes, states, global scalars) and edges (wires, cup, cap, with swaps representing crossings). As such, when we talk about ``a W-node in a diagram $D$'', we actually refer to a vertex (of the Z type) of said diagram. One argument in the proof require the notion of ``effective Z-path'', that is a path that goes only through Z-spiders and trivial W-nodes -- which W-nodes that could be replaced by identity wires without it making any difference -- and that can be used in non-trivial computations.
\begin{definition}[Sole effective output of a W-node]
	Let $D$ be a diagram with a collections of $n$ distinguished $W$-node. Let's call `$a_1$',$\dots$,`$a_n$' one output of each, as shown in \Cref{eq:single-out-W} below.
	We say that wires `$a_1$',$\dots$,`$a_n$' are \emph{jointly} the \emph{sole effective outputs} of the W-nodes if \Cref{eq:sole-effective-output} is satisfied:\\[-1em]
	\begin{equation}
	\label{eq:single-out-W}
	
\scalebox{\scaling}{
	\input{./figures/D.tikz}%
} = 
\scalebox{\scaling}{
	\input{./figures/sole-effective-output-with.tikz}%
}
	\end{equation}
	\begin{equation}
	\interp{
\scalebox{\scaling}{
	\input{./figures/sole-effective-output-without.tikz}%
}}=\interp{
\scalebox{\scaling}{
	\input{./figures/sole-effective-output-ket0.tikz}%
}}
	\end{equation}
\end{definition}
In particular, if we only consider one W-node that has a sole effective output -- we call such a node a \emph{trivial W-node} -- it follows that:
\def\fig{sole-effective-output}
\begin{equation}
\label{eq:sole-effective-output}
\interp{\scalebox{\scaling}{
}} = \interp{\scalebox{\scaling}{
}}
\end{equation}
\begin{definition}[Effective Z-path]
	An effective Z-path of a diagram $D$ is a path going from a boundary to another (inputs and/or outputs), such that:
	\begin{itemize}
		\item For each W-node it goes through, it cannot use two outputs of the W-node (so it must use its input). Considering all those W-nodes at once, those outputs must be jointly the sole effective outputs of those W-nodes.
		\item There exists a state $\ket{\phi}$ that is not of the form $\lambda \ket{0}$ or $\lambda \ket{1}$ for some $\lambda \in \mathbb{C}$, that \emph{inhabits} the path. That is, if we feed to $\interp{D}$ the state $\ket{\phi}$ and/or the effect $\bra{\phi}$ to the two inputs and/or outputs of the path, then the result is not $\vec 0$. 
	\end{itemize}
\end{definition}

We can now prove \Cref{thm:minimality}:
\begin{proof}
We consider each of these equations individually, and give more details in \Cref{sec:proof_minimality}:
\begin{itemize}
\setlength{\itemsep}{0.2em}
\item[\nameref{ax:Z-spider}] When applying the transformation that turns all Z-spider parameters and global scalars to their real part ($r\mapsto \Re(r)$), Equation~\nameref{ax:Z-spider} is the only one that is not preserved.
\item[\nameref{ax:W-assoc}] It is the only equation permitting to create non-trivial W-nodes with arity $>d$ from a diagram that only has non-trivial W-nodes with arity $\leq d$.
\item[\nameref{ax:W-unit}] It is the only equation that can create nodes connected to boundaries from a node-free diagram. 
\item[\nameref{ax:hopf}] To each wire in a diagram $D$, we associate a number $0\leq k < d$ (or more graphically we annotate each wire by some number $k$)\footnote{The annotations produced here are upper bounds on the value of the kets that can go through the wires. As such, they are very closely related to the \emph{capacities} from \Cref{sec:fdhilb}.}. The procedure to do so is as follows:
\begin{enumerate}
	\item annotate all wires with $d-1$
	\item rewrite the annotations using the following rules, until a fixed point is reached:\\
	$
	\labelitemii~~ \def\fig{ket-0-ofs}\scalebox{\scaling}{
}~\overset{a\neq0}\to~\scalebox{\scaling}{
}
	\hspace*{3.5em}
	\labelitemii~~ \def\fig{ket-1-ofs}\scalebox{\scaling}{
}~\overset{a>1}\to~\scalebox{\scaling}{
}
	\hspace*{3.5em}
	\labelitemii~~ \def\fig{W-node-ofs}\scalebox{\scaling}{
}\to\scalebox{\scaling}{
}
	$\\
	$
	\labelitemii~~ \def\fig{Z-spider-ofs}\scalebox{\scaling}{
}~\to~\scalebox{\scaling}{
}\quad \text{ with } a:=\min(a_i)
	$
\end{enumerate}
This simple procedure obviously terminates, as a step is only applied if at least one of the annotations is decreased. By considering inputs and outputs of $D$ (which are the only wires that can be guaranteed to remain during rewrites with \zw), we can check that Equation \nameref{ax:hopf} is the only one that can modify the outcome of the procedure.
\item[\nameref{ax:bialgebra-Z-W}] Consider diagrams as graphs, and use the above definition of an ``effective Z-path''. All equations except \nameref{ax:bialgebra-Z-W} preserve the existence of effective Z-paths, although proving it for \nameref{ax:bialgebra-W} is quite involved and relies on the following lemma proved in the appendix:
	\begin{lemma}\label{lem:double-W-simplification}~\\
		Let $D$ be a diagram of the form shown in \Cref{eq:double-W}, such that we have \Cref{eq:double-W-simplification}.
		\begin{equation}
		\label{eq:double-W}
		
\scalebox{\scaling}{
	\input{./figures/D.tikz}%
} = 
\scalebox{\scaling}{
	\input{./figures/lem-double-W.tikz}%
}
		\end{equation}
		\begin{equation}
		\label{eq:double-W-simplification}
		\interp{
\scalebox{\scaling}{
	\input{./figures/D.tikz}%
}} = \interp{
\scalebox{\scaling}{
	\input{./figures/lem-W-above.tikz}%
}} = \interp{
\scalebox{\scaling}{
	\input{./figures/lem-W-below.tikz}%
}}
		\end{equation}
		Then \Cref{eq:double-W-without} is satisfied:
		\begin{equation}
		\label{eq:double-W-without}
		\interp{
\scalebox{\scaling}{
	\input{./figures/lem-W-without.tikz}%
}} = \interp{
\scalebox{\scaling}{
	\input{./figures/lem-W-ket0.tikz}%
}}
		\end{equation}
\end{lemma}
\item[\nameref{ax:bialgebra-W}] Consider diagrams as graphs, and define a ``W-path'' in the diagram as a path 1) that goes from a boundary to another boundary, 2) which cannot use two outputs of a W-node (if it goes through a W-node, it has to use the input edge) and 3) that does not go through a Z-spider. All equations, except \nameref{ax:bialgebra-W}, preserve the existence of a W-path. \nameref{ax:bialgebra-W} is the only equation that can bring the number of W-paths from non-zero to zero (which is done by adding a Z-spider on the path).
\item[\nameref{ax:sum}] 
Take the interpretation that maps all the Z-spider parameters (and the global scalars) to their absolute value ($r\mapsto |r|$). This interpretation preserves all equations except \nameref{ax:sum}.
\item[\nameref{ax:one}] It is the only equation that can create generators out of empty diagrams.
\item[\nameref{ax:copy}] It is the only equation that can create a non-real scalar from a diagram with only real scalars.
\item[\nameref{ax:loop}] Let $\varpi:=e^{i\frac\pi{d-1}}$ be the first $(2d-2)$-th root of unity. Consider the $\dagger$-compact monoidal functor (i.e.~functor that preserves compositions, symmetry and compact structure) that maps the generators as follows:\\
$
\labelitemii~~ 
\scalebox{\scaling}{
	\begin{tikzpicture}
	\begin{pgfonlayer}{nodelayer}
		\node [style=ket] (0) at (0, 0.125) {$1$};
		\node [style=none] (1) at (0, -0.25) {};
	\end{pgfonlayer}
	\begin{pgfonlayer}{edgelayer}
		\draw (1.center) to (0);
	\end{pgfonlayer}
\end{tikzpicture}
}\mapsto \varpi
\scalebox{\scaling}{
	\begin{tikzpicture}
	\begin{pgfonlayer}{nodelayer}
		\node [style=ket] (0) at (0, 0.125) {$1$};
		\node [style=none] (1) at (0, -0.25) {};
	\end{pgfonlayer}
	\begin{pgfonlayer}{edgelayer}
		\draw (1.center) to (0);
	\end{pgfonlayer}
\end{tikzpicture}
}
\hspace*{5em}
\labelitemii~~ 
\scalebox{\scaling}{
	\input{./figures/Z-spider.tikz}%
}\mapsto 
\scalebox{\scaling}{
	\input{./figures/minimality-l-Z-spider.tikz}%
}
\hspace*{5em}
\labelitemii~~ 
\scalebox{\scaling}{
	\input{./figures/W-n.tikz}%
}\mapsto
\scalebox{\scaling}{
	\input{./figures/W-n.tikz}%
}
$\\
When $d>2$, all equations but \nameref{ax:loop} are preserved by this functor. We can make the argument work when $d=2$ by choosing any $\varpi$ such that $\varpi^2\neq1$ and by working up to a scalar factor.
\item[\nameref{ax:decomp-Z}] Take the interpretation that maps $
\scalebox{\scaling}{
	\begin{tikzpicture}
	\begin{pgfonlayer}{nodelayer}
		\node [style=ket] (0) at (0, 0.125) {$1$};
		\node [style=none] (1) at (0, -0.25) {};
	\end{pgfonlayer}
	\begin{pgfonlayer}{edgelayer}
		\draw (1.center) to (0);
	\end{pgfonlayer}
\end{tikzpicture}
}$ (and subsequently all $
\scalebox{\scaling}{
	\begin{tikzpicture}
	\begin{pgfonlayer}{nodelayer}
		\node [style=ket] (0) at (-0.25, 0.25) {$k$};
		\node [style=none] (1) at (-0.25, -0.25) {};
	\end{pgfonlayer}
	\begin{pgfonlayer}{edgelayer}
		\draw (1.center) to (0);
	\end{pgfonlayer}
\end{tikzpicture}
}$) to $\def\fig{ket-0}\scalebox{\scaling}{
}$. All rules hold (up to a the scalar factor) except \nameref{ax:decomp-Z}.\qedhere
\end{itemize}
\end{proof}

\subsection{Completeness}

\label{sec:qudit-completeness}

Completeness of an equational theory with respect to a semantics is the fundamental property that ensures that semantical equivalence of diagrams is entirely captured by the equational theory. Minimality is worthless without some form of completeness, as it is extremely simple to design minimal, but not complete, equational theories. For instance, the empty equational theory (that contains no axiom) is minimal but clearly not complete for the qudit ZW-diagrams. We hence show in this section that we indeed have completeness.

\subsubsection{Normal Form and Universality}

The usual way to prove completeness is to show that any diagram can be put in a normal form, and that this normal form is unique and similar for all equivalent diagrams. As is customary in a category that is compact-closed, we can focus on states, as there is an isomorphism between operators and states \cite{Abramsky2004categorical}: 
\def\fig{map-state}
\begin{align*}
\scalebox{\scaling}{
}
\eq{}\scalebox{\scaling}{
}
~~=:~~\scalebox{\scaling}{
}
\end{align*}

Proving completeness requires a fair amount of diagrammatic derivations, especially when starting from a minimal equational theory, to get enough material to define a normalisation strategy. These lemmas and their proofs are postponed to the appendix, pages \pageref{sec:proofs-qudit} to \pageref{sec:appendix-fdhilb}.
\bigskip

\noindent
\begin{minipage}{0.6\columnwidth}
\begin{definition}
We define $\mathcal N:\cat{Qudit}_d\to\cat{ZW\!}_d$ as the functor that maps any $n$-qudit state\\
$\hspace*{1.5em}\ket{\psi} = r_0\ket{0...0}+r_1\ket{0...01}+...+r_i\ket{x_1^i...x_n^i}+...$\\
to the diagram on the right.\\
We say of any diagram in the image of $\mathcal N$ that it is in normal form.
\end{definition}
\end{minipage}
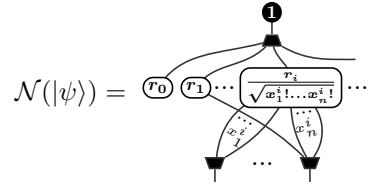
\begin{minipage}{0.39\columnwidth}
\phantom{.}\hfill$\mathcal N(\ket\psi)=
\scalebox{\scaling}{
	\input{./figures/NF.tikz}%
}$
\end{minipage}

\bigskip\noindent
This construction is a direct generalisation of the normal form of the qubit ZW-diagrams~\cite{Hadzihasanovic2015diagrammatic}, which is also considered in~\cite{Hadzihasanovic2017algebra} in the context of $q$-arithmetic. It creates a diagram whose interpretation is the starting state:
\begin{proposition}
\label{prop:NF-preserves-semantics}
$\forall \ket\psi\in\cat{Qudit}_d[0,n],~\interp{\mathcal N(\ket\psi)} = \ket\psi$.
\end{proposition}

Where for any category $\mathcal{C}$, we write $\mathcal{C}[n,m]$ for the set all the morphisms from $n$ to $m$. As a simple consequence of this proposition, any qudit operator can be represented by a diagram of $\cat{ZW\!}_d$:
\begin{corollary}[Universality]
\label{cor:universality}
$\forall f \in \cat{Qudit}_d[n,m],~\exists D_f\in\cat{ZW\!}_d[n,m],~
\interp{D_f} = f$.
\end{corollary}

Since we defined the normal form as the image of a map from the semantics, any diagram can only be associated to a unique normal form.

\subsubsection{Completeness}

Our goal now is to show that any diagram can be put in normal form. To do so, we show that all generators can be put in normal form, and that all compositions of diagrams in normal form can be put in normal form.

We start by showing the latter for the tensor product:
\begin{proposition}
\label{prop:NF-tensor}
The spatial composition of diagrams in normal form can be put in normal form, i.e.~$\zw\vdash\mathcal{N}(v_1)\otimes\mathcal{N}(v_2)=\mathcal{N}(v_1\otimes v_2)$.
\end{proposition}

When turning arbitrary operators into states, the sequential composition turns into the application of cups $
\scalebox{\scaling}{
	\begin{tikzpicture}
	\begin{pgfonlayer}{nodelayer}
		\node [style=none] (0) at (0.5, 0.125) {};
		\node [style=none] (1) at (1, 0.125) {};
		\node [style=none] (2) at (0.75, -0.125) {};
	\end{pgfonlayer}
	\begin{pgfonlayer}{edgelayer}
		\draw [bend right=90, looseness=1.75] (0.center) to (1.center);
	\end{pgfonlayer}
\end{tikzpicture}
}$ onto pairs of outputs of the state, as:
\def\fig{NF-compo}
\begin{align*}
\scalebox{\scaling}{
}
\eq{}\scalebox{\scaling}{
}
\eq{\textit{Prop}.~\ref{prop:NF-tensor}}\scalebox{\scaling}{
}
\end{align*}

\begin{proposition}
\label{prop:NF-cup}
The diagram obtained by applying a cup $
\scalebox{\scaling}{
	\begin{tikzpicture}
	\begin{pgfonlayer}{nodelayer}
		\node [style=none] (0) at (0.5, 0.125) {};
		\node [style=none] (1) at (1, 0.125) {};
		\node [style=none] (2) at (0.75, -0.125) {};
	\end{pgfonlayer}
	\begin{pgfonlayer}{edgelayer}
		\draw [bend right=90, looseness=1.75] (0.center) to (1.center);
	\end{pgfonlayer}
\end{tikzpicture}
}$ to two outputs of a diagram in normal form can be put in normal form.
\end{proposition}

Then we move on to showing that all the generators can be put in normal form. To do so, the following lemma will prove useful:

\begin{lemma}
\label{lem:W21-to-NF}
The diagram obtained by applying \scalebox{0.8}{$
\scalebox{\scaling}{
	\input{./figures/W21.tikz}%
}$} to two outputs of a normal form can be put in normal form.
\end{lemma}

\begin{proposition}
\label{prop:NF-generators}
All generators of the $\cat{ZW\!}_d$-calculus can be put in normal form.
\end{proposition}

Putting all the latter results together, we can show the completeness of the language:
\begin{theorem}[Completeness for Qudit Systems]
\label{thm:completeness-qudit}
The language is complete: for any two diagrams $D_1$ and $D_2$ of the $\cat{ZW\!}_d$-calculus:
\[\interp{D_1}=\interp{D_2}\iff \zw\vdash D_1=D_2\]
\end{theorem}

\begin{proof}
By \Cref{prop:NF-generators}, any generator of the language can be put in normal form. Thanks to Propositions~\ref{prop:NF-tensor} and~\ref{prop:NF-cup}, compositions of diagrams in normal form can be put in normal form. As a consequence, any diagram can be put in normal form. By uniqueness of this normal form, if two diagrams share the same semantics, they can be rewritten into the same diagram. This proves completeness.
\end{proof}

\section{Finite Dimensional Hilbert Spaces}

\label{sec:fdhilb}

In the previous setting, all systems are required to be $d$-dimensional for some fixed $d$. Here we relax that constraint, which allows us to go ``mixed-dimensional'' and to represent morphisms of $\cat{FdHilb}$\footnote{Technically, the skeleton of $\cat{FdHilb}$, i.e.~where all $d$-dimensional Hilbert spaces are identified with the canonical representative $\mathbb C^d$. We take the liberty in this paper to name $\cat{FdHilb}$ this skeleton.}.

$\cat{FdHilb}$ is the strict symmetric monoidal $\dag$-compact category of finite dimensional Hilbert spaces \cite{Abramsky2004categorical}. Its objects are tensor products of finite dimensional Hilbert spaces $\mathbb C^d (d \in \mathbb N\setminus\{0\})$, and its morphisms are linear maps between them. The symmetry and the compact structure are naturally extended from that of $\cat{Qudit}_d$.

In this new setting, we will be able to represent \emph{all} morphisms of $\cat{FdHilb}$, at the cost of annotating the wires of the diagrams to keep track of their dimensions. Instead of the dimension itself, we rather annotate the wire with its dimension $-1$, i.e.~with the largest $k$ such that $\ket k$ is allowed on the wire. We call such $k$ the \emph{capacity} of the wire. This makes the bookkeeping a little bit less tedious.

\subsection{Diagrams and Interpretation}

\label{sec:fdhilb-diags}

We also require the following constraints for the capacities around each generator:
\begin{itemize}
\item All capacities around a Z-spider are the same
\item The input capacity of the W-node must be larger than (or equal to) each of its outputs
\end{itemize}
The first constraint follows from the fact that Z-spiders in ZW can be seen as a generalisation of graph edges -- more precisely they can be seen as hyperedges. Hence the whole hyperedge should have a single capacity. The second constraint simply comes from the fact that a larger capacity on the outputs of a W-node will never be used, so we might as well prevent it. When considering $1\to1$ W-nodes, which represent projections, this restriction allows us to see at a glance which side has the largest dimension.

The first restriction further allows us to put the capacity annotation on the Z-spider rather than on all its legs, making annotating diagrams less cumbersome.

We now work with a $\dagger$-compact symmetric monoidal category, which is not a prop anymore. Our base objects are $\mathbb C^{d}$ for $d\in\mathbb N\setminus\{0\}$. Every pair of objects can be composed with $\otimes$ to form a third object, with $\otimes$ being associative, and with the tensor unit $I$ being $I:=\mathbb C^1$. We work with a strict monoidal category, so we consider $I\otimes \mathbb C^d = \mathbb C^d = \mathbb C^d\otimes I$. To simplify notations, we represent objects $\mathbb C^{d_1}\otimes\mathbb C^{d_2}\otimes ... \otimes \mathbb C^{d_n}$ by a list of capacities $\obfhilb{d_1-1,d_2-1,...,d_n-1}$ (the tensor product simply becomes the concatenation of lists). The tensor unit is represented by $
\obfhilb{}$. Since $\obfhilb{0} =\obfhilb{}$, we forbid $0$ capacities on the wires.

In this new setting, the generators are generalised as follows:
\begin{itemize}
\setlength{\itemsep}{1em}
\item Z-spiders $
\scalebox{\scaling}{
	\input{./figures/Z-spider-annot.tikz}%
}:\obfhilb{\overbrace{a,...,a}^n} \to \obfhilb{\overbrace{a,...,a}^m}$ with $r\in\mathbb C$ and $a\geq1$
\item W-nodes $
\scalebox{\scaling}{
	\input{./figures/W-n-annot.tikz}%
}: \obfhilb{a} \to \obfhilb{b_1,...,b_n}$ with $a\geq \underset{1\leq i \leq n}{\max}(b_i)$ and $b_i\geq1$
\item state $\ket1$ $
\scalebox{\scaling}{
	\input{./figures/ket-1-annot.tikz}%
}:\obfhilb{} \to \obfhilb{a}$ with $a\geq 1$
\item global scalars $r:\obfhilb{} \to\obfhilb{}$ with $r\in\mathbb C$
\end{itemize}
with the symmetry and the compact structure being generalised to

\noindent
\begin{minipage}{0.4\columnwidth}
\begin{itemize}
\item $
\scalebox{\scaling}{
	\input{./figures/swap-annot.tikz}%
}:\obfhilb{a,b} \to\obfhilb{b,a}$
\end{itemize}
\end{minipage}
\begin{minipage}{0.59\columnwidth}
\begin{itemize}
\item $
\scalebox{\scaling}{
	\input{./figures/cap-annot.tikz}%
}:\obfhilb{}\to\obfhilb{a,a}$ \quad and \quad $
\scalebox{\scaling}{
	\input{./figures/cup-annot.tikz}%
}:\obfhilb{a,a}\to\obfhilb{}$,
\end{itemize}
\end{minipage}\\
(again with $a,b\geq1$) and the identity to $
\scalebox{\scaling}{
	\input{./figures/id-annot.tikz}%
}:\obfhilb{a}\to\obfhilb{a}$ ($a\geq1$).

Diagrams can still be composed together both sequentially and in parallel. The sequential composition prevents us from composing diagrams with unmatched objects (e.g.~two Z-spiders with different capacities in sequence). Diagrams with capacities are called $\cat{ZW_{\!f}}$-diagrams, and are graphical representations of the morphisms of the $\dag$-compact symmetric monoidal category $\cat{FdHilb}$ (the dagger functor can be given in $\cat{ZW_{\!f}}$ in a similar way as $\cat{ZW\!}_d$).

The compact structure still allows us to define upside-down versions of the W-node and the kets. Again, we give the $\obfhilb{}\to\obfhilb{a}$ W-node a special symbol (for $a\geq1$): 
\def\fig{ket-0-annot}
$
\scalebox{\scaling}{
}~~:=~~\scalebox{\scaling}{
}
$, 
and we generalise the ket symbol inductively as follows (for $k\geq1$ and $a > k$): 
\def\fig{ket-k-def-qf-induct}
$
\scalebox{\scaling}{
}
~~:=~~\scalebox{\scaling}{
}
$.

The interpretation of these diagrams is now a monoidal functor $\interp{.}:\cat{ZW_{\!f}}\to\cat{FdHilb}$ inductively defined as:\vspace*{-1.5em}\\
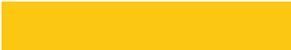
\begin{minipage}[t]{0.75\columnwidth}
\begin{align*}
\interp{D_2\circ D_1}
&= \interp{D_2}\circ\interp{D_1}\\
\interp{D_1\otimes D_2}
&= \interp{D_1}\otimes\interp{D_2}\\
\interp{~
\scalebox{\scaling}{
	\input{./figures/id-annot.tikz}%
}~}
&=\sum_{k=0}^a \ketbra k\\
\interp{
\scalebox{\scaling}{
	\input{./figures/cap-annot.tikz}%
}}
&=\interp{
\scalebox{\scaling}{
	\input{./figures/cup-annot.tikz}%
}}^\dag=\sum_{k=0}^a \ket{k,k}\\
\interp{r}&=r\\
\interp{
\scalebox{\scaling}{
	\input{./figures/W-n-annot.tikz}%
}}
&= \hspace*{-1em}\sum_{\substack{0\leq k_i \leq b_i\\ k_1{+}...{+}k_n \leq a}} \hspace*{-1em}\sqrt{\binom{k_1{+}...{+}k_n}{k_1,...,k_n}}\ketbra{k_1,...,k_n}{k_1{+}...{+}k_n}
\end{align*}
\end{minipage}
\hspace*{-10em}
\begin{minipage}[t]{0.48\columnwidth}
\begin{align*}
\interp{
\scalebox{\scaling}{
	\input{./figures/swap-annot.tikz}%
}}
&=\sum_{k=0}^a\sum_{\ell=0}^b \ketbra{\ell,k}{k,\ell}\\
\interp{
\scalebox{\scaling}{
	\input{./figures/Z-spider-annot.tikz}%
}}
&= \sum_{k=0}^a r^k \sqrt{k!}^{n+m-2}\ketbra{k^m}{k^n}\\
\interp{
\scalebox{\scaling}{
	\input{./figures/ket-1-annot.tikz}%
}}
&= \ket1
\end{align*}
\end{minipage}

By composition, one can check that, for $k>1$ and $a\geq k$: 
$
\interp{
\scalebox{\scaling}{
	\input{./figures/ket-k-annot.tikz}%
}}
= \sqrt{k!}\ket k
$.

\noindent
Notice that we use the same notation for the interpretation of $\cat{ZW\!}_d$-diagrams, and for the interpretation of the $\cat{ZW\!_f}$-diagrams. Which interpretation we are referring to should be clear from the context.

\subsection{Complete Equational Theory}

\label{sec:fdhilb-ET}

We once again equip the language with an equational theory \zwf, defined in \Cref{fig:equational-theory-fdhilb}. This equational theory only slightly differs from the one for qudit systems in \Cref{fig:equational-theory}. It is interesting to notice that 1) the associativity of the W-node is broken down into two equations \nameref{ax:W-assoc-qf} and \nameref{ax:W-assoc-2-qf}, whose choice depends on the capacities involved, 2) the W-bialgebra equation \nameref{ax:bialgebra-W-qf} does not need a context anymore, but instead side conditions on the capacities, 3) we managed to remove Equation \nameref{ax:one}, 4) we now need an equation \nameref{ax:ket-1} that states that a $\ket1$, when ``injected'' into a larger dimensional Hilbert space, is still a $\ket1$.

We also notice the existence of an interesting equation, that we did not include in \Cref{fig:equational-theory-fdhilb} as it turns out to be derivable; and which states that a $0\to1$ Z-state can be ``copied'' by the W-node, as follows: 
\def\fig{Z-copy-annot}$
\scalebox{\scaling}{
}
\eq{}\scalebox{\scaling}{
}$.

\begin{figure*}[!htb]
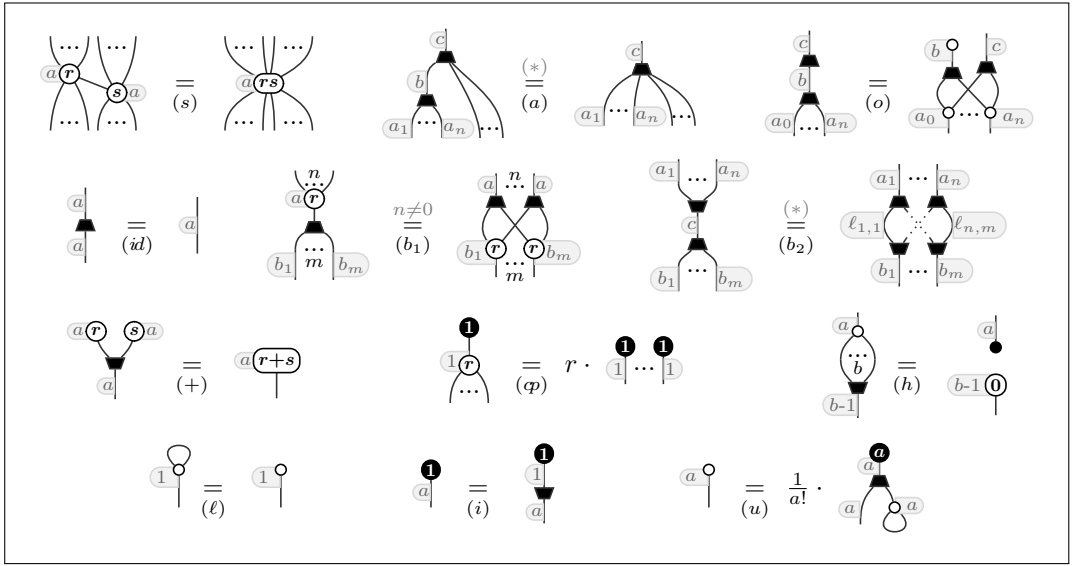

\boxed{
\begin{minipage}{0.984\textwidth}
\medskip
\hspace*{0em}
\def\fig{Z-spider-rule-annot}
$\label{ax:Z-spider-qf}\xlabel[(s)]{ax:Z-spider-qf}
\scalebox{\scaling}{
	\input{./figures/\fig/\fig_00.tikz}%
}
\eq[~]{(s)}\scalebox{\scaling}{
	\input{./figures/\fig/\fig_01.tikz}%
}$
\hfill
\def\fig{W-assoc-annot}
$\label{ax:W-assoc-qf}\xlabel[(a)]{ax:W-assoc-qf}
\scalebox{\scaling}{
	\input{./figures/\fig/\fig_00.tikz}%
}
\overset{{\color{gray}(*)}}{\eq[]{(a)}}~\scalebox{\scaling}{
	\input{./figures/\fig/\fig_01.tikz}%
}$
\hfill
\def\fig{W-assoc-annot2}
$\label{ax:W-assoc-2-qf}\xlabel[(o)]{ax:W-assoc-2-qf}
\scalebox{\scaling}{
	\input{./figures/\fig/\fig_00.tikz}%
}
\!\!\eq[]{(o)}\scalebox{\scaling}{
	\input{./figures/\fig/\fig_01.tikz}%
}$
\hspace*{0em}
\medskip

\hspace*{1em}
\def\fig{Z-id-annot}
$\label{ax:id-qf}\xlabel[(i\hspace*{-1pt}d)]{ax:id-qf}
\scalebox{\scaling}{
	\input{./figures/\fig/\fig_02.tikz}%
}
\eq[~]{(i\hspace*{-1pt}d)}\scalebox{\scaling}{
	\input{./figures/\fig/\fig_01.tikz}%
}$
\hfill
\def\fig{Z-W-bialgebra-annot}
$\label{ax:bialgebra-Z-W-qf}\xlabel[(b_1)]{ax:bialgebra-Z-W-qf}
\scalebox{\scaling}{
	\input{./figures/\fig/\fig_00.tikz}%
}
\overset{{\color{gray}n\neq0}}{\eq[~]{(b_1)}}\scalebox{\scaling}{
	\input{./figures/\fig/\fig_01.tikz}%
}$
\hfill
\def\fig{W-bialgebra-annot}
$\label{ax:bialgebra-W-qf}\xlabel[(b_2)]{ax:bialgebra-W-qf}
\scalebox{\scaling}{
	\input{./figures/\fig/\fig_00.tikz}%
}
\overset{{\color{gray}(*)}}{\eq{(b_2)}}\scalebox{\scaling}{
	\input{./figures/\fig/\fig_01.tikz}%
}$
\hspace*{1em}
\medskip

\hspace*{1em}
\def\fig{sum-annot}
$\label{ax:sum-qf}\xlabel[(+)]{ax:sum-qf}
\scalebox{\scaling}{
	\input{./figures/\fig/\fig_00.tikz}%
}
\!\!\!\eq[~]{(+)}\scalebox{\scaling}{
	\input{./figures/\fig/\fig_01.tikz}%
}$
\hfill
\def\fig{copy-1-annot}
$\label{ax:scalar-qf}\xlabel[(c\!p)]{ax:scalar-qf}
\scalebox{\scaling}{
	\input{./figures/\fig/\fig_00.tikz}%
}
\eq[~]{(c\!p)} r\cdot \scalebox{\scaling}{
	\input{./figures/\fig/\fig_01.tikz}%
}$
\hfill
\def\fig{Hopf-annot}
$\label{ax:hopf-qf}\xlabel[(h)]{ax:hopf-qf}
\scalebox{\scaling}{
	\input{./figures/\fig/\fig_00.tikz}%
}
\eq[]{(h)}~\scalebox{\scaling}{
	\input{./figures/\fig/\fig_01.tikz}%
}$
\hspace*{1em}
\medskip

\hspace*{4em}
\def\fig{Z-loop-removal-annot}
$\label{ax:loop-removal-qf}\xlabel[(\ell)]{ax:loop-removal-qf}
\scalebox{\scaling}{
	\input{./figures/\fig/\fig_00.tikz}%
}
\!\!\!\!\!\!\eq[]{(\ell)}~\scalebox{\scaling}{
	\input{./figures/\fig/\fig_01.tikz}%
}$
\hfill
\def\fig{ket-1-dim-change}
$\label{ax:ket-1}\xlabel[(i)]{ax:ket-1}
\scalebox{\scaling}{
	\input{./figures/\fig/\fig_00.tikz}%
}
\eq[~]{(i)}~\scalebox{\scaling}{
	\input{./figures/\fig/\fig_01.tikz}%
}$
\hfill
\def\fig{alternative-Z-decomp-qf}
$\label{ax:decomp-Z-qf}\xlabel[(u)]{ax:decomp-Z-qf}
\scalebox{\scaling}{
	\input{./figures/\fig/\fig_00.tikz}%
}
\eq[~]{(u)}\frac1{a!}\cdot\scalebox{\scaling}{
	\input{./figures/\fig/\fig_01.tikz}%
}$
\hspace*{4em}
\medskip
\end{minipage}}
\caption{Equational theory \zwf for the finite-dimensional $\cat{ZW}$-calculus. In \nameref{ax:W-assoc-qf}, we require that $b=c$ or $b\geq \sum_i a_i$; and in \nameref{ax:bialgebra-W-qf} that $c\geq\min(\sum a_i, \sum b_i)$ on the lhs, and that $\ell_{ij}=\min(a_i,b_j)$ on the rhs.}
\label{fig:equational-theory-fdhilb}
\end{figure*}

The category $\cat{Qudit}_d$ is a full subcategory of $\cat{FdHilb}$, and as such there is an obvious inclusion functor $\cat{Qudit}_d \overset{i_d}\hookrightarrow \cat{FdHilb}$. This inclusion transports to the ZW-calculi: we can turn any $\cat{ZW\!}_d$-diagram into a $\cat{ZW_{\!f}}$-diagram through $\iota_d$ in such a way that the following diagram commutes: 
$
\scalebox{\scaling}{
	\input{./figures/qudit-2-qufinite-cd.tikz}%
}$.\\
The functor $\iota_d$ simply takes a $\cat{ZW\!}_d$-diagram and annotates all its wires with $d-1$.

We show that the present equational theory is complete. To do so, we need to adapt the notion of normal form from qudit systems (and we again use the map/state duality to focus on states rather than arbitrary morphisms):
\bigskip

\noindent
\begin{minipage}{0.58\columnwidth}
\begin{definition}
We define $\mathcal N:\cat{FdHilb}\to\cat{ZW_{\!f}}$ as the functor that maps any $n$-ary state $\ket{\psi} \in \cat{FdHilb}[\obfhilb{}, \obfhilb{a_1, ..., a_n}]$:\\
$\hspace*{1em}\ket{\psi} = r_0\ket{0...0}+r_1\ket{0...01}+...+r_i\ket{x_1^i...x_n^i}+...$\\
to the diagram on the right.\\
We say of any diagram in the image of $\mathcal N$ that it is in normal form.
\end{definition}
\end{minipage}
\begin{minipage}{0.41\columnwidth}
\phantom{.}\hfill$\mathcal N(\ket\psi)=
\scalebox{\scaling}{
	\input{./figures/NF-annot.tikz}%
}$
\end{minipage}

\bigskip
\noindent
We can once again show that $\mathcal N$ builds a diagram that represents $\ket\psi$:
\begin{lemma}
$\forall \ket\psi\in\cat{FdHilb}[\obfhilb{},\obfhilb{a_1, ..., a_n}],~\interp{\mathcal N(\ket\psi)} = \ket\psi$
\end{lemma}
We hence get universality of the language as a direct consequence:
\begin{corollary}[Universality of $\cat{ZW_f}$]~

$
\forall f \in \cat{FdHilb}[\obfhilb{a_1, ..., a_n},\obfhilb{b_1, ..., b_m}],~\exists D_f\in\cat{ZW_f}[\obfhilb{a_1, ..., a_n},\obfhilb{b_1, ..., b_m}],~
\interp{D_f} = f
$
\end{corollary}

Most of the arguments given for the minimality of \zw can be adapted to arguments for the necessity of the equations of \zwf, and the few remaining equations can be given a new argument.
We hence have:
\begin{theorem}
\label{thm:minimality-ZWf}
The equational theory \zwf is minimal.
\end{theorem}
\begin{proof}
	The argument given for the necessity of Equation \nameref{ax:bialgebra-W} now works for the necessity of Equation \nameref{ax:W-assoc-2-qf}, and that of Equation \nameref{ax:one} now works for Equation \nameref{ax:bialgebra-W-qf}. Namely, Equation \nameref{ax:bialgebra-W-qf} is the only that can create a non-empty diagram from an empty diagram. A less artificial argument for Equation \nameref{ax:bialgebra-W-qf} is that it is the only equation that can create a capacity $>k$ from a diagram whose capacities are all $\leq k$ for $k\geq1$. 
	Moreover:\\
	Equation \nameref{ax:ket-1} is the only equation that can create a $\ket1$ with capacity $a\neq1$, from a diagram whose $\ket1$s are all on capacity $1$.\\
	We can reformulate the argument of Equation \nameref{ax:W-assoc-qf} as: It is the only equation permitting to create non-trivial W-nodes with arity $\geq3$ from a diagram where all non-trivial W-nodes have arity $\leq2$.\\
	In the argument of Equation \nameref{ax:loop-removal-qf}, one can take for $\varpi$ any complex number such that $\varpi^2\neq1$, and by working up to a scalar factor, as is done initially for $d=2$.\\
	In the argument of Equation \nameref{ax:hopf-qf}, we can instantiate the protocol by annotating the wires with their capacities, then continue with the protocol as explained in the initial argument.\\
	All the other arguments work right off the bat for their mixed-dimensional counterpart, hence the result of minimality.
\end{proof}

Using the normal form, we can then leverage the completeness from the qudit ZW-calculus to get the similar result in the current setting:

\begin{theorem}[Completeness for Finite Dimensional Systems]
\label{thm:completeness-finite}
The language is complete: for any two diagrams $D_1$ and $D_2$ of the $\cat{ZW_{\!f}}$-calculus:
\[\interp{D_1}=\interp{D_2}\iff \zwf\vdash D_1=D_2\]
\end{theorem}

\begin{proof}
The right-to-left implication (soundness) is again a straightforward verification. The other is proven in its entirety in \Cref{sec:appendix-fdhilb}, and uses the completeness of $\zw$ for qudit systems. The idea is to show that i) we can turn both $D_1$ and $D_2$ into a diagram with a high enough capacity $d$ everywhere (except boundaries), and that ii) all the equations of \zw can be proven in \zwf (through $\iota_d$).
\end{proof}

\section{Related Work}

\subsection{Qudit Framework}

The first and only result to date of a complete equational theory for a graphical language describing qudit systems comes from the ``ZXW-calculus'' \cite{Poor2023completeness}. There, the authors start from a qudit version of the ZX-calculus and most probably end up requiring a W-node in the definition of a normal form, and hence in the equational theory leading to completeness. We argue here that we can get a complete equational theory purely inside the ZW-calculus. By keeping the number of generators as low as possible, we also end up with few, intuitive equations in the equational theory.

The W-node we used is a different generalisation of the qubit W-node than the one used in \cite{Poor2023completeness}. The version we used offers two advantages with respect to the aims of the paper. First, it allows to use a single parameter in the Z-spiders (which aligns with the spirit of keeping things as minimal as possible) and to sum such parameters together, while the other version requires to have $(d-1)$-sized lists of coefficients as parameters in order to get a \nameref{ax:sum}-like rule to sum coefficients together. Second, it allows us to define $\ket k$ (up to a scalar) as a composition using only $\ket 1$ and the W-node. Again, this lowers the number of generators, as all the $\ket k$ (for $k>1$) become syntactic sugar.

Focussing on ZW-calculus is not a new idea. The first ever completeness proof for qubit graphical languages was in the (qubit) ZW-calculus, introduced in \cite{Coecke2010compositional} and tweaked and made complete in \cite{Hadzihasanovic2015diagrammatic,Hadzihasanovic2018complete}. The ZW-calculus noticeably has very nice combinatorial properties different from those of its counterparts, which in particular allows for a very natural notion of normal form. It is hence not suprising that some attempts were made to get a complete equational theory of qudit systems purely in ZW. There have then been tentative generalisations for qudit systems, in particular in \cite{Hadzihasanovic2017algebra} where $q$-arithmetic is used, and in \cite{Wang2021nonanyonic} where the W-node is generalised in a different way (and that we encounter in \cite{Poor2023completeness}). It is to be noted that our two main generators are essentially the same as in \cite{Hadzihasanovic2017algebra}, except with usual arithmetic instead of $q$-arithmetic. While some equations are sound with respect to the $q$-arithmetic semantics, others are truly specific to the standard arithmetic. Adapting the results of the present paper to $q$-arithmetic semantics hence seems non-trivial. Other presentations for qudit systems have also been proposed (without proof of completeness) in \cite{Roy2023qudit,DeBeaudrap2023simple}. Finally, complete presentations for fragments of qudit quantum mechanics can be found e.g.~in \cite{Poor2023qupit,Booth2024graphical}.

Another system we are close to is $\cat{QPath}$ \cite{DeFelice2022quantum}. Our W-node is merely the ``triangle'' node of $\cat{QPath}$ that we truncated to a finite dimension\footnote{The idea of truncating this tensor has also been considered in \cite{deFelice2023lightmatter} during a translation between graphical languages.}, and we generalised their ``line weight'' to an $n$-ary Z-spider. The degree-2 Z-spider furthermore has exactly the same interpretation as the line weight. While in $\cat{QPath}$ the triangle nodes satisfy a bialgebra, this is not the case when truncating to finite dimension. Here we could either resort to define a ``fermionic swap'' that would replace the usual swap in the bialgebra (as in \cite{Hadzihasanovic2017algebra} and \cite{Wang2021nonanyonic}), or give a context in which the bialgebra works (as is done in \cite{Poor2023completeness}). While such a ``fermionic swap'' exists in our setting, it does not have all the nice properties of the qubit fermionic swap, that in particular allow us to see it as a quasi-proper swap. Instead we went with the latter solution, which as it turns out works in our setting, despite the W-node having a different interpretation from that of \cite{Poor2023completeness}, and we end up with Equation \nameref{ax:bialgebra-W}.

\subsection{Finite Dimensional Framework}

Another complete presentation of a graphical language for $\cat{FdHilb}$ was announced recently before the first version of the current paper \cite{Wang2023completeness}. This one builds upon the aforementioned ZXW-calculus, and introduces a new generator that takes two systems, of dimensions $a$ and $b$, and builds a system of dimension $a\times b$. Our approach builds upon $\cat{ZW\!}_d$, the qudit version of the ZW-calculus from \Cref{sec:qudit} and hence starts with fewer generators and equations. As a consequence, the graphical language for $\cat{FdHilb}$ we end up with has fewer equations as well. 
Moreover, we did not require a new generator, and simply promoted the qudit W-node to work with any mix of dimensions in a natural manner, which was enough to provide us with universality.

A version of the ZX-calculus for $\cat{FdHilb}$ was recently provided and shown to be complete \cite{Poor2024zxcalculus}. The proof of completeness for their graphical language was obtained by transporting the property from the $\cat{ZW_{\!f}}$-calculus of the first version of the current paper to the ZX-calculus, through a system of translations between the two languages.

\section{Conclusion}

In this paper, we explored the potential for a minimal yet complete diagrammatic language for quantum mechanics beyond qubit systems. This starts with a well-chosen generalisation of the generators of the ZW-calculus, allowing us to have few and intuitive equations. For both qudit systems and finite dimensional systems, we showed that the diagrams are universal, and that the equational theories are both minimal and complete for their respective interpretation. 


\appendix

\section{Asymmetric Presentation for Qudit Systems}
\label{sec:asymmetric-presentation}

In this section, we give an alternative semantics for the \zw-diagrams, which breaks the up/down symmetry of the generators. On the one hand, the dagger-functor becomes less natural; on the other hand, the combinatorics associated to the diagrams becomes simpler (except for the cap):\\
\renewcommand{\interp}[1]{\left\llbracket #1 \right\rrbracket_{\rotatebox[origin=c]{-90}{$\rightsquigarrow$}}}
\begin{minipage}{0.42\columnwidth}
\begin{align*}
\interp{D_2\circ D_1}
&= \interp{D_2}\circ\interp{D_1}\\
\interp{D_1\otimes D_2}
&= \interp{D_1}\otimes\interp{D_2}\\
\interp{~
\scalebox{\scaling}{
	\begin{tikzpicture}
	\begin{pgfonlayer}{nodelayer}
		\node [style=none] (0) at (0, 0.25) {};
		\node [style=none] (1) at (0, -0.25) {};
	\end{pgfonlayer}
	\begin{pgfonlayer}{edgelayer}
		\draw (0.center) to (1.center);
	\end{pgfonlayer}
\end{tikzpicture}
}~}
&=\sum_{k} \ketbra k\\
\interp{
\scalebox{\scaling}{
	\input{./figures/swap.tikz}%
}}
&=\sum_{k,\ell} \ketbra{\ell,k}{k,\ell}\\
\interp{
\scalebox{\scaling}{
	\begin{tikzpicture}
	\begin{pgfonlayer}{nodelayer}
		\node [style=none] (0) at (0.5, 0.125) {};
		\node [style=none] (1) at (1, 0.125) {};
		\node [style=none] (2) at (0.75, -0.125) {};
	\end{pgfonlayer}
	\begin{pgfonlayer}{edgelayer}
		\draw [bend right=90, looseness=1.75] (0.center) to (1.center);
	\end{pgfonlayer}
\end{tikzpicture}
}}
&=\sum_{k} k!\bra{k,k}\\
\interp{
\scalebox{\scaling}{
	\begin{tikzpicture}
	\begin{pgfonlayer}{nodelayer}
		\node [style=none] (0) at (-0.5, -0.125) {};
		\node [style=none] (1) at (0, -0.125) {};
		\node [style=none] (2) at (-0.25, 0.125) {};
	\end{pgfonlayer}
	\begin{pgfonlayer}{edgelayer}
		\draw [bend left=90, looseness=1.75] (0.center) to (1.center);
	\end{pgfonlayer}
\end{tikzpicture}
}}
&=\sum_{k} \frac1{k!}\ket{k,k}
\end{align*}
\end{minipage}
\hspace*{-1em}
\begin{minipage}{0.59\columnwidth}
\begin{align*}
\interp{
\scalebox{\scaling}{
	\input{./figures/Z-spider.tikz}%
}}
&= \sum_{k=0}^{d-1} r^k k!^{n-1}\ketbra{k^m}{k^n}\\
\interp{
\scalebox{\scaling}{
	\input{./figures/W-n.tikz}%
}}
&= \!\sum_{\substack{k\in\{0,...,d-1\}\\i_1+...+i_n=k}}\! \binom{k}{i_1,...,i_n}\ketbra{i_1,...,i_n}{k}\\
\def\fig{W-n-1}\interp{\scalebox{\scaling}{
}}
&= \!\sum_{\substack{k\in\{0,...,d-1\}\\i_1+...+i_n=k}}\! \ketbra{k}{i_1,...,i_n}\\
\interp{
\scalebox{\scaling}{
	\begin{tikzpicture}
	\begin{pgfonlayer}{nodelayer}
		\node [style=ket] (0) at (0, 0.125) {$1$};
		\node [style=none] (1) at (0, -0.25) {};
	\end{pgfonlayer}
	\begin{pgfonlayer}{edgelayer}
		\draw (1.center) to (0);
	\end{pgfonlayer}
\end{tikzpicture}
}}
&= \ket 1 \qquad \text{and}\qquad \def\fig{bra-1}\interp{\scalebox{\scaling}{
}}=\bra1\\
\interp{r}&=r
\end{align*}
\end{minipage}\\ 
As a direct consequence, we have:
\begin{align*}
\interp{
\scalebox{\scaling}{
	\begin{tikzpicture}
	\begin{pgfonlayer}{nodelayer}
		\node [style=ket] (0) at (-0.25, 0.25) {$k$};
		\node [style=none] (1) at (-0.25, -0.25) {};
	\end{pgfonlayer}
	\begin{pgfonlayer}{edgelayer}
		\draw (1.center) to (0);
	\end{pgfonlayer}
\end{tikzpicture}
}} = \ket k \qquad\text{and}\qquad \def\fig{bra-k}\interp{\scalebox{\scaling}{
}} = k!\bra k
\end{align*}
This semantics being equivalent to $\interp{-}$, the equational theory of \zw-diagrams remains universal, sound and complete in this setting. 

\renewcommand{\interp}[1]{\left\llbracket #1 \right\rrbracket}

\section{Lemmas and Proofs for Completeness of Qudit}
\label{sec:proofs-qudit}
In this section, we give the proofs necessary for the completeness of the qudit setting. To get there, we also provide a set of useful derivable equations:

\begin{multicols}{2}

\begin{lemma}
\label{lem:scalar-1}
\def\fig{empty-diag-is-one}
\begin{align*}
\zw\vdash~1
\eq{}\scalebox{\scaling}{
}
\end{align*}
\end{lemma}

\begin{lemma}
\label{lem:Z-id}
\def\fig{Z-id-prf}
\begin{align*}
\zw\vdash~\scalebox{\scaling}{
}
\eq{}\scalebox{\scaling}{
}
\end{align*}
\end{lemma}

\begin{lemma}
\label{lem:W-bialgebra-gen}
\def\fig{W-bialgebra-old}
\begin{align*}
\zw\vdash~\scalebox{\scaling}{
}
\eq{}\scalebox{\scaling}{
}
\end{align*}
\end{lemma}

\begin{lemma}
\label{lem:vacuum}
\def\fig{W-counit}
\begin{align*}
\zw\vdash~\scalebox{\scaling}{
}
\eq{}\scalebox{\scaling}{
}
\end{align*}
\end{lemma}

\begin{lemma}
\label{lem:braket-1-1}
\def\fig{braket-1-1-is-empty}
\begin{align*}
\zw\vdash~\scalebox{\scaling}{
}
\eq{}\scalebox{\scaling}{
}
\end{align*}
\end{lemma}

\begin{lemma}
\label{lem:ket0}
\def\fig{ket-0}
\begin{align*}
\zw\vdash~\scalebox{\scaling}{
}
\eq{}\scalebox{\scaling}{
}
\end{align*}
\end{lemma}

\begin{lemma}
\label{lem:copy}
\def\fig{copy-k}
\begin{align*}
\zw\vdash~\scalebox{\scaling}{
}
\eq{}r^k\cdot\scalebox{\scaling}{
}
\end{align*}
\end{lemma}

\begin{lemma}
\label{lem:sum-gen}
\def\fig{sum-gen}
\begin{align*}
\zw\vdash~\scalebox{\scaling}{
}
\eq{}\scalebox{\scaling}{
}
\end{align*}
\end{lemma}

\begin{lemma}
\label{lem:skewed-hopf}
\def\fig{Z-W-loop-prf}
\begin{align*}
\zw\vdash~\scalebox{\scaling}{
}
\eq{}\scalebox{\scaling}{
}
\end{align*}
\end{lemma}

\begin{lemma}
\label{lem:ket-sum}
\def\fig{ket-sum}
\begin{align*}
\zw\vdash~\scalebox{\scaling}{
	\input{./figures/\fig/\fig_12.tikz}%
}
\eq[]{}
\begin{cases}
\scalebox{\scaling}{
	\input{./figures/\fig/\fig_14.tikz}%
} & \text{ if }k+\ell<d\\[1em]
0\cdot\scalebox{\scaling}{
	\input{./figures/\fig/\fig_11.tikz}%
} & \text{ if }k+\ell\geq d
\end{cases}
\end{align*}
\end{lemma}

\begin{lemma}
\label{lem:qubit-hopf}
\def\fig{qubit-Hopf}
\begin{align*}
\zw\vdash~\scalebox{\scaling}{
}
\eq{}\scalebox{\scaling}{
	\input{./figures/\fig/\fig_05.tikz}%
}
\end{align*}
\end{lemma}

\begin{lemma}
\label{lem:cup-distrib-aux}
\def\fig{NF-cup-distrib-lemma}
\begin{align*}
\zw\vdash~\scalebox{\scaling}{
}
\eq{}\scalebox{\scaling}{
}
\end{align*}
\end{lemma}

\begin{lemma}
\label{lem:Pascal}
\def\fig{lemma-Pascal-multinomial}
\begin{align*}
\zw\vdash~\scalebox{\scaling}{
}
\eq{}\scalebox{\scaling}{
}
\end{align*}
where on the right-hand-side, there are as many Z-spiders as there are ways to decompose $n$ as a sum of $m$ natural numbers: $n = i_1+...+i_m$.
\end{lemma}

\begin{lemma}
\label{lem:Not-through-Z}
\def\fig{Not-through-Z-prf}
\begin{align*}
\zw\vdash~\scalebox{\scaling}{
}
\eq{}\scalebox{\scaling}{
	\input{./figures/\fig/\fig_10.tikz}%
}
\end{align*}
\end{lemma}

\begin{lemma}
\label{lem:partition}
With $0<k<d$:
\def\fig{partition-prf}
\begin{align*}
\zw\vdash~\scalebox{\scaling}{
}
\eq{}\scalebox{\scaling}{
	\input{./figures/\fig/\fig_05.tikz}%
}
\end{align*}
\end{lemma}

\begin{lemma}
\label{lem:ket-1}
If $0<k<d$:
\def\fig{ket-1}
\begin{align*}
\zw\vdash~ k!~\scalebox{\scaling}{
}
\eq{}\scalebox{\scaling}{
}
\end{align*}
\end{lemma}

\begin{lemma}
\label{lem:dot-product}
If $k\neq \ell$:
\def\fig{braket-k}
\begin{align*}
\zw\vdash~\scalebox{\scaling}{
}
\eq[]{}k!
\quad\text{and}\quad
\def\fig{braket-k-l}
\zw\vdash~\scalebox{\scaling}{
}
\eq[]{}0
\end{align*}
\end{lemma}

\begin{lemma}
\label{lem:not-through-phase}
\def\fig{Not-through-phase}
\begin{align*}
\zw\vdash~\scalebox{\scaling}{
}
\eq{}r\cdot\scalebox{\scaling}{
}
\end{align*}
\end{lemma}

\begin{lemma}
\label{lem:mult-to-Z-effect}
If $0\leq k<d$:
\def\fig{multiplier-to-Z-effect}
\begin{align*}
\zw\vdash~\scalebox{\scaling}{
}
\eq{}\scalebox{\scaling}{
	\input{./figures/\fig/\fig_06.tikz}%
}
\end{align*}
\end{lemma}

\begin{lemma}
\label{lem:Z-loop-on-not}
If $0\leq k<d$, we have:
\def\fig{Z-loop-on-P}
\begin{align*}
\zw\vdash~\scalebox{\scaling}{
}
\eq{}\scalebox{\scaling}{
	\input{./figures/\fig/\fig_10.tikz}%
}
\end{align*}
\end{lemma}

\begin{lemma}
\label{lem:W-bialgebra-qubit-ctxt}
The bialgebra between W-nodes can be used in the following context, if $0\leq k_i<d$ for all $i$:
\def\fig{W-bialgebra-context}
\begin{align*}
\zw\vdash~\scalebox{\scaling}{
}
\eq[]{}~\scalebox{\scaling}{
}
\end{align*}
\end{lemma}

\end{multicols}

As we have quite a number of them, we provide in \Cref{fig:lem-deps} a small (automatically generated) graph of all dependencies between lemmas, propositions and theorems, so as to convince the reader that there is no circular proof. Equations from the equational theories are not displayed in the graph, they are assumed to be available at any point. The graph also includes the dependencies for the proofs of \Cref{sec:fdhilb}.

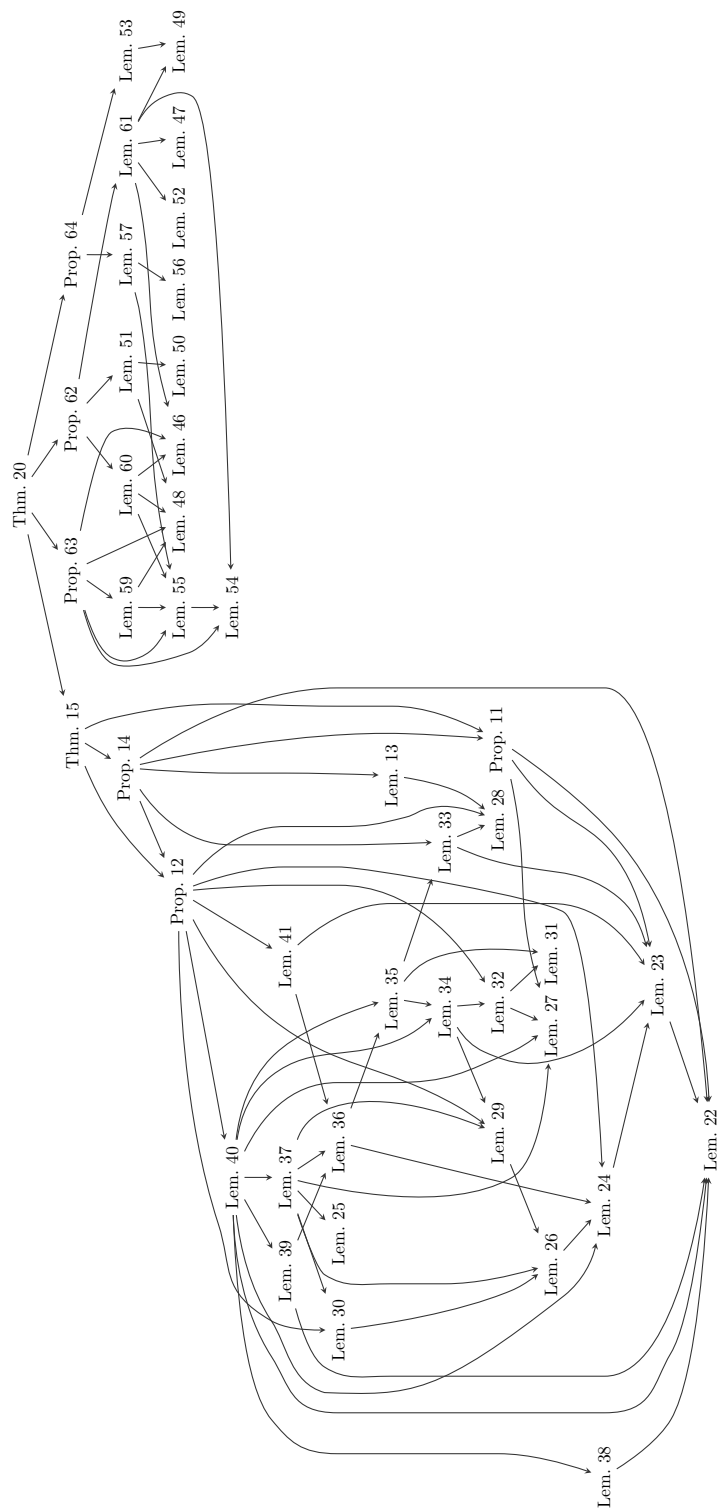
\begin{figure}[!htb]
\centering
\scalebox{0.67}{\rotatebox{90}{
\begin{tikzpicture}[>=stealth,line join=bevel,yscale=0.4,xscale=0.27]
\node (v0) at (1645.0bp,907.0bp) [draw,draw=none] {Thm.~\ref{thm:completeness-qudit}};
\node (v1) at (1578.0bp,833.0bp) [draw,draw=none] {Prop.~\ref{prop:NF-generators}};
\node (v2) at (1637.0bp,315.0bp) [draw,draw=none] {Prop.~\ref{prop:NF-tensor}};
\node (v3) at (1321.0bp,759.0bp) [draw,draw=none] {Prop.~\ref{prop:NF-cup}};
\node (v5) at (805.0bp,19.0bp) [draw,draw=none] {Lem.~\ref{lem:Z-id}};
\node (v13) at (1424.0bp,389.0bp) [draw,draw=none] {Lem.~\ref{lem:Pascal}};
\node (v24) at (1562.0bp,463.0bp) [draw,draw=none] {Lem.~\ref{lem:W21-to-NF}};
\node (v4) at (1127.0bp,93.0bp) [draw,draw=none] {Lem.~\ref{lem:W-bialgebra-gen}};
\node (v12) at (1042.0bp,241.0bp) [draw,draw=none] {Lem.~\ref{lem:copy}};
\node (v6) at (669.0bp,167.0bp) [draw,draw=none] {Lem.~\ref{lem:vacuum}};
\node (v8) at (819.0bp,315.0bp) [draw,draw=none] {Lem.~\ref{lem:skewed-hopf}};
\node (v9) at (410.0bp,537.0bp) [draw,draw=none] {Lem.~\ref{lem:ket-sum}};
\node (v10) at (1090.0bp,315.0bp) [draw,draw=none] {Lem.~\ref{lem:cup-distrib-aux}};
\node (v14) at (1470.0bp,315.0bp) [draw,draw=none] {Lem.~\ref{lem:sum-gen}};
\node (v22) at (727.0bp,685.0bp) [draw,draw=none] {Lem.~\ref{lem:Z-loop-on-not}};
\node (v23) at (1186.0bp,611.0bp) [draw,draw=none] {Lem.~\ref{lem:W-bialgebra-qubit-ctxt}};
\node (v7) at (546.0bp,241.0bp) [draw,draw=none] {Lem.~\ref{lem:ket0}};
\node (v11) at (1192.0bp,241.0bp) [draw,draw=none] {Lem.~\ref{lem:qubit-hopf}};
\node (v15) at (1081.0bp,389.0bp) [draw,draw=none] {Lem.~\ref{lem:Not-through-Z}};
\node (v16) at (1099.0bp,463.0bp) [draw,draw=none] {Lem.~\ref{lem:partition}};
\node (v17) at (800.0bp,537.0bp) [draw,draw=none] {Lem.~\ref{lem:ket-1}};
\node (v18) at (727.0bp,611.0bp) [draw,draw=none] {Lem.~\ref{lem:dot-product}};
\node (v19) at (610.0bp,537.0bp) [draw,draw=none] {Lem.~\ref{lem:braket-1-1}};
\node (v20) at (105.0bp,167.0bp) [draw,draw=none] {Lem.~\ref{lem:not-through-phase}};
\node (v21) at (533.0bp,611.0bp) [draw,draw=none] {Lem.~\ref{lem:mult-to-Z-effect}};
\node (v25) at (2422.0bp,833.0bp) [draw,draw=none] {Lem.~\ref{lem:W-bialgebra-ctxt}};
\node (v26) at (2413.0bp,759.0bp) [draw,draw=none] {Lem.~\ref{lem:W-assoc-gen}};
\node (v27) at (2090.0bp,759.0bp) [draw,draw=none] {Lem.~\ref{lem:embedding}};
\node (v28) at (3068.0bp,833.0bp) [draw,draw=none] {Lem.~\ref{lem:copy-0-qf}};
\node (v29) at (3087.0bp,759.0bp) [draw,draw=none] {Lem.~\ref{lem:ket-0-qf}};
\node (v30) at (1911.0bp,759.0bp) [draw,draw=none] {Lem.~\ref{lem:ket-k-dim-qf}};
\node (v31) at (1911.0bp,685.0bp) [draw,draw=none] {Lem.~\ref{lem:ket-k-forms}};
\node (v32) at (2644.0bp,833.0bp) [draw,draw=none] {Lem.~\ref{lem:NF-qf-other-dim-qf}};
\node (v33) at (2572.0bp,759.0bp) [draw,draw=none] {Lem.~\ref{lem:zcp-1}};
\node (v34) at (1911.0bp,833.0bp) [draw,draw=none] {Lem.~\ref{lem:Z-restrict-NF-qf}};
\node (v35) at (2165.0bp,833.0bp) [draw,draw=none] {Lem.~\ref{lem:Z-loop-removal-qf}};
\node (v36) at (2249.0bp,759.0bp) [draw,draw=none] {Lem.~\ref{lem:Z-id-qf}};
\node (v37) at (2871.0bp,833.0bp) [draw,draw=none] {Lem.~\ref{lem:discard-ket-qf}};
\node (v38) at (2887.0bp,759.0bp) [draw,draw=none] {Lem.~\ref{lem:scalar-1-qf}};
\node (v39) at (2720.0bp,759.0bp) [draw,draw=none] {Lem.~\ref{lem:Pascal-qf}};
\node (v40) at (2303.0bp,907.0bp) [draw,draw=none] {Prop.~\ref{prop:qufinite-derives-qudit}};
\node (v42) at (1988.0bp,907.0bp) [draw,draw=none] {Prop.~\ref{prop:qufinite-to-qudit}};
\node (v43) at (2644.0bp,907.0bp) [draw,draw=none] {Prop.~\ref{prop:qudit-NF-to-qufinite-NF}};
\node (v44) at (2145.0bp,981.0bp) [draw,draw=none] {Thm.~\ref{thm:completeness-finite}};
\draw [->] (v0) ..controls (1620.2bp,879.37bp) and (1610.7bp,869.15bp)  .. (v1);
\draw [->] (v0) ..controls (1668.1bp,877.86bp) and (1676.5bp,864.99bp)  .. (1681.0bp,852.0bp) .. controls (1726.9bp,719.73bp) and (1706.0bp,678.0bp)  .. (1706.0bp,538.0bp) .. controls (1706.0bp,538.0bp) and (1706.0bp,538.0bp)  .. (1706.0bp,462.0bp) .. controls (1706.0bp,416.52bp) and (1678.3bp,370.06bp)  .. (v2);
\draw [->] (v0) ..controls (1543.5bp,878.79bp) and (1506.2bp,866.53bp)  .. (1474.0bp,852.0bp) .. controls (1431.2bp,832.69bp) and (1385.4bp,804.01bp)  .. (v3);
\draw [->] (v1) ..controls (1603.1bp,758.33bp) and (1655.8bp,588.33bp)  .. (1649.0bp,444.0bp) .. controls (1647.4bp,409.67bp) and (1643.4bp,370.32bp)  .. (v2);
\draw [->] (v1) ..controls (1477.6bp,803.87bp) and (1432.7bp,791.31bp)  .. (v3);
\draw [->] (v1) ..controls (1641.2bp,770.43bp) and (1744.0bp,656.11bp)  .. (1744.0bp,538.0bp) .. controls (1744.0bp,538.0bp) and (1744.0bp,538.0bp)  .. (1744.0bp,166.0bp) .. controls (1744.0bp,77.18bp) and (1079.5bp,34.478bp)  .. (v5);
\draw [->] (v1) ..controls (1485.7bp,791.54bp) and (1418.0bp,748.61bp)  .. (1418.0bp,686.0bp) .. controls (1418.0bp,686.0bp) and (1418.0bp,686.0bp)  .. (1418.0bp,536.0bp) .. controls (1418.0bp,495.1bp) and (1420.4bp,447.82bp)  .. (v13);
\draw [->] (v1) ..controls (1572.5bp,785.94bp) and (1567.0bp,732.01bp)  .. (1567.0bp,686.0bp) .. controls (1567.0bp,686.0bp) and (1567.0bp,686.0bp)  .. (1567.0bp,610.0bp) .. controls (1567.0bp,569.29bp) and (1565.0bp,522.24bp)  .. (v24);
\draw [->] (v2) ..controls (1575.5bp,286.26bp) and (1550.3bp,273.5bp)  .. (1529.0bp,260.0bp) .. controls (1460.8bp,216.84bp) and (1459.2bp,182.02bp)  .. (1386.0bp,148.0bp) .. controls (1338.8bp,126.09bp) and (1282.8bp,113.0bp)  .. (v4);
\draw [->] (v2) ..controls (1583.7bp,262.57bp) and (1505.7bp,191.64bp)  .. (1429.0bp,148.0bp) .. controls (1245.0bp,43.287bp) and (987.15bp,23.607bp)  .. (v5);
\draw [->] (v2) ..controls (1555.2bp,296.92bp) and (1552.1bp,296.44bp)  .. (1549.0bp,296.0bp) .. controls (1357.1bp,268.66bp) and (1302.9bp,293.6bp)  .. (v12);
\draw [->] (v3) ..controls (1346.8bp,712.88bp) and (1372.0bp,660.33bp)  .. (1372.0bp,612.0bp) .. controls (1372.0bp,612.0bp) and (1372.0bp,612.0bp)  .. (1372.0bp,536.0bp) .. controls (1372.0bp,499.55bp) and (1305.3bp,245.0bp)  .. (1277.0bp,222.0bp) .. controls (1236.7bp,189.22bp) and (901.31bp,174.97bp)  .. (v6);
\draw [->] (v3) ..controls (1207.2bp,719.82bp) and (1119.3bp,682.95bp)  .. (1057.0bp,630.0bp) .. controls (949.33bp,538.48bp) and (965.46bp,477.72bp)  .. (874.0bp,370.0bp) .. controls (865.61bp,360.12bp) and (855.69bp,349.98bp)  .. (v8);
\draw [->] (v3) ..controls (1126.8bp,758.83bp) and (854.36bp,752.69bp)  .. (631.0bp,704.0bp) .. controls (537.58bp,683.64bp) and (491.44bp,702.41bp)  .. (429.0bp,630.0bp) .. controls (413.92bp,612.52bp) and (409.71bp,586.24bp)  .. (v9);
\draw [->] (v3) ..controls (1327.5bp,711.97bp) and (1334.0bp,658.06bp)  .. (1334.0bp,612.0bp) .. controls (1334.0bp,612.0bp) and (1334.0bp,612.0bp)  .. (1334.0bp,536.0bp) .. controls (1334.0bp,436.11bp) and (1219.8bp,369.93bp)  .. (v10);
\draw [->] (v3) ..controls (1395.8bp,716.34bp) and (1456.0bp,671.03bp)  .. (1456.0bp,612.0bp) .. controls (1456.0bp,612.0bp) and (1456.0bp,612.0bp)  .. (1456.0bp,536.0bp) .. controls (1456.0bp,494.25bp) and (1460.5bp,483.14bp)  .. (1475.0bp,444.0bp) .. controls (1481.4bp,426.84bp) and (1490.4bp,425.72bp)  .. (1495.0bp,408.0bp) .. controls (1499.2bp,391.65bp) and (1498.6bp,386.49bp)  .. (1495.0bp,370.0bp) .. controls (1493.0bp,361.01bp) and (1489.4bp,351.73bp)  .. (v14);
\draw [->] (v3) ..controls (1146.6bp,736.86bp) and (943.85bp,712.29bp)  .. (v22);
\draw [->] (v3) ..controls (1280.4bp,714.12bp) and (1237.0bp,667.11bp)  .. (v23);
\draw [->] (v4) ..controls (991.15bp,61.623bp) and (919.37bp,45.572bp)  .. (v5);
\draw [->] (v6) ..controls (810.49bp,143.76bp) and (931.59bp,124.72bp)  .. (v4);
\draw [->] (v7) ..controls (592.77bp,212.62bp) and (612.06bp,201.33bp)  .. (v6);
\draw [->] (v8) ..controls (707.77bp,284.66bp) and (653.16bp,270.26bp)  .. (v7);
\draw [->] (v9) ..controls (421.34bp,476.23bp) and (445.19bp,371.11bp)  .. (492.0bp,296.0bp) .. controls (498.51bp,285.56bp) and (507.4bp,275.59bp)  .. (v7);
\draw [->] (v10) ..controls (1128.5bp,286.79bp) and (1144.2bp,275.75bp)  .. (v11);
\draw [->] (v10) ..controls (1072.5bp,287.71bp) and (1066.0bp,277.96bp)  .. (v12);
\draw [->] (v13) ..controls (1405.2bp,359.7bp) and (1397.1bp,346.42bp)  .. (1391.0bp,334.0bp) .. controls (1351.2bp,253.11bp) and (1379.4bp,205.51bp)  .. (1310.0bp,148.0bp) .. controls (1288.5bp,130.22bp) and (1261.4bp,118.29bp)  .. (v4);
\draw [->] (v13) ..controls (1440.8bp,361.71bp) and (1447.0bp,351.96bp)  .. (v14);
\draw [->] (v15) ..controls (1015.6bp,361.53bp) and (998.69bp,349.76bp)  .. (988.0bp,334.0bp) .. controls (959.92bp,292.61bp) and (956.05bp,267.42bp)  .. (977.0bp,222.0bp) .. controls (998.61bp,175.14bp) and (1046.8bp,139.09bp)  .. (v4);
\draw [->] (v15) ..controls (978.63bp,359.87bp) and (932.92bp,347.31bp)  .. (v8);
\draw [->] (v15) ..controls (1084.2bp,362.13bp) and (1085.4bp,352.97bp)  .. (v10);
\draw [->] (v16) ..controls (1154.5bp,435.14bp) and (1170.2bp,423.21bp)  .. (1180.0bp,408.0bp) .. controls (1206.9bp,366.16bp) and (1203.1bp,305.51bp)  .. (v11);
\draw [->] (v16) ..controls (1223.5bp,434.42bp) and (1297.6bp,418.0bp)  .. (v13);
\draw [->] (v16) ..controls (1092.5bp,436.05bp) and (1090.2bp,426.77bp)  .. (v15);
\draw [->] (v17) ..controls (780.08bp,482.66bp) and (750.17bp,402.11bp)  .. (726.0bp,334.0bp) .. controls (708.9bp,285.8bp) and (689.72bp,229.45bp)  .. (v6);
\draw [->] (v17) ..controls (903.31bp,511.12bp) and (967.63bp,495.63bp)  .. (v16);
\draw [->] (v18) ..controls (593.53bp,580.58bp) and (533.74bp,565.26bp)  .. (525.0bp,556.0bp) .. controls (496.34bp,525.64bp) and (506.0bp,505.75bp)  .. (506.0bp,464.0bp) .. controls (506.0bp,464.0bp) and (506.0bp,464.0bp)  .. (506.0bp,388.0bp) .. controls (506.0bp,345.45bp) and (522.0bp,298.32bp)  .. (v7);
\draw [->] (v18) ..controls (830.58bp,588.29bp) and (852.48bp,575.76bp)  .. (866.0bp,556.0bp) .. controls (911.41bp,489.64bp) and (863.92bp,390.13bp)  .. (v8);
\draw [->] (v18) ..controls (598.6bp,580.84bp) and (536.38bp,566.7bp)  .. (v9);
\draw [->] (v18) ..controls (700.14bp,538.74bp) and (652.18bp,383.5bp)  .. (726.0bp,296.0bp) .. controls (757.31bp,258.89bp) and (892.72bp,247.28bp)  .. (v12);
\draw [->] (v18) ..controls (754.15bp,583.22bp) and (764.73bp,572.79bp)  .. (v17);
\draw [->] (v18) ..controls (682.51bp,582.62bp) and (664.16bp,571.33bp)  .. (v19);
\draw [->] (v20) ..controls (143.32bp,125.92bp) and (182.45bp,90.373bp)  .. (224.0bp,74.0bp) .. controls (316.68bp,37.475bp) and (609.87bp,25.2bp)  .. (v5);
\draw [->] (v21) ..controls (394.17bp,584.98bp) and (348.04bp,571.97bp)  .. (333.0bp,556.0bp) .. controls (304.38bp,525.6bp) and (314.0bp,505.75bp)  .. (314.0bp,464.0bp) .. controls (314.0bp,464.0bp) and (314.0bp,464.0bp)  .. (314.0bp,166.0bp) .. controls (314.0bp,77.918bp) and (607.98bp,38.732bp)  .. (v5);
\draw [->] (v21) ..controls (640.56bp,581.0bp) and (692.06bp,567.11bp)  .. (v17);
\draw [->] (v22) ..controls (562.47bp,681.04bp) and (447.06bp,669.49bp)  .. (355.0bp,630.0bp) .. controls (294.21bp,603.92bp) and (238.0bp,604.15bp)  .. (238.0bp,538.0bp) .. controls (238.0bp,538.0bp) and (238.0bp,538.0bp)  .. (238.0bp,166.0bp) .. controls (238.0bp,97.377bp) and (298.27bp,99.442bp)  .. (362.0bp,74.0bp) .. controls (429.38bp,47.104bp) and (635.23bp,30.662bp)  .. (v5);
\draw [->] (v22) ..controls (579.09bp,673.23bp) and (496.27bp,659.88bp)  .. (429.0bp,630.0bp) .. controls (379.77bp,608.13bp) and (300.77bp,604.34bp)  .. (280.0bp,556.0bp) .. controls (273.42bp,434.86bp) and (273.71bp,400.04bp)  .. (482.0bp,222.0bp) .. controls (512.78bp,196.45bp) and (555.07bp,182.99bp)  .. (v6);
\draw [->] (v22) ..controls (848.48bp,645.36bp) and (924.0bp,605.47bp)  .. (924.0bp,538.0bp) .. controls (924.0bp,538.0bp) and (924.0bp,538.0bp)  .. (924.0bp,462.0bp) .. controls (924.0bp,382.93bp) and (944.89bp,362.28bp)  .. (988.0bp,296.0bp) .. controls (994.71bp,285.68bp) and (1003.7bp,275.75bp)  .. (v12);
\draw [->] (v22) ..controls (855.33bp,673.0bp) and (903.59bp,659.29bp)  .. (938.0bp,630.0bp) .. controls (1006.7bp,571.55bp) and (965.7bp,516.72bp)  .. (1019.0bp,444.0bp) .. controls (1026.9bp,433.19bp) and (1037.4bp,423.04bp)  .. (v15);
\draw [->] (v22) ..controls (864.48bp,674.58bp) and (927.78bp,661.13bp)  .. (976.0bp,630.0bp) .. controls (1029.5bp,595.44bp) and (1068.0bp,528.92bp)  .. (v16);
\draw [->] (v22) ..controls (727.0bp,658.13bp) and (727.0bp,648.97bp)  .. (v18);
\draw [->] (v22) ..controls (515.78bp,679.72bp) and (293.31bp,668.79bp)  .. (224.0bp,630.0bp) .. controls (178.93bp,604.78bp) and (153.0bp,589.65bp)  .. (153.0bp,538.0bp) .. controls (153.0bp,538.0bp) and (153.0bp,538.0bp)  .. (153.0bp,314.0bp) .. controls (153.0bp,270.85bp) and (133.85bp,223.95bp)  .. (v20);
\draw [->] (v22) ..controls (651.64bp,656.03bp) and (618.87bp,643.87bp)  .. (v21);
\draw [->] (v23) ..controls (1244.9bp,567.1bp) and (1296.0bp,519.38bp)  .. (1296.0bp,464.0bp) .. controls (1296.0bp,464.0bp) and (1296.0bp,464.0bp)  .. (1296.0bp,240.0bp) .. controls (1296.0bp,180.08bp) and (1234.4bp,139.3bp)  .. (v4);
\draw [->] (v23) ..controls (1020.0bp,579.04bp) and (928.91bp,562.05bp)  .. (v17);
\draw [->] (v24) ..controls (1553.1bp,424.46bp) and (1543.5bp,393.24bp)  .. (1528.0bp,370.0bp) .. controls (1520.9bp,359.35bp) and (1511.3bp,349.32bp)  .. (v14);
\draw [->] (v25) ..controls (2418.8bp,806.13bp) and (2417.6bp,796.97bp)  .. (v26);
\draw [->] (v25) ..controls (2290.0bp,803.38bp) and (2228.7bp,790.09bp)  .. (v27);
\draw [->] (v28) ..controls (3074.8bp,806.05bp) and (3077.3bp,796.77bp)  .. (v29);
\draw [->] (v30) ..controls (1911.0bp,732.13bp) and (1911.0bp,722.97bp)  .. (v31);
\draw [->] (v32) ..controls (2535.6bp,814.65bp) and (2532.8bp,814.31bp)  .. (2530.0bp,814.0bp) .. controls (2299.6bp,788.01bp) and (2236.9bp,808.29bp)  .. (v30);
\draw [->] (v32) ..controls (2617.3bp,805.28bp) and (2606.9bp,794.94bp)  .. (v33);
\draw [->] (v34) ..controls (1980.3bp,804.12bp) and (2010.2bp,792.08bp)  .. (v27);
\draw [->] (v34) ..controls (1911.0bp,806.13bp) and (1911.0bp,796.97bp)  .. (v30);
\draw [->] (v35) ..controls (2137.0bp,805.13bp) and (2126.0bp,794.58bp)  .. (v27);
\draw [->] (v35) ..controls (2065.8bp,803.87bp) and (2021.4bp,791.31bp)  .. (v30);
\draw [->] (v35) ..controls (2196.4bp,805.05bp) and (2208.9bp,794.37bp)  .. (v36);
\draw [->] (v37) ..controls (2954.9bp,804.04bp) and (2991.8bp,791.73bp)  .. (v29);
\draw [->] (v37) ..controls (2947.0bp,805.7bp) and (2963.3bp,794.14bp)  .. (2973.0bp,778.0bp) .. controls (2981.7bp,763.53bp) and (2984.6bp,752.24bp)  .. (2973.0bp,740.0bp) .. controls (2939.8bp,705.02bp) and (2252.3bp,691.16bp)  .. (v31);
\draw [->] (v37) ..controls (2774.5bp,815.99bp) and (2766.6bp,814.92bp)  .. (2759.0bp,814.0bp) .. controls (2568.3bp,790.88bp) and (2515.5bp,809.16bp)  .. (v36);
\draw [->] (v37) ..controls (2876.7bp,806.13bp) and (2878.8bp,796.97bp)  .. (v38);
\draw [->] (v37) ..controls (2812.9bp,804.29bp) and (2788.2bp,792.51bp)  .. (v39);
\draw [->] (v40) ..controls (2348.2bp,878.62bp) and (2366.9bp,867.33bp)  .. (v25);
\draw [->] (v40) ..controls (2250.2bp,878.45bp) and (2228.1bp,866.91bp)  .. (v35);
\draw [->] (v40) ..controls (2514.8bp,883.02bp) and (2645.1bp,867.89bp)  .. (2759.0bp,852.0bp) .. controls (2763.4bp,851.39bp) and (2767.9bp,850.74bp)  .. (v37);
\draw [->] (v42) ..controls (2018.4bp,862.45bp) and (2050.7bp,816.31bp)  .. (v27);
\draw [->] (v42) ..controls (1854.0bp,878.9bp) and (1818.0bp,866.72bp)  .. (1807.0bp,852.0bp) .. controls (1796.9bp,838.46bp) and (1798.2bp,828.43bp)  .. (1807.0bp,814.0bp) .. controls (1815.1bp,800.66bp) and (1827.7bp,790.47bp)  .. (v30);
\draw [->] (v42) ..controls (1846.9bp,880.35bp) and (1807.0bp,867.87bp)  .. (1795.0bp,852.0bp) .. controls (1774.3bp,824.7bp) and (1816.2bp,744.73bp)  .. (1820.0bp,740.0bp) .. controls (1830.1bp,727.47bp) and (1844.0bp,717.18bp)  .. (v31);
\draw [->] (v42) ..controls (1959.3bp,879.13bp) and (1948.0bp,868.58bp)  .. (v34);
\draw [->] (v42) ..controls (2166.7bp,882.75bp) and (2263.8bp,866.7bp)  .. (2276.0bp,852.0bp) .. controls (2291.3bp,833.56bp) and (2280.9bp,806.77bp)  .. (v36);
\draw [->] (v43) ..controls (2817.8bp,877.68bp) and (2897.9bp,864.38bp)  .. (2969.0bp,852.0bp) .. controls (2974.4bp,851.05bp) and (2980.0bp,850.07bp)  .. (v28);
\draw [->] (v43) ..controls (2644.0bp,880.13bp) and (2644.0bp,870.97bp)  .. (v32);
\draw [->] (v44) ..controls (1955.3bp,952.69bp) and (1847.8bp,937.2bp)  .. (v0);
\draw [->] (v44) ..controls (2205.9bp,952.24bp) and (2231.9bp,940.4bp)  .. (v40);
\draw [->] (v44) ..controls (2084.5bp,952.24bp) and (2058.7bp,940.4bp)  .. (v42);
\draw [->] (v44) ..controls (2330.0bp,953.31bp) and (2429.5bp,938.95bp)  .. (v43);
\end{tikzpicture}
}}

\caption{Dependencies between lemmas, propositions and theorems. Arrows go from the consequences to their dependencies.}
\label{fig:lem-deps}
\end{figure}


\begin{proof}[Proof of \Cref{lem:scalar-1}]
\def\fig{empty-diag-is-one}
\begin{align*}
\zw\vdash~~\scalebox{\scaling}{
}
&\eq{\nameref{ax:one}}\scalebox{\scaling}{
}
\eq{\nameref{ax:copy}}1
\qedhere
\end{align*}
\end{proof}

\begin{proof}[Proof of \Cref{lem:Z-id}]
\def\fig{Z-id-prf}
\begin{align*}
\zw\vdash~\scalebox{\scaling}{
}
&\eq{\nameref{ax:Z-spider}}\scalebox{\scaling}{
}
\eq{\nameref{ax:W-unit}}\scalebox{\scaling}{
}
\eq{\nameref{ax:bialgebra-W}}\scalebox{\scaling}{
}
\eq{\nameref{ax:W-unit}}\scalebox{\scaling}{
}
\qedhere
\end{align*}
\end{proof}

\begin{proof}[Proof of \Cref{lem:W-bialgebra-gen}]
\def\fig{W-bialgebra-old}
\begin{align*}
\zw\vdash~\scalebox{\scaling}{
}
&\eq{\nameref{ax:bialgebra-W}\\\nameref{ax:Z-spider}}\scalebox{\scaling}{
}
\eq{\nameref{ax:Z-spider}\\\nameref{ax:bialgebra-Z-W}}\scalebox{\scaling}{
}
\eq{\nameref{ax:Z-spider}\\\ref{lem:Z-id}}\scalebox{\scaling}{
}
\qedhere
\end{align*}
\end{proof}

\begin{proof}[Proof of \Cref{lem:vacuum}]
\def\fig{W-counit}
\begin{align*}
\zw\vdash~\scalebox{\scaling}{
}
&\eq{\nameref{ax:one}\\\nameref{ax:W-unit}\\\nameref{ax:W-assoc}}\scalebox{\scaling}{
}
\eq{\ref{lem:W-bialgebra-gen}}\scalebox{\scaling}{
}
\eq{\nameref{ax:one}\\\nameref{ax:W-unit}}\scalebox{\scaling}{
}
\qedhere
\end{align*}
\end{proof}

Notice that in the case of $0$ inputs, we get $
\scalebox{\scaling}{
	\begin{tikzpicture}
	\begin{pgfonlayer}{nodelayer}
		\node [style=ket] (0) at (0, 0.25) {};
		\node [style=ket] (1) at (0, -0.25) {};
	\end{pgfonlayer}
	\begin{pgfonlayer}{edgelayer}
		\draw (0) to (1);
	\end{pgfonlayer}
\end{tikzpicture}
}\eq{}
\scalebox{\scaling}{
	\input{./figures/empty-diag.tikz}%
}$.

\begin{proof}[Proof of \Cref{lem:braket-1-1}]
\def\fig{braket-1-aux}
\begin{align*}
\zw\vdash~\scalebox{\scaling}{
}
&\eq{\nameref{ax:copy}}\scalebox{\scaling}{
}
\eq{\nameref{ax:Z-spider}\\\nameref{ax:W-unit}\\\nameref{ax:W-assoc}}\scalebox{\scaling}{
}
\eq{\nameref{ax:loop}}\scalebox{\scaling}{
}
\eq{\nameref{ax:W-assoc}\\\nameref{ax:W-unit}}\scalebox{\scaling}{
}
\eq{\nameref{ax:copy}}
\scalebox{\scaling}{
	\input{./figures/empty-diag.tikz}%
}
\qedhere
\end{align*}
\end{proof}

\begin{proof}[Proof of \Cref{lem:ket0}]
\def\fig{ket-0-prf}
\begin{align*}
\zw\vdash~\scalebox{\scaling}{
}
&\eq{\ref{lem:vacuum}}\scalebox{\scaling}{
}
\eq{\nameref{ax:hopf}}\scalebox{\scaling}{
}
\eq{\nameref{ax:bialgebra-Z-W}}\scalebox{\scaling}{
}
\eq{\nameref{ax:W-assoc}}\scalebox{\scaling}{
}
\qedhere
\end{align*}
\end{proof}

\begin{proof}[Proof of \Cref{lem:copy}]
When there are 0 outputs:
\def\fig{copy-k-0}
\begin{align*}
\zw\vdash~\scalebox{\scaling}{
}
\eq{\nameref{ax:W-assoc}}\scalebox{\scaling}{
}
\eq{\nameref{ax:Z-spider}\\\nameref{ax:bialgebra-Z-W}}\scalebox{\scaling}{
}
\eq{\nameref{ax:copy}}r^k\cdot\scalebox{\scaling}{
}
\eq{\nameref{ax:W-assoc}}r^k\cdot\scalebox{\scaling}{
}
\eq{\nameref{ax:one}}r^k
\end{align*}
and when there is at least one output:
\def\fig{copy-k}
\begin{align*}
\zw\vdash~\scalebox{\scaling}{
}
&\eq{\nameref{ax:W-assoc}}\scalebox{\scaling}{
}
\eq{\nameref{ax:bialgebra-Z-W}}\scalebox{\scaling}{
}
\eq{\nameref{ax:copy}}r^k\cdot\scalebox{\scaling}{
}
\eq{\nameref{ax:W-assoc}}r^k\cdot\scalebox{\scaling}{
}
\qedhere
\end{align*}
\end{proof}

\begin{proof}[Proof of \Cref{lem:sum-gen}]
\def\fig{sum-gen}
\begin{align*}
\zw\vdash~\scalebox{\scaling}{
}
&\eq{\nameref{ax:Z-spider}\\\nameref{ax:bialgebra-Z-W}}\scalebox{\scaling}{
}
\eq{\nameref{ax:sum}}\scalebox{\scaling}{
}
\eq{\nameref{ax:Z-spider}}\scalebox{\scaling}{
}
\qedhere
\end{align*}
\end{proof}

\begin{proof}[Proof of \Cref{lem:skewed-hopf}]
\def\fig{Z-W-loop-prf}
\begin{align*}
\zw\vdash~\scalebox{\scaling}{
}
&\eq{\nameref{ax:bialgebra-Z-W}}\scalebox{\scaling}{
}
\eq{\nameref{ax:bialgebra-Z-W}}...
\eq{\nameref{ax:bialgebra-Z-W}}\scalebox{\scaling}{
}
\eq{\nameref{ax:hopf}\\\ref{lem:ket0}}\scalebox{\scaling}{
}
\eq[~]{\nameref{ax:bialgebra-Z-W}\\\nameref{ax:W-assoc}\\\nameref{ax:W-unit}}\scalebox{\scaling}{
}
\qedhere
\end{align*}
\end{proof}

\begin{proof}[Proof of \Cref{lem:ket-sum}]
First, notice that if $k+\ell<d$, the result is a mere use of the definition of $
\scalebox{\scaling}{
	\begin{tikzpicture}
	\begin{pgfonlayer}{nodelayer}
		\node [style=ket] (0) at (-0.25, 0.25) {$k$};
		\node [style=none] (1) at (-0.25, -0.25) {};
	\end{pgfonlayer}
	\begin{pgfonlayer}{edgelayer}
		\draw (1.center) to (0);
	\end{pgfonlayer}
\end{tikzpicture}
}$, together with a use of Equation \nameref{ax:W-assoc}.
Otherwise, if $K\geq d$:
\def\fig{ket-sum}
\begin{align*}
\zw\vdash~\scalebox{\scaling}{
	\input{./figures/\fig/\fig_05.tikz}%
}
\eq[]{\nameref{ax:copy}}\scalebox{\scaling}{
	\input{./figures/\fig/\fig_06.tikz}%
}
\eq[]{\nameref{ax:W-assoc}\\\nameref{ax:Z-spider}}\scalebox{\scaling}{
	\input{./figures/\fig/\fig_07.tikz}%
}
\eq[]{\nameref{ax:hopf}\\\ref{lem:ket0}}\scalebox{\scaling}{
	\input{./figures/\fig/\fig_08.tikz}%
}
\eq[]{\nameref{ax:bialgebra-Z-W}}\scalebox{\scaling}{
	\input{./figures/\fig/\fig_09.tikz}%
}
\eq[]{\ref{lem:ket0}\\\nameref{ax:W-unit}}\scalebox{\scaling}{
	\input{./figures/\fig/\fig_10.tikz}%
}
\eq[]{\nameref{ax:copy}}0\cdot\scalebox{\scaling}{
	\input{./figures/\fig/\fig_11.tikz}%
}
\end{align*}
Hence:
\begin{align*}
\zw\vdash~\scalebox{\scaling}{
	\input{./figures/\fig/\fig_12.tikz}%
}
&\eq{\nameref{ax:W-assoc}}\scalebox{\scaling}{
	\input{./figures/\fig/\fig_13.tikz}%
}
\eq{}
\begin{cases}
\scalebox{\scaling}{
	\input{./figures/\fig/\fig_14.tikz}%
} & \text{ if }k+\ell<d\\[1.5em]
0\cdot\scalebox{\scaling}{
	\input{./figures/\fig/\fig_11.tikz}%
} & \text{ if }k+\ell\geq d
\end{cases}
\qedhere
\end{align*}
\end{proof}

\begin{proof}[Proof of \Cref{lem:qubit-hopf}]
\def\fig{qubit-Hopf-2}
\begin{align*}
\zw\vdash~\scalebox{\scaling}{
}
&\eq{\nameref{ax:one}}\scalebox{\scaling}{
}
\eq{\nameref{ax:copy}}\scalebox{\scaling}{
}
\eq{\nameref{ax:bialgebra-Z-W}\\\nameref{ax:W-assoc}\\\nameref{ax:Z-spider}}\scalebox{\scaling}{
}\\
&\eq{\nameref{ax:W-assoc}\\\nameref{ax:Z-spider}\\\nameref{ax:hopf}\\\nameref{ax:bialgebra-Z-W}}\scalebox{\scaling}{
}
\eq{\nameref{ax:W-assoc}\\\nameref{ax:W-unit}}\scalebox{\scaling}{
	\input{./figures/\fig/\fig_05.tikz}%
}
\eq{\nameref{ax:copy}}\scalebox{\scaling}{
	\input{./figures/\fig/\fig_06.tikz}%
}
\eq{\nameref{ax:one}}\scalebox{\scaling}{
	\input{./figures/\fig/\fig_07.tikz}%
}
\qedhere
\end{align*}
\end{proof}

\begin{proof}[Proof of \Cref{lem:cup-distrib-aux}]
\def\fig{NF-cup-distrib-lemma}
\begin{align*}
\zw\vdash~\scalebox{\scaling}{
}
&\eq{\nameref{ax:bialgebra-Z-W}}\scalebox{\scaling}{
}
\eq{\ref{lem:qubit-hopf}\\\ref{lem:copy}}\scalebox{\scaling}{
}
\eq{\nameref{ax:W-unit}}\scalebox{\scaling}{
}
\qedhere
\end{align*}
\end{proof}

\begin{proof}[Proof of \Cref{lem:Pascal}]
By induction on $n$. The base case $n=1$ is direct. Suppose we have proven the case $n$, then we can prove the case $n+1$ as follows:
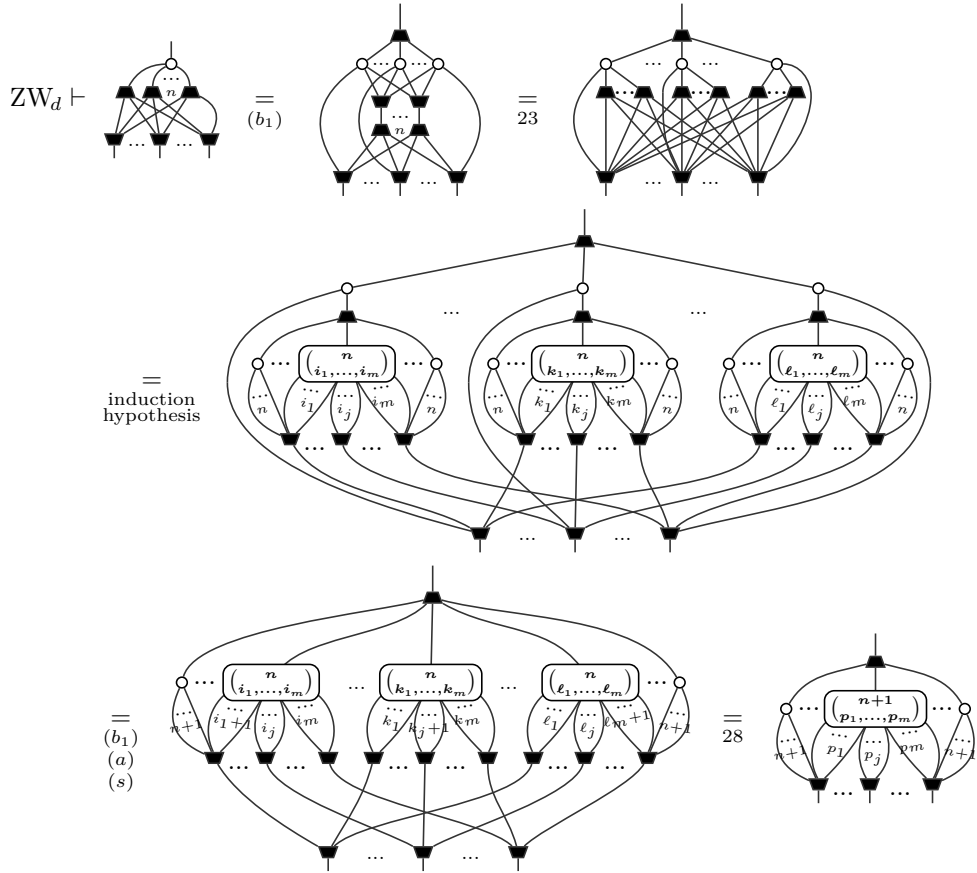
\def\fig{lemma-Pascal-multinomial-prf}
\begin{align*}
\zw\vdash&~\scalebox{\scaling}{
}
\eq{\nameref{ax:bialgebra-Z-W}}\scalebox{\scaling}{
}
\eq{\ref{lem:W-bialgebra-gen}}\scalebox{\scaling}{
}\\
&\eq{\text{induction}\\\text{hypothesis}}\scalebox{\scaling}{
}\\
&\eq{\nameref{ax:bialgebra-Z-W}\\\nameref{ax:W-assoc}\\\nameref{ax:Z-spider}}\scalebox{\scaling}{
}
\eq{\ref{lem:sum-gen}}\scalebox{\scaling}{
	\input{./figures/\fig/\fig_05.tikz}%
}
\end{align*}
where $\sum\limits_j i_j = \sum\limits_j k_j = \sum\limits_j \ell_j = n$ for all $j$, and $\sum\limits_j p_j = n+1$. The last equation is obtained by summing together all white nodes with the same connections. The nodes with connections $(p_1,...,p_m)$ are exactly the ones carrying parameters $\binom{n}{p_1-1,...,p_m},...,\binom{n}{p_1,...,p_m-1}$, which sum to $\binom{n+1}{p_1,...,p_m}$ thanks to the generalised Pascal formula. Assuming $\binom{n}{...,-1,...}=0$, the formula still holds when some of the $p_i$ are zero.
\end{proof}

\begin{proof}[Proof of \Cref{lem:Not-through-Z}]
\def\fig{Not-through-Z-prf}
\begin{align*}
\zw\vdash~\scalebox{\scaling}{
}
&\eq{\nameref{ax:copy}}\scalebox{\scaling}{
}
\eq{\nameref{ax:bialgebra-Z-W}}\scalebox{\scaling}{
}
\eq{\ref{lem:W-bialgebra-gen}}\scalebox{\scaling}{
}\\
&
\eq{\ref{lem:cup-distrib-aux}}\scalebox{\scaling}{
	\input{./figures/\fig/\fig_07.tikz}%
}
\eq{\ref{lem:skewed-hopf}\\\nameref{ax:W-assoc}}\scalebox{\scaling}{
	\input{./figures/\fig/\fig_08.tikz}%
}
\eq{\nameref{ax:W-unit}\\\nameref{ax:loop}}\scalebox{\scaling}{
	\input{./figures/\fig/\fig_10.tikz}%
}
\qedhere
\end{align*}
\end{proof}

\begin{proof}[Proof of \Cref{lem:partition}]
\def\fig{partition-prf}
\begin{align*}
\zw\vdash~\scalebox{\scaling}{
}
&\eq{\nameref{ax:W-assoc}}\scalebox{\scaling}{
}
\eq{\nameref{ax:bialgebra-W}}\scalebox{\scaling}{
}
\eq{\nameref{ax:copy}\\\nameref{ax:one}}\scalebox{\scaling}{
}\\
&\eq{\ref{lem:Pascal}}\scalebox{\scaling}{
}
\eq{\ref{lem:qubit-hopf}\\\nameref{ax:W-assoc}}\scalebox{\scaling}{
	\input{./figures/\fig/\fig_05.tikz}%
}
\eq{\ref{lem:Not-through-Z}}\scalebox{\scaling}{
	\input{./figures/\fig/\fig_06.tikz}%
}
\qedhere
\end{align*}
\end{proof}

\begin{proof}[Proof of \Cref{lem:ket-1}]
\def\fig{ket-1}
\begin{align*}
\zw\vdash~k!\cdot\scalebox{\scaling}{
}
&\eq{\nameref{ax:copy}}\scalebox{\scaling}{
}
\eq{\nameref{ax:W-unit}\\\nameref{ax:W-assoc}\\\nameref{ax:bialgebra-Z-W}\\\ref{lem:vacuum}}\scalebox{\scaling}{
}
\eq{\ref{lem:partition}}\scalebox{\scaling}{
}
\eq{\nameref{ax:W-assoc}}\scalebox{\scaling}{
}
\qedhere
\end{align*}
\end{proof}

\begin{proof}[Proof of \Cref{lem:dot-product}]
First:
\def\fig{braket-k}
\begin{align*}
\zw\vdash~\scalebox{\scaling}{
}
\eq{\ref{lem:ket-sum}}\scalebox{\scaling}{
}
\eq{\nameref{ax:copy}}\scalebox{\scaling}{
}
\eq{\ref{lem:ket-1}}k!\cdot\scalebox{\scaling}{
}
\eq{\ref{lem:braket-1-1}}k!
\end{align*}
Then, supposing w.l.o.g.~that $k>\ell$:
\def\fig{braket-k-l}
\begin{align*}
\zw\vdash~\scalebox{\scaling}{
}
&\eq{\ref{lem:ket-sum}}\scalebox{\scaling}{
}
\eq{\ref{lem:copy}}\scalebox{\scaling}{
}
\eq{\ref{lem:skewed-hopf}}\scalebox{\scaling}{
}
\eq{\nameref{ax:W-unit}\\\nameref{ax:bialgebra-Z-W}}\scalebox{\scaling}{
}\\
&\eq{\ref{lem:ket-1}}(k-l)!\cdot\scalebox{\scaling}{
	\input{./figures/\fig/\fig_05.tikz}%
}
\eq{\ref{lem:ket0}}(k-l)!\cdot\scalebox{\scaling}{
	\input{./figures/\fig/\fig_06.tikz}%
}
\eq{\nameref{ax:copy}}0
\qedhere
\end{align*}
\end{proof}

\begin{proof}[Proof of \Cref{lem:not-through-phase}]
\def\fig{Not-through-phase}
\begin{align*}
\zw\vdash~\scalebox{\scaling}{
}
&\eq{\ref{lem:Z-id}\\\nameref{ax:Z-spider}}\scalebox{\scaling}{
}
\eq{\nameref{ax:bialgebra-Z-W}}\scalebox{\scaling}{
}
\eq{\nameref{ax:copy}}r\cdot\scalebox{\scaling}{
}
\qedhere
\end{align*}
\end{proof}

\begin{proof}[Proof of \Cref{lem:mult-to-Z-effect}]
First notice that the equation derives from \nameref{ax:bialgebra-Z-W} when $k=0$. Then, for $0<k<d$:
\def\fig{multiplier-to-Z-effect}
\begin{align*}
\zw\vdash~\scalebox{\scaling}{
}
&\eq{\ref{lem:ket-1}}\frac1{k!}\cdot\scalebox{\scaling}{
}
\eq{\nameref{ax:bialgebra-Z-W}\\\nameref{ax:Z-spider}\\\nameref{ax:W-assoc}}\frac1{k!}\cdot\scalebox{\scaling}{
}
\eq{\nameref{ax:W-assoc}\\`\nameref{ax:Z-spider}\\\nameref{ax:bialgebra-Z-W}}\frac1{k!}\cdot\scalebox{\scaling}{
}\\
&\eq{\nameref{ax:Z-spider}\\\ref{lem:Z-id}\\\nameref{ax:W-assoc}}\frac1{k!}\cdot\scalebox{\scaling}{
}
\eq{\nameref{ax:bialgebra-Z-W}\\\nameref{ax:Z-spider}\\\nameref{ax:W-assoc}}\frac1{k!}\cdot\scalebox{\scaling}{
	\input{./figures/\fig/\fig_05.tikz}%
}
\eq{\ref{lem:ket-1}}\scalebox{\scaling}{
	\input{./figures/\fig/\fig_06.tikz}%
}
\qedhere
\end{align*}
\end{proof}

\begin{proof}[Proof of \Cref{lem:Z-loop-on-not}]
If $k=0$, the equation derives from \Cref{lem:copy} and \Cref{lem:vacuum}. Else, if $k<d$:
\def\fig{Z-loop-on-P}
\begin{align*}
\zw\vdash~&\scalebox{\scaling}{
}
\eq[~]{\nameref{ax:bialgebra-Z-W}\\\nameref{ax:Z-spider}}\scalebox{\scaling}{
}
\eq[~]{\ref{lem:Not-through-Z}}\scalebox{\scaling}{
}
\eq[~]{\nameref{ax:Z-spider}\\\ref{lem:partition}}\scalebox{\scaling}{
}\\
&\eq[~]{\ref{lem:Z-id}\\\nameref{ax:bialgebra-Z-W}\\\nameref{ax:Z-spider}\\\nameref{ax:W-assoc}}\scalebox{\scaling}{
}
\eq[~]{\nameref{ax:bialgebra-Z-W}\\\nameref{ax:W-assoc}}\scalebox{\scaling}{
	\input{./figures/\fig/\fig_05.tikz}%
}
\eq[~]{\ref{lem:copy}\\\ref{lem:dot-product}} k!\scalebox{\scaling}{
	\input{./figures/\fig/\fig_06.tikz}%
}
\eq[~]{\nameref{ax:Z-spider}\\\ref{lem:partition}} k!\scalebox{\scaling}{
	\input{./figures/\fig/\fig_07.tikz}%
}\\
&\eq[~]{\ref{lem:mult-to-Z-effect}\\\ref{lem:not-through-phase}}\scalebox{\scaling}{
	\input{./figures/\fig/\fig_08.tikz}%
}
\eq[~]{\ref{lem:Not-through-Z}}\scalebox{\scaling}{
	\input{./figures/\fig/\fig_09.tikz}%
}
\eq[~]{\nameref{ax:Z-spider}}\scalebox{\scaling}{
	\input{./figures/\fig/\fig_10.tikz}%
}
\qedhere
\end{align*}
\end{proof}

\begin{proof}[Proof of \Cref{lem:W-bialgebra-qubit-ctxt}]
To apply \nameref{ax:bialgebra-W}, we would need to have $k_i$ connections between the leftmost W-node and the $i$-th W-node of the bottom of the bialgebra, but so far we only have one. We can get more, in the following way:
\def\fig{W-bialgebra-context}
\begin{align*}
\zw\vdash~\scalebox{\scaling}{
}
\eq{\ref{lem:ket-1}}\frac1{(d-1)!}\scalebox{\scaling}{
}
\eq{\nameref{ax:bialgebra-Z-W}\\\nameref{ax:Z-spider}\\\nameref{ax:W-assoc}}\frac1{(d-1)!}\scalebox{\scaling}{
}
\end{align*}
as $k_i<d$, we now have enough connections to apply \nameref{ax:bialgebra-W}. Doing so and undoing the transformations on the left part, we get:
\begin{align*}
\zw\vdash~\frac1{(d-1)!}\scalebox{\scaling}{
}
&\eq[]{\ref{lem:W-bialgebra-gen}}\frac1{(d-1)!}\scalebox{\scaling}{
}
\eq[]{}\scalebox{\scaling}{
}
\qedhere
\end{align*}
\end{proof}

\begin{proof}[Proof of \Cref{prop:NF-preserves-semantics}]
First notice that:
\begin{align*}
\interp{
\scalebox{\scaling}{
	\input{./figures/qubit-W.tikz}%
}} 
&= \ket{1,0,...,0}+\ket{0,1,...,0}+\ldots+\ket{0,0,...,1}\\
&=\interp{
\scalebox{\scaling}{
	\begin{tikzpicture}
	\begin{pgfonlayer}{nodelayer}
		\node [style=ket] (0) at (0.5, 0.25) {$1$};
		\node [style=none] (1) at (0.5, -0.25) {};
	\end{pgfonlayer}
	\begin{pgfonlayer}{edgelayer}
		\draw (1.center) to (0);
	\end{pgfonlayer}
\end{tikzpicture}
}
\scalebox{\scaling}{
	\begin{tikzpicture}
	\begin{pgfonlayer}{nodelayer}
		\node [style=dot] (0) at (0.5, 0.25) {};
		\node [style=none] (1) at (0.5, -0.25) {};
	\end{pgfonlayer}
	\begin{pgfonlayer}{edgelayer}
		\draw (1.center) to (0);
	\end{pgfonlayer}
\end{tikzpicture}
}...
\scalebox{\scaling}{
	\begin{tikzpicture}
	\begin{pgfonlayer}{nodelayer}
		\node [style=dot] (0) at (0.5, 0.25) {};
		\node [style=none] (1) at (0.5, -0.25) {};
	\end{pgfonlayer}
	\begin{pgfonlayer}{edgelayer}
		\draw (1.center) to (0);
	\end{pgfonlayer}
\end{tikzpicture}
}}
+\interp{
\scalebox{\scaling}{
	\begin{tikzpicture}
	\begin{pgfonlayer}{nodelayer}
		\node [style=dot] (0) at (0.5, 0.25) {};
		\node [style=none] (1) at (0.5, -0.25) {};
	\end{pgfonlayer}
	\begin{pgfonlayer}{edgelayer}
		\draw (1.center) to (0);
	\end{pgfonlayer}
\end{tikzpicture}
}
\scalebox{\scaling}{
	\begin{tikzpicture}
	\begin{pgfonlayer}{nodelayer}
		\node [style=ket] (0) at (0.5, 0.25) {$1$};
		\node [style=none] (1) at (0.5, -0.25) {};
	\end{pgfonlayer}
	\begin{pgfonlayer}{edgelayer}
		\draw (1.center) to (0);
	\end{pgfonlayer}
\end{tikzpicture}
}...
\scalebox{\scaling}{
	\begin{tikzpicture}
	\begin{pgfonlayer}{nodelayer}
		\node [style=dot] (0) at (0.5, 0.25) {};
		\node [style=none] (1) at (0.5, -0.25) {};
	\end{pgfonlayer}
	\begin{pgfonlayer}{edgelayer}
		\draw (1.center) to (0);
	\end{pgfonlayer}
\end{tikzpicture}
}}
+ \ldots
+\interp{
\scalebox{\scaling}{
	\begin{tikzpicture}
	\begin{pgfonlayer}{nodelayer}
		\node [style=dot] (0) at (0.5, 0.25) {};
		\node [style=none] (1) at (0.5, -0.25) {};
	\end{pgfonlayer}
	\begin{pgfonlayer}{edgelayer}
		\draw (1.center) to (0);
	\end{pgfonlayer}
\end{tikzpicture}
}
\scalebox{\scaling}{
	\begin{tikzpicture}
	\begin{pgfonlayer}{nodelayer}
		\node [style=dot] (0) at (0.5, 0.25) {};
		\node [style=none] (1) at (0.5, -0.25) {};
	\end{pgfonlayer}
	\begin{pgfonlayer}{edgelayer}
		\draw (1.center) to (0);
	\end{pgfonlayer}
\end{tikzpicture}
}...
\scalebox{\scaling}{
	\begin{tikzpicture}
	\begin{pgfonlayer}{nodelayer}
		\node [style=ket] (0) at (0.5, 0.25) {$1$};
		\node [style=none] (1) at (0.5, -0.25) {};
	\end{pgfonlayer}
	\begin{pgfonlayer}{edgelayer}
		\draw (1.center) to (0);
	\end{pgfonlayer}
\end{tikzpicture}
}}
\end{align*}
In a given term, each $\ket0$ will merely cancel the Z-spider connected to it, and bring no contribution to the resulting state. Every contribution is brought by the $\ket1$. The number of parallel edges dictates to which basis state the contribution will go to. Consider for instance the $i$th term:
\def\fig{NF-explained}
\begin{align*}
\interp{\scalebox{\scaling}{
}}
&=\interp{\scalebox{\scaling}{
}}
=\interp{\scalebox{\scaling}{
}}\\
&=\frac{r_i}{\sqrt{x_1^{i}!...x_n^{i}!}}\interp{\scalebox{\scaling}{
}}
=\frac{r_i}{\sqrt{x_1^{i}!...x_n^{i}!}}\interp{\scalebox{\scaling}{
}}\\
&=r_i\ket{x^i_1,...,x^i_n}
\end{align*}
We do indeed recover the $i$th term in $\ket\psi$. We have something similar for all terms, so in the end $\interp{\mathcal N(\ket\psi)} = \ket\psi$.
\end{proof}

\begin{proof}[Proof of \Cref{prop:NF-tensor}]
\def\scaling{0.8}
Putting two diagrams in normal form side by side, we get:
\def\fig{NF-tensor}
\begin{align*}
\zw\vdash~~&\scalebox{\scaling}{
}
\eq{\ref{lem:copy}}\scalebox{\scaling}{
}
\eq{\nameref{ax:bialgebra-Z-W}}\scalebox{\scaling}{
}\\
&\eq{\ref{lem:W-bialgebra-gen}}\scalebox{\scaling}{
}
\eq{\nameref{ax:bialgebra-Z-W}\\\nameref{ax:W-unit}}\scalebox{\scaling}{
}\\
&\eq{\nameref{ax:bialgebra-Z-W}\\\nameref{ax:Z-spider}}\scalebox{\scaling}{
	\input{./figures/\fig/\fig_05.tikz}%
}
\eq{\nameref{ax:Z-spider}\\\nameref{ax:bialgebra-Z-W}}\scalebox{\scaling}{
	\input{./figures/\fig/\fig_06.tikz}%
}\\
&\eq{\nameref{ax:W-unit}\\\nameref{ax:bialgebra-Z-W}}\scalebox{\scaling}{
	\input{./figures/\fig/\fig_07.tikz}%
}
\eq{\nameref{ax:Z-spider}\\\ref{lem:Z-id}}\scalebox{\scaling}{
	\input{./figures/\fig/\fig_08.tikz}%
}
\end{align*}
The diagram obtained at the end of this process is indeed in normal form.
\end{proof}

\begin{proof}[Proof of \Cref{prop:NF-cup}]
First, we can show that the cup can ``distribute'' to each white node in the normal form:
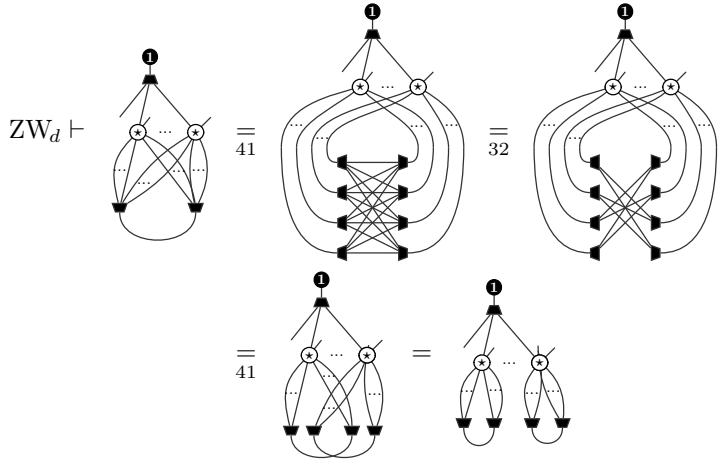
\def\fig{NF-cup-distrib}
{\def\scaling{0.8}
\begin{align*}
\zw\vdash~\scalebox{\scaling}{
}
&\eq{\ref{lem:W-bialgebra-qubit-ctxt}}\scalebox{\scaling}{
}
\eq{\ref{lem:cup-distrib-aux}}\scalebox{\scaling}{
}\\
&\eq{\ref{lem:W-bialgebra-qubit-ctxt}}\scalebox{\scaling}{
}
\eq{}\scalebox{\scaling}{
}
\end{align*}}
We can then treat each white node independently. If a white node had $k$ connections with the first output, and $\ell$ with the second one, with $k\neq \ell$, then it can simply be removed as:
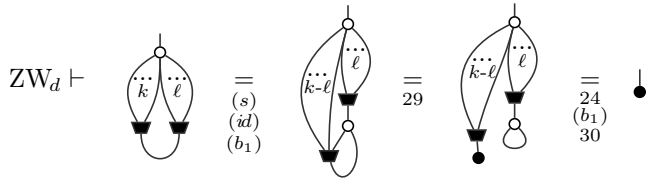
\def\fig{NF-cup-k-l}
\begin{align*}
\zw\vdash~~
\scalebox{\scaling}{
}
\eq{\nameref{ax:Z-spider}\\\nameref{ax:W-unit}\\\nameref{ax:bialgebra-Z-W}}\scalebox{\scaling}{
}
\eq{\ref{lem:skewed-hopf}}\scalebox{\scaling}{
}
\eq{\ref{lem:vacuum}\\\nameref{ax:bialgebra-Z-W}\\\ref{lem:ket-sum}}\scalebox{\scaling}{
}
\end{align*}
This $\scalebox{\scaling}{
}$ then copies through the white node (with \nameref{ax:bialgebra-Z-W}) and gets absorbed by the top and bottom W-nodes of the normal form (by \nameref{ax:W-unit}). If $k=\ell$, however the "cup gadget" disappears and the white node gets a new parameter, as:
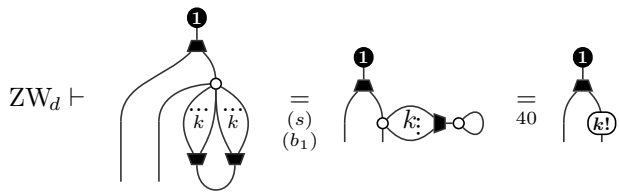
\def\fig{NF-cup-k}
\begin{align*}
\zw\vdash~~
\scalebox{\scaling}{
}
\eq{\nameref{ax:Z-spider}\\\nameref{ax:bialgebra-Z-W}}\scalebox{\scaling}{
}
\eq{\ref{lem:Z-loop-on-not}}\scalebox{\scaling}{
}
\end{align*}
After the last two simplifications, it is possible that two white nodes end up with exactly the same connections. It is then possible to merge them (performing the sum of their parameters), using \Cref{lem:sum-gen}. 
After doing all these simplifications, the diagram is again in normal form.
\end{proof}

\begin{proof}[Proof of \Cref{lem:W21-to-NF}]
\def\fig{app-W21-to-NF}
By simply using \nameref{ax:W-assoc}, we get:
{\def\scaling{0.8}
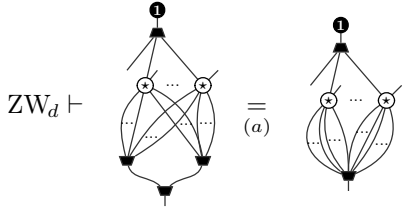
\begin{align*}
\zw\vdash~~
\scalebox{\scaling}{
}
\eq{\nameref{ax:W-assoc}}\scalebox{\scaling}{
}
\end{align*}}
then, if a white node ends up with more than $d-1$ connections with the resulting W-node, we can remove it by \nameref{ax:hopf} (again the \def\fig{NF-cup-k-l}$\scalebox{\scaling}{
}$ generated by the rule copies through the white node and gets absorbed by the top and bottom W-nodes of the normal form, and white nodes with the same connections can be merged with \Cref{lem:sum-gen}). If two nodes end up with the same connections, they can be merged using \Cref{lem:sum-gen}. We end up with a diagram in normal form.
\end{proof}

\begin{proof}[Proof of \Cref{prop:NF-generators}]
We deal with each generator one at a time:\\[1em]
\labelitemi~~ \def\fig{Z-flex}$\scalebox{\scaling}{
}$: 
First, if $n=1$:
\def\fig{Z-NF-prf}
\begin{align*}
\zw\vdash~&~~\scalebox{\scaling}{
}
\eq{\nameref{ax:decomp-Z}}\frac1{(d-1)!}\cdot\scalebox{\scaling}{
}
\eq{\nameref{ax:copy}}\frac1{(d-1)!}\cdot\scalebox{\scaling}{
}
\eq{\ref{lem:Pascal}}\frac1{(d-1)!}\cdot\scalebox{\scaling}{
}\\
&\eq{\nameref{ax:bialgebra-Z-W}\\\nameref{ax:Z-spider}}\frac1{(d-1)!}\cdot\scalebox{\scaling}{
}
\eq{\textit{Prop}.~\ref{prop:NF-cup}}\frac1{(d-1)!}\cdot\scalebox{\scaling}{
	\input{./figures/\fig/\fig_05.tikz}%
}\\
&\eq{\nameref{ax:copy}}\scalebox{\scaling}{
	\input{./figures/\fig/\fig_06.tikz}%
}
\eq{\nameref{ax:bialgebra-Z-W}\\\nameref{ax:Z-spider}}\scalebox{\scaling}{
	\input{./figures/\fig/\fig_07.tikz}%
}
\end{align*}
Then, if $n>1$:
\def\fig{NF-Z}
\begin{align*}
\zw\vdash~~
\scalebox{\scaling}{
}
\eq{\nameref{ax:Z-spider}}\scalebox{\scaling}{
}
\eq{}\scalebox{\scaling}{
}
\eq{\nameref{ax:bialgebra-Z-W}\\\nameref{ax:Z-spider}}\scalebox{\scaling}{
}
\end{align*}
If $n=0$, then $\scalebox{\scaling}{
}$ can be obtained by compositions: $\scalebox{\scaling}{
}\eq{}\scalebox{\scaling}{
	\input{./figures/\fig/\fig_05.tikz}%
}$. It can hence be turned in normal form by Propositions \ref{prop:NF-tensor} and \ref{prop:NF-cup}.\\[1em]
\labelitemi~~ $
\scalebox{\scaling}{
	\input{./figures/W-bent.tikz}%
}$: First, if $n=0$:
\def\fig{NF-ket-0}
\begin{align*}
\zw\vdash~~
\scalebox{\scaling}{
}
\eq{\nameref{ax:one}\\\nameref{ax:W-unit}}\scalebox{\scaling}{
}
\end{align*}
Then, when $n=1$:
\def\fig{NF-W-1}
\begin{align*}
\zw\vdash~\scalebox{\scaling}{
}
\eq{\nameref{ax:W-unit}\\\ref{lem:Z-id}}\scalebox{\scaling}{
}
\end{align*}
which can be put in normal form thanks to the previous item. For the general case, thanks to:
\def\fig{NF-W}
\begin{align*}
\zw\vdash~~
\scalebox{\scaling}{
}
\eq{}\scalebox{\scaling}{
}
\end{align*}
the $n$-ary W-node can be seen as applying $
\scalebox{\scaling}{
	\input{./figures/W21.tikz}%
}$ to pairs of outputs (until only one is left) of the normal form of $
\scalebox{\scaling}{
	\begin{tikzpicture}
	\begin{pgfonlayer}{nodelayer}
		\node [style=none] (0) at (-0.25, -0.125) {};
		\node [style=none] (1) at (0.25, -0.125) {};
		\node [style=white dot] (2) at (0, 0.125) {};
	\end{pgfonlayer}
	\begin{pgfonlayer}{edgelayer}
		\draw [in=90, out=-180, looseness=1.25] (2) to (0.center);
		\draw [in=90, out=0, looseness=1.25] (2) to (1.center);
	\end{pgfonlayer}
\end{tikzpicture}
}...
\scalebox{\scaling}{
	\begin{tikzpicture}
	\begin{pgfonlayer}{nodelayer}
		\node [style=none] (0) at (-0.25, -0.125) {};
		\node [style=none] (1) at (0.25, -0.125) {};
		\node [style=white dot] (2) at (0, 0.125) {};
	\end{pgfonlayer}
	\begin{pgfonlayer}{edgelayer}
		\draw [in=90, out=-180, looseness=1.25] (2) to (0.center);
		\draw [in=90, out=0, looseness=1.25] (2) to (1.center);
	\end{pgfonlayer}
\end{tikzpicture}
}$ (notice that the diagram obtained by swapping 2 outputs of a normal form is directly in normal form). Thanks to \Cref{lem:W21-to-NF}, the obtained diagram can be put in normal form.\\[1em]
\labelitemi~~ \def\fig{NF-ket-1}$\scalebox{\scaling}{
}$: we simply have
\begin{align*}
\zw\vdash~~
\scalebox{\scaling}{
}
\eq{\ref{lem:Z-id}}\scalebox{\scaling}{
}
\end{align*}
By \Cref{prop:NF-tensor} and \Cref{lem:W21-to-NF}, all the $
\scalebox{\scaling}{
	\begin{tikzpicture}
	\begin{pgfonlayer}{nodelayer}
		\node [style=ket] (0) at (-0.25, 0.25) {$k$};
		\node [style=none] (1) at (-0.25, -0.25) {};
	\end{pgfonlayer}
	\begin{pgfonlayer}{edgelayer}
		\draw (1.center) to (0);
	\end{pgfonlayer}
\end{tikzpicture}
}$ can be put in normal form as well.
\end{proof}

\section{Lemmas and Proofs for Minimality of Qudit}
\label{sec:proof_minimality}
In this section, we detail the proof of \Cref{thm:minimality}. For each equation $(eq)$, we define a non-standard semantics $\interp{\dots}_{(eq)}$ for which every equation is sound except $(eq)$. Contrary to $\interp{\dots}$, some of those semantics will be non-compositional, meaning that it is not enough to check that the equations are sound ``alone'', one needs to check explicitly that the equations are sound when applied to a fragment of a bigger diagram. 
Before that, we recall the definition of effective Z-path,  clarify what we mean by ``up to a scalar factor'', and prove a lemma that will be relevant to handle \nameref{ax:bialgebra-W}.

\begin{definition}[Sole effective output of a W-node]
	Let $D$ be a diagram with a collections of $n$ distinguished $W$-node. Let's call `$a_1$',$\dots$,`$a_n$' one output of each, as shown in \Cref{eq:single-out-W-app} below.
	We say that wires `$a_1$',$\dots$,`$a_n$' are \emph{jointly} the \emph{sole effective outputs} of the W-nodes if \Cref{eq:sole-effective-output-app} is satisfied:\\[-1em]
		\begin{equation}
		\label{eq:single-out-W-app}
		
\scalebox{\scaling}{
	\input{./figures/D.tikz}%
} = 
\scalebox{\scaling}{
	\input{./figures/sole-effective-output-with.tikz}%
}
		\end{equation}
		\begin{equation}
		\interp{
\scalebox{\scaling}{
	\input{./figures/sole-effective-output-without.tikz}%
}}=\interp{
\scalebox{\scaling}{
	\input{./figures/sole-effective-output-ket0.tikz}%
}}
		\end{equation}
\end{definition}
In particular, if we only consider one W-node that has a sole effective output -- we call such a node a \emph{trivial W-node} -- it follows that:
\def\fig{sole-effective-output}
\begin{equation}
\label{eq:sole-effective-output-app}
\interp{\scalebox{\scaling}{
}} = \interp{\scalebox{\scaling}{
}}
\end{equation}
\begin{definition}[Effective Z-path]
	An effective Z-path of a diagram $D$ is a path going from a boundary to another (inputs and/or outputs), such that:
	\begin{itemize}
		\item For each W-node it goes through, it cannot use two outputs of the W-node (so it must use its input). And considering all those W-nodes at once, those outputs must be jointly the sole effective outputs of those W-nodes.
		\item There exists a state $\ket{\phi}$ that is not of the form $\lambda \ket{0}$ or $\lambda \ket{1}$ for some $\lambda \in \mathbb{C}$, that \emph{inhabits} the path. That is, if we feed to $\interp{D}$ the state $\ket{\phi}$ and/or the effect $\bra{\phi}$ to the two inputs and/or outputs of the path, then the result is not $\vec 0$. 
	\end{itemize}
\end{definition}

\begin{definition}[Up to scalar factors]
	The category ``$\cat{Qudit}_d$ up to scalar factors'' has the same objects as $\cat{Qudit}_d$ and its morphisms are equivalence classes of morphisms of $\cat{Qudit}_d$ for the following relation: $f \equiv g \iff \exists \lambda \neq 0, f = \lambda \cdot g$. 
\end{definition}

\begin{lemma}\label{lem:double-W-simplification-app}~\\
	Let $D$ be a diagram of the form shown in \Cref{eq:double-W-app}, such that we have \Cref{eq:double-W-simplification-app}.
		\begin{equation}
		\label{eq:double-W-app}
		
\scalebox{\scaling}{
	\input{./figures/D.tikz}%
} = 
\scalebox{\scaling}{
	\input{./figures/lem-double-W.tikz}%
}
		\end{equation}
		\begin{equation}
		\label{eq:double-W-simplification-app}
		\interp{
\scalebox{\scaling}{
	\input{./figures/D.tikz}%
}} = \interp{
\scalebox{\scaling}{
	\input{./figures/lem-W-above.tikz}%
}} = \interp{
\scalebox{\scaling}{
	\input{./figures/lem-W-below.tikz}%
}}
		\end{equation}
		Then \Cref{eq:double-W-without-app} is satisfied:
		\begin{equation}
		\label{eq:double-W-without-app}
		\interp{
\scalebox{\scaling}{
	\input{./figures/lem-W-without.tikz}%
}} = \interp{
\scalebox{\scaling}{
	\input{./figures/lem-W-ket0.tikz}%
}}
		\end{equation}
\end{lemma}
\begin{proof}
	We sometimes write $
\scalebox{\scaling}{
	\begin{tikzpicture}
	\begin{pgfonlayer}{nodelayer}
		\node [style=ket] (0) at (-0.25, 0.25) {$0$};
		\node [style=none] (1) at (-0.25, -0.25) {};
	\end{pgfonlayer}
	\begin{pgfonlayer}{edgelayer}
		\draw (1.center) to (0);
	\end{pgfonlayer}
\end{tikzpicture}
}$ for	\def\fig{ket-0-def}$\scalebox{\scaling}{
}$. 
	For the sake of readability, we will make the proof in the case where there are only four W-nodes, two above and two below. But the same proof holds for any number. Using linearity, it is enough to prove that for any $a_1,\dots,a_n$ and any $b_1,\dots,b_m$, we have
	\[ \interp{
\scalebox{\scaling}{
	\input{./figures/lem-proof-W-without.tikz}%
}} = \interp{
\scalebox{\scaling}{
	\input{./figures/lem-proof-W-ket0.tikz}%
}}\]
	where
	\[ 
\scalebox{\scaling}{
	\input{./figures/lem-proof-W-without.tikz}%
} := 
\scalebox{\scaling}{
	\input{./figures/lem-proof-a1an-b1bn.tikz}%
}\]
	
	Decomposing $\interp{D''}$, there exist some coefficients $\lambda_{\substack{x_1,x_2\\y_1,y_2}} \in \mathbb{C}$ such that:
	\[ \interp{D''} = \sum_{\substack{x_1,x_2\\y_1,y_2}} \lambda_{\substack{x_1,x_2\\y_1,y_2}} \interp{
\scalebox{\scaling}{
	\input{./figures/lem-proof-x1x2-y1y2.tikz}%
}} \]
	In the following, we extend the notation $\lambda_{\substack{x_1,x_2\\y_1,y_2}}$ to be $0$ when at least one index is out of range. By simple computation, we obtain:
	\[ \interp{
\scalebox{\scaling}{
	\input{./figures/lem-proof-double-W.tikz}%
}} =  \sum_{\substack{x_1,x_2\\y_1,y_2}} \sum_{\substack{k_{11}\\k_{12}\\k_{21}\\k_{22}}} \frac{1}{k_{11}!k_{12}!k_{21}!k_{22}!} \lambda_{\substack{x_1+k_{11}+k_{12},x_2+k_{21}+k_{22}\\y_1+k_{11}+k_{21},y_2+k_{12}+k_{22}}} \interp{
\scalebox{\scaling}{
	\input{./figures/lem-proof-x1x2-y1y2.tikz}%
}} \]
	
	\[ \interp{
\scalebox{\scaling}{
	\input{./figures/lem-proof-W-above.tikz}%
}}=  \sum_{x_1,x_2} \sum_{\substack{k_{11}\\k_{12}\\k_{21}\\k_{22}}} \frac{1}{k_{11}!k_{12}!k_{21}!k_{22}!} \lambda_{\substack{x_1+k_{11}+k_{12},x_2+k_{21}+k_{22}\\k_{11}+k_{21},k_{12}+k_{22}}} \interp{
\scalebox{\scaling}{
	\input{./figures/lem-proof-x1x2-00.tikz}%
}} \]
	
	\[\interp{
\scalebox{\scaling}{
	\input{./figures/lem-proof-W-below.tikz}%
}} =  \sum_{y_1,y_2} \sum_{\substack{k_{11}\\k_{12}\\k_{21}\\k_{22}}} \frac{1}{k_{11}!k_{12}!k_{21}!k_{22}!} \lambda_{\substack{k_{11}+k_{12},k_{21}+k_{22}\\y_1+k_{11}+k_{21},y_2+k_{12}+k_{22}}} \interp{
\scalebox{\scaling}{
	\input{./figures/lem-proof-00-y1y2.tikz}%
}} \]
	Then, using \Cref{eq:double-W-simplification} they must all be equal, so comparing them on each element of the basis, we have the following identity for every $x_1,x_2,y_1,y_2$ except when $x_1 = x_2 = y_1 = y_2 = 0$:
	\[ \sum_{k_{11},k_{12},k_{21},k_{22}} \lambda_{\substack{x_1+k_{11}+k_{12},x_2+k_{21}+k_{22}\\k_{11}+k_{21},k_{12}+k_{22}}} = 0 \]
	Looking at the case where $x_1 = x_2 = y_1 = y_2 = d-1$, we now obtain that $\lambda_{\substack{d-1,d-1\\d-1,d-1}} = 0$. Similarly if all but one of the four are equal to $d-1$, we also have $\lambda_{\substack{x_1,x_2\\y_1,y_2}} = 0$. Then, by descending induction we can prove that we have $\lambda_{\substack{x_1,x_2\\y_1,y_2}} = 0$ for all $x_1,x_2,y_1,y_2$ except when $x_1 = x_2 = y_1 = y_2 = 0$. Said otherwise, we proved that we have: 
	\[ \interp{
\scalebox{\scaling}{
	\input{./figures/lem-proof-W-without.tikz}%
}} = \interp{
\scalebox{\scaling}{
	\input{./figures/lem-proof-W-ket0.tikz}%
}}\]
\end{proof}
\noindent We can now prove the minimality of every equation of the equational theory.
\begin{itemize}
	\setlength{\itemsep}{0.2em}
	\item[\nameref{ax:Z-spider}] The non-standard semantics $\interp{\dots}_{\nameref{ax:Z-spider}}$ is defined like $\interp{\dots}$ except for the semantics of the Z-spider and the global scalars where $r$ is replaced by its real part $\textup{Re}(r)$:
	\[ \interp{
\scalebox{\scaling}{
	\input{./figures/Z-spider.tikz}%
}}_{\nameref{ax:Z-spider}} = \sum_{k=0}^{d-1} (\textup{Re}(r))^k \sqrt{k!}^{n+m-2}\ketbra{k^m}{k^n} \qquad \qquad \interp{r}_{\nameref{ax:Z-spider}} = \textup{Re}(r) \]
	The Equation~\nameref{ax:Z-spider} is unsound by taking $s = r = \textbf{i}$. We remind the reader that in order to merge two global scalars into one ($r \otimes s$ into $rs$), we rely on Equation~\nameref{ax:Z-spider}, hence this operation is no longer sound either.
	All the equations where the coefficients of Z-spiders is $0$ or $1$ are trivially sound. Equations~\nameref{ax:bialgebra-Z-W}, \nameref{ax:sum} and \nameref{ax:copy} are easily checked to be sound.
	\item[\nameref{ax:W-assoc}] The non-standard semantics $\interp{\dots}_{\nameref{ax:W-assoc}} : \cat{ZW\!}_d \to \{0,1\}$ sends a diagram to the integer $1$ if it contains a non-trivial W-node with arity $>d$ and $0$ if it has only non-trivial W-nodes with arity $\leq d$. The Equation~\nameref{ax:W-assoc} is unsound by taking a diagram with two W-nodes of arity $d$ which are merged into one W-node of arity $2d-1$. Then, an important remark is that when applying any of the Equations to a fragment of a diagram, W-nodes that are not modified directly by the equation will keep their trivial or non-trivial status, and keep their arity. As such, we only need to check the W-nodes modified directly by the equations. All the equations that never manipulate W-nodes of arity $>d$ are trivially sound. For Equations~\nameref{ax:bialgebra-W} and \nameref{ax:bialgebra-Z-W}, we simply check that if all the W-nodes on the left-hand-side have a sole effective output, then so are all the W-nodes on the right-hand-side, and reciprocally.
	\item[\nameref{ax:W-unit}] The non-standard semantics $\interp{\dots}_{\nameref{ax:W-unit}}$ sends diagrams $n \to m$ to disjoints sets of pairs of $\{-1,\dots,-n,$ $1,\dots,m\}$, where each pair means that the corresponding inputs/outputs are linked to one another \emph{without} any node between them. In particular for the generators:	
	\begin{align*}
	\interp{~
\scalebox{\scaling}{
	\begin{tikzpicture}
	\begin{pgfonlayer}{nodelayer}
		\node [style=none] (0) at (0, 0.25) {};
		\node [style=none] (1) at (0, -0.25) {};
	\end{pgfonlayer}
	\begin{pgfonlayer}{edgelayer}
		\draw (0.center) to (1.center);
	\end{pgfonlayer}
\end{tikzpicture}
}~}_{\nameref{ax:W-unit}}
	&=\{\{-1,1\}\} & 
	\interp{
\scalebox{\scaling}{
	\input{./figures/swap.tikz}%
}}_{\nameref{ax:W-unit}}
	&=\{\{-1,2\},\{-2,1\}\} \\
	\interp{
\scalebox{\scaling}{
	\begin{tikzpicture}
	\begin{pgfonlayer}{nodelayer}
		\node [style=none] (0) at (-0.5, -0.125) {};
		\node [style=none] (1) at (0, -0.125) {};
		\node [style=none] (2) at (-0.25, 0.125) {};
	\end{pgfonlayer}
	\begin{pgfonlayer}{edgelayer}
		\draw [bend left=90, looseness=1.75] (0.center) to (1.center);
	\end{pgfonlayer}
\end{tikzpicture}
}}_{\nameref{ax:W-unit}}
	&=\{\{1,2\}\} & \interp{
\scalebox{\scaling}{
	\begin{tikzpicture}
	\begin{pgfonlayer}{nodelayer}
		\node [style=none] (0) at (0.5, 0.125) {};
		\node [style=none] (1) at (1, 0.125) {};
		\node [style=none] (2) at (0.75, -0.125) {};
	\end{pgfonlayer}
	\begin{pgfonlayer}{edgelayer}
		\draw [bend right=90, looseness=1.75] (0.center) to (1.center);
	\end{pgfonlayer}
\end{tikzpicture}
}}_{\nameref{ax:W-unit}} & =\{\{-1,-2\}\}\\
	\interp{
\scalebox{\scaling}{
	\input{./figures/Z-spider.tikz}%
}}_{\nameref{ax:W-unit}}
	&= \varnothing&
	\interp{
\scalebox{\scaling}{
	\input{./figures/W-n.tikz}%
}}_{\nameref{ax:W-unit}}
	&= \varnothing\\
	\interp{
\scalebox{\scaling}{
	\begin{tikzpicture}
	\begin{pgfonlayer}{nodelayer}
		\node [style=ket] (0) at (0, 0.125) {$1$};
		\node [style=none] (1) at (0, -0.25) {};
	\end{pgfonlayer}
	\begin{pgfonlayer}{edgelayer}
		\draw (1.center) to (0);
	\end{pgfonlayer}
\end{tikzpicture}
}}_{\nameref{ax:W-unit}}
	&= \varnothing&
	\interp{r}_{\nameref{ax:W-unit}}&= \varnothing\\
	\end{align*}
	This semantics is compositional, although writing the formula for the composition is quite technical because of the cups and caps:
	\[ 
	\interp{D_1 \otimes D_2}_{\nameref{ax:W-unit}} = \interp{D_1}_{\nameref{ax:W-unit}} \cup \left\{ \{f(a),f(b)\} ~\middle|~ \{a,b\} \in \interp{D_2}_{\nameref{ax:W-unit}}, f(a) = \begin{cases} a-n_1 &\text{if } a < 0 \\ a+m_1 &\text{if } a > 0 \end{cases} \right\}
\]
\[ 
\interp{D_2 \circ  D_1}_{\nameref{ax:W-unit}} = \left\{ \{a,b\} ~\middle|~ \begin{array}{l} \exists c_1,\dots,c_n \text{ with } a = c_1, b= c_n \text{ such that } \\ \forall k, \{c_k,c_{k+1}\} \in  F(\interp{D_1}_{\nameref{ax:W-unit}}) \cup G(\interp{D_2}_{\nameref{ax:W-unit}}) \end{array} \right\}
\]
\[ \text{ where } F(\interp{D_1}_{\nameref{ax:W-unit}}) = \left\{ \{f(a),f(b)\} ~\middle|~ \{a,b\} \in \interp{D_1}_{\nameref{ax:W-unit}}, f(x) = \begin{cases} x &\text{if } x < 0 \\ (0,x) &\text{if } x > 0 \end{cases} \right\}  \]
\[ \text{ and } G(\interp{D_2}_{\nameref{ax:W-unit}}) = \left\{ \{g(a),g(b)\} ~\middle|~ \{a,b\} \in \interp{D_2}_{\nameref{ax:W-unit}}, g(x) = \begin{cases} (0,-x) &\text{if } x < 0 \\ x &\text{if } x > 0 \end{cases} \right\}  \]
	For all the equations except Equation~\nameref{ax:W-unit}, their soundness is trivial as both side are $\varnothing$, and Equation~\nameref{ax:W-unit} is trivially unsound as $\varnothing \neq  \{\{1,1\}\}$.
	\item[\nameref{ax:hopf}] To each wire in a diagram $D$, we associate a number $0\leq k < d$ (or more graphically we annotate each wire by some number $k$)\footnote{The annotations produced here are upper bounds on the value of the kets that can go through the wires. As such, they are very closely related to the \emph{capacities} from \Cref{sec:fdhilb}.}. The procedure to do so is as follows:
	\begin{enumerate}
		\item annotate all wires with $d-1$
		\item rewrite the annotations using the following rules, until a fixed point is reached:\\
		$
		\labelitemii~~ \def\fig{ket-0-ofs}\scalebox{\scaling}{
}~\overset{a\neq0}\to~\scalebox{\scaling}{
}
		\hspace*{3.5em}
		\labelitemii~~ \def\fig{ket-1-ofs}\scalebox{\scaling}{
}~\overset{a>1}\to~\scalebox{\scaling}{
}
		\hspace*{3.5em}
		\labelitemii~~ \def\fig{W-node-ofs}\scalebox{\scaling}{
}\to\scalebox{\scaling}{
}
		$\\
		$
		\labelitemii~~ \def\fig{Z-spider-ofs}\scalebox{\scaling}{
}~\to~\scalebox{\scaling}{
}\quad \text{ with } a:=\min(a_i)
		$
	\end{enumerate}
	This simple procedure obviously terminates, as a step is only applied if at least one of the annotations is decreased. We then define $\interp{D}_{\nameref{ax:hopf}}$ as the $(n+m)$-tuple of all the final annotations of its inputs and outputs. Equation~\nameref{ax:hopf} is trivially unsound. 
	 When we apply any of the other equations to a fragment of $D$, the annotations of the wires outside of this rewritten fragment are unmodified as both sides of the equation behave the same with respect to the above procedure. Said otherwise, those equations are sound.	 
	 \item[\nameref{ax:bialgebra-Z-W}] Consider diagrams as graphs, and define an ``effective Z-path'' in the diagram as defined above. We consider the existence of such a path, that is $\interp{\dots}_{\nameref{ax:bialgebra-Z-W}} : \cat{ZW\!}_d \to \{0,1\}$ sends diagrams to $1$ if such a path exists and $0$ otherwise. We consider a diagram $D$ with an effective Z-path, and we apply one of the equations to a fragment of $D$:
	 \begin{itemize}
	 	\item If we apply the equation outside of the path, then that path remains an effective Z-path. In particular, Equations~\nameref{ax:hopf} (only inhabited by $\ket{0}$), \nameref{ax:sum}, \nameref{ax:one}, \nameref{ax:copy} (only inhabited by $\ket{1}$), \nameref{ax:loop} and \nameref{ax:decomp-Z} can never be on such a path, so they are trivially sound.
	 	\item If we apply the equation on the path, then we need to check that this path can still be drawn through the rewritten version of the diagram. Equations~\nameref{ax:Z-spider}, \nameref{ax:W-assoc}, \nameref{ax:W-unit}, trivially preserve such an effective path. Equation~\nameref{ax:bialgebra-W} follows from \Cref{lem:double-W-simplification-app}.
	 	Lastly, Equation~\nameref{ax:bialgebra-Z-W} is unsound in the simple case of $D$ being the left-hand side of the equation with $n=m=2$: the effective Z-path going from the first input to the second input (through the Z-spider) does not exists on the right-hand-side as the W-nodes are not trivial.
	 \end{itemize}
	 \item[\nameref{ax:bialgebra-W}] Consider diagrams as graphs, and define a ``W-path'' in the diagram as a path 1) that goes from a boundary to another boundary, 2) which cannot use two outputs of a W-node (if it goes through a W-node, it has to use the input edge) and 3) that does not go through a Z-spider.  We consider the existence of such a path, that is $\interp{\dots}_{\nameref{ax:bialgebra-Z-W}} : \cat{ZW\!}_d \to \{0,1\}$ sends diagrams to $1$ if such a path exists and $0$ otherwise.  We consider a diagram $D$ with a W-path, and we apply one of the equations to a fragment of $D$:
	 \begin{itemize}
	 	\item If we apply the equation outside of the path, then that path remains a W-path. In particular, Equations~\nameref{ax:Z-spider}, \nameref{ax:hopf}, \nameref{ax:bialgebra-Z-W}, \nameref{ax:sum}, \nameref{ax:one}, \nameref{ax:copy}, \nameref{ax:loop} and \nameref{ax:decomp-Z} can never be on such a path, so they are trivially sound.	 	
	 	\item If we apply the equation on the path, then we need to check that this path can still be drawn through the rewritten version of the diagram. Equations~\nameref{ax:W-assoc} and \nameref{ax:W-unit} trivially preserve such a W-path. Lastly, Equation~\nameref{ax:bialgebra-W} is unsound in the simple case of $D$ being the left-hand-side of the equation with $n=m=2$: the W-path going from the first input to the first output (through two W-nodes) does not exist on the right-hand-side because of the Z-spider blocking it.
	 \end{itemize}
 	\item[\nameref{ax:sum}] The non-standard semantics $\interp{\dots}_{\nameref{ax:sum}}$ is defined like $\interp{\dots}$ except for the semantics of the Z-spider and the global scalars where $r$ is replaced by its absolute value $|r|$:
 	\[ \interp{
\scalebox{\scaling}{
	\input{./figures/Z-spider.tikz}%
}}_{\nameref{ax:Z-spider}} = \sum_{k=0}^{d-1} |r|^k \sqrt{k!}^{n+m-2}\ketbra{k^m}{k^n} \qquad \qquad \interp{r}_{\nameref{ax:Z-spider}} = |r|\]
 	The Equation~\nameref{ax:sum} is unsound by taking $r = 1$ and $s = -1$.
 	All the equations where the coefficients of Z-spiders are $0$ or $1$ are trivially sound. Equations~\nameref{ax:Z-spider}, \nameref{ax:bialgebra-Z-W} and \nameref{ax:copy} are easily checked to be sound.
	\item[\nameref{ax:copy}] The non-standard semantics $\interp{\dots}_{\nameref{ax:copy}} : \cat{ZW\!}_d \to \{0,1\}$ sends a diagram to the integer $1$ if it contains a global non-real scalar and $0$ if it does not. This semantics is compositional. All the equations other than Equation~\nameref{ax:copy} are trivially sound, and Equation~\nameref{ax:copy} is unsound by taking $r = \textbf{i}$.  We remind the reader that in order to merge two global scalars into one ($r \otimes s$ into $rs$), we rely on Equation~\nameref{ax:copy}, hence this operation is no longer sound either.
	\item[\nameref{ax:one}] The non-standard semantics $\interp{\dots}_{\nameref{ax:one}} : \cat{ZW\!}_d \to \{0,1\}$ sends a diagram to the integer $1$ if it is non-empty and $0$ if it is empty. This semantics is compositional. All the equations other than Equation~\nameref{ax:one} are trivially sound, and Equation~\nameref{ax:one} is unsound with, for example, $k = 1$.
	\item[\nameref{ax:loop}] Let $\varpi:=e^{i\frac\pi{d-1}}$ be the first $(2d-2)$-th root of unity. Consider the $\dagger$-compact monoidal functor $F$ (i.e.~functor that preserves compositions, symmetry and compact structure) that maps the generators as follows:\\
	$
	\labelitemii~~ 
\scalebox{\scaling}{
	\begin{tikzpicture}
	\begin{pgfonlayer}{nodelayer}
		\node [style=ket] (0) at (0, 0.125) {$1$};
		\node [style=none] (1) at (0, -0.25) {};
	\end{pgfonlayer}
	\begin{pgfonlayer}{edgelayer}
		\draw (1.center) to (0);
	\end{pgfonlayer}
\end{tikzpicture}
}\mapsto \varpi
\scalebox{\scaling}{
	\begin{tikzpicture}
	\begin{pgfonlayer}{nodelayer}
		\node [style=ket] (0) at (0, 0.125) {$1$};
		\node [style=none] (1) at (0, -0.25) {};
	\end{pgfonlayer}
	\begin{pgfonlayer}{edgelayer}
		\draw (1.center) to (0);
	\end{pgfonlayer}
\end{tikzpicture}
}
	\hspace*{5em}
	\labelitemii~~ 
\scalebox{\scaling}{
	\input{./figures/Z-spider.tikz}%
}\mapsto 
\scalebox{\scaling}{
	\input{./figures/minimality-l-Z-spider.tikz}%
}
	\hspace*{5em}
	\labelitemii~~ 
\scalebox{\scaling}{
	\input{./figures/W-n.tikz}%
}\mapsto
\scalebox{\scaling}{
	\input{./figures/W-n.tikz}%
}
	$\\
	When $d>2$, all equations but \nameref{ax:loop} are preserved by $F$. Hence, we can simply define $\interp{\dots}_{\nameref{ax:loop}}$ to be $\interp{F(\dots)}$ and all the equations except Equation~\nameref{ax:loop} will be sound, and Equation~\nameref{ax:loop} is trivially unsound. We can make the argument work when $d=2$ by choosing any $\varpi \neq 0$ such that $\varpi^2\neq1$ and by working up to scalar factors.
	\item[\nameref{ax:decomp-Z}] The non-standard semantics $\interp{\dots}_{\nameref{ax:Z-spider}}$ is defined like $\interp{\dots}$ except for the semantics of the states where $
\scalebox{\scaling}{
	\begin{tikzpicture}
	\begin{pgfonlayer}{nodelayer}
		\node [style=ket] (0) at (0, 0.125) {$1$};
		\node [style=none] (1) at (0, -0.25) {};
	\end{pgfonlayer}
	\begin{pgfonlayer}{edgelayer}
		\draw (1.center) to (0);
	\end{pgfonlayer}
\end{tikzpicture}
}$ (and subsequently all $
\scalebox{\scaling}{
	\begin{tikzpicture}
	\begin{pgfonlayer}{nodelayer}
		\node [style=ket] (0) at (-0.25, 0.25) {$k$};
		\node [style=none] (1) at (-0.25, -0.25) {};
	\end{pgfonlayer}
	\begin{pgfonlayer}{edgelayer}
		\draw (1.center) to (0);
	\end{pgfonlayer}
\end{tikzpicture}
}$) is sent to $\ket{0}$, and except the fact that we work in $\cat{Qudit}_d$ up to scalar factors.
	Equation~\nameref{ax:decomp-Z} is trivially unsound. All the equations where no state appears are trivially sound, and Equations~\nameref{ax:one}, \nameref{ax:copy} and \nameref{ax:loop} are easily checked to be sound.
	\qedhere
\end{itemize}

\section{Lemmas and Proofs for Completeness of FdHilb}
\label{sec:appendix-fdhilb}
In this section, for ease of notation, we will omit the box over the capacity annotations. 
To get to the proof of completeness, we again require some lemmas beforehand.

\begin{lemma}
\label{lem:Z-id-qf}
\def\fig{Z-id-qf-prf}
\begin{align*}
\zwf\vdash~\scalebox{\scaling}{
}
\eq{}\scalebox{\scaling}{
	\input{./figures/\fig/\fig_05.tikz}%
}
\end{align*}
\end{lemma}

\begin{proof}
\def\fig{Z-id-qf-prf}
\begin{align*}
\zwf\vdash~\scalebox{\scaling}{
}
&\eq{\nameref{ax:Z-spider-qf}}\scalebox{\scaling}{
}
\eq{\nameref{ax:id-qf}}\scalebox{\scaling}{
}
\eq{\nameref{ax:W-assoc-qf}\\\nameref{ax:bialgebra-Z-W-qf}}\scalebox{\scaling}{
}
\eq{\nameref{ax:W-assoc-2-qf}}\scalebox{\scaling}{
}
\eq{\nameref{ax:W-assoc-qf}\\\nameref{ax:id-qf}}\scalebox{\scaling}{
	\input{./figures/\fig/\fig_05.tikz}%
}
\qedhere
\end{align*}
\end{proof}

\begin{lemma}
\label{lem:scalar-1-qf}
\def\fig{1-is-empty}
\begin{align*}
\zwf\vdash~1
\def\fig{scalar-1-is-idempotent}
\eq{}\scalebox{\scaling}{
}
\def\fig{1-is-empty}
\eq{}\scalebox{\scaling}{
}
\eq{}\scalebox{\scaling}{
}
\def\fig{braket-1-qf-aux}
\eq{}\scalebox{\scaling}{
}
\eq{}
\scalebox{\scaling}{
	\input{./figures/empty.tikz}%
}
\end{align*}
\end{lemma}

\begin{proof}
First, let us prove:
\def\fig{scalar-1-is-idempotent}
\begin{align*}
\zwf\vdash~1&\eq{\nameref{ax:scalar-qf}}
\scalebox{\scaling}{
}
\eq{\nameref{ax:Z-spider-qf}}\scalebox{\scaling}{
}
\eq{\nameref{ax:scalar-qf}}1\cdot\scalebox{\scaling}{
}
\eq{\nameref{ax:scalar-qf}}1\otimes1
\end{align*}
Then:
\def\fig{1-is-empty}
\begin{align*}
\zwf\vdash~ 1\eq[~]{}1\otimes1
\eq[~]{\nameref{ax:scalar-qf}}1\cdot\scalebox{\scaling}{
}
&\eq[~]{\nameref{ax:W-unit}\\\nameref{ax:W-assoc-qf}\\\nameref{ax:loop-removal-qf}}1\cdot\scalebox{\scaling}{
}
\eq[~]{\nameref{ax:decomp-Z-qf}}\scalebox{\scaling}{
}
\eq[~]{\nameref{ax:loop-removal-qf}}\!\!\scalebox{\scaling}{
}\!\!
\eq[~]{\nameref{ax:bialgebra-Z-W-qf}}\scalebox{\scaling}{
}
\eq[~]{\nameref{ax:bialgebra-W-qf}}
\scalebox{\scaling}{
	\input{./figures/empty.tikz}%
}
\end{align*}
Finally:
\def\fig{braket-1-qf-aux}
\begin{align*}
\zwf\vdash~ \scalebox{\scaling}{
}
&\eq{\nameref{ax:scalar-qf}}\scalebox{\scaling}{
}
\eq{\nameref{ax:loop-removal-qf}}\scalebox{\scaling}{
}
\eq{}
\scalebox{\scaling}{
	\input{./figures/empty.tikz}%
}
\qedhere
\end{align*}
\end{proof}

\begin{lemma}
\label{lem:embedding}
If $b\geq a$:
\def\fig{embedding}
\begin{align*}
\zwf\vdash~\scalebox{\scaling}{
}
\eq{\nameref{ax:bialgebra-W-qf}}\scalebox{\scaling}{
}
\eq{\nameref{ax:id-qf}}\scalebox{\scaling}{
}
\end{align*}
\end{lemma}

\begin{lemma}
\label{lem:ket-0-qf}
\def\fig{ket-0-qf-prf}
\begin{align*}
\zwf\vdash~\scalebox{\scaling}{
}
\eq{}\scalebox{\scaling}{
}
\end{align*}
\end{lemma}

\begin{proof}
\def\fig{ket-0-qf-prf}
\begin{align*}
\zwf\vdash~\scalebox{\scaling}{
}
&\eq{\nameref{ax:bialgebra-W-qf}}\scalebox{\scaling}{
}
\eq{\nameref{ax:hopf-qf}}\scalebox{\scaling}{
}
\eq{\nameref{ax:bialgebra-Z-W-qf}}\scalebox{\scaling}{
}
\eq{\nameref{ax:W-assoc-qf}}\scalebox{\scaling}{
}
\qedhere
\end{align*}
\end{proof}

\begin{lemma}
\label{lem:W-assoc-gen}
Equation \nameref{ax:W-assoc-2-qf} can be generalised to:
\def\fig{W-assoc-qf2-gen}
\begin{align*}
\zwf\vdash~\scalebox{\scaling}{
}
\eq{}\scalebox{\scaling}{
	\input{./figures/\fig/\fig_07.tikz}%
}
\end{align*}
\end{lemma}

\begin{proof}
\def\fig{W-assoc-qf2-gen-bis}
\begin{align*}
\zwf\vdash~\scalebox{\scaling}{
}
&\eq{\nameref{ax:W-assoc-qf}}\scalebox{\scaling}{
}
\eq{\nameref{ax:W-assoc-2-qf}}\scalebox{\scaling}{
}
\eq{\nameref{ax:W-assoc-qf}}\scalebox{\scaling}{
}
\qedhere
\end{align*}
\end{proof}

\begin{lemma}
\label{lem:W-bialgebra-ctxt}
\def\fig{W-bialgebra-qudit-from-qf}
\begin{align*}
\zwf\vdash~\scalebox{\scaling}{
}
\eq{}\scalebox{\scaling}{
	\input{./figures/\fig/\fig_06.tikz}%
}
\end{align*}
\end{lemma}

\begin{proof}
With $n$ the smallest between the number of inputs and the number of outputs:
\def\fig{W-bialgebra-qudit-from-qf}
\begin{align*}
\zwf\vdash~\scalebox{\scaling}{
}
&\eq{\ref{lem:embedding}}\scalebox{\scaling}{
}
\eq{\ref{lem:W-assoc-gen}}\scalebox{\scaling}{
}
\eq{\nameref{ax:bialgebra-W-qf}}\scalebox{\scaling}{
}\\
&\eq{\nameref{ax:bialgebra-Z-W-qf}\\\nameref{ax:W-assoc-qf}\\\nameref{ax:Z-spider-qf}}\scalebox{\scaling}{
}
\eq{\nameref{ax:bialgebra-Z-W-qf}\\\nameref{ax:Z-spider-qf}}\scalebox{\scaling}{
	\input{./figures/\fig/\fig_05.tikz}%
}
\eq{\nameref{ax:bialgebra-Z-W-qf}\\\nameref{ax:W-assoc-qf}}\scalebox{\scaling}{
	\input{./figures/\fig/\fig_06.tikz}%
}
\qedhere
\end{align*}
\end{proof}

\begin{lemma}
\label{lem:Pascal-qf}
\def\fig{lemma-Pascal-multinomial-qf}
\begin{align*}
\zwf\vdash~\scalebox{\scaling}{
}
\eq{}\scalebox{\scaling}{
}
\end{align*}
\end{lemma}

\begin{proof}
The proof is the same as that of \Cref{lem:Pascal}, where we can use \Cref{lem:W-bialgebra-ctxt} in place of \nameref{ax:bialgebra-W}.
\end{proof}

\begin{lemma}
\label{lem:copy-0-qf}
\def\fig{copy-0}
\begin{align*}
\zwf\vdash~\scalebox{\scaling}{
}
\eq{}\scalebox{\scaling}{
}
\end{align*}
\end{lemma}

\begin{proof}
The case where there are no outputs is dealt with
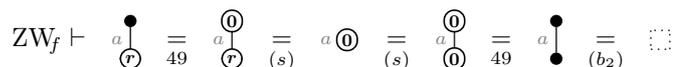
\def\fig{scalar-1-with-ket-0}
\begin{align*}
\zwf\vdash~\scalebox{\scaling}{
}
\eq{\ref{lem:ket-0-qf}}\scalebox{\scaling}{
}
\eq{\nameref{ax:Z-spider-qf}}\scalebox{\scaling}{
}
\eq{\nameref{ax:Z-spider-qf}}\scalebox{\scaling}{
}
\eq{\ref{lem:ket-0-qf}}\scalebox{\scaling}{
}
\eq{\nameref{ax:bialgebra-W-qf}}
\scalebox{\scaling}{
	\input{./figures/empty-diag.tikz}%
}
\end{align*}
The general case is a direct application of \nameref{ax:bialgebra-Z-W-qf} with $m=0$.
\end{proof}

\begin{lemma}
\label{lem:ket-k-forms}
\def\fig{ket-k-several-forms}
\begin{align*}
\zwf\vdash~\scalebox{\scaling}{
}
\eq{}\scalebox{\scaling}{
}
\eq{}\scalebox{\scaling}{
}
\end{align*}
\end{lemma}

\begin{proof}
\def\fig{ket-k-several-forms}
\begin{align*}
\zwf\vdash~\scalebox{\scaling}{
}
\eq{}\scalebox{\scaling}{
}
\eq{\nameref{ax:W-assoc-qf}}\scalebox{\scaling}{
}
\eq{\nameref{ax:W-assoc-qf}}\scalebox{\scaling}{
}
\eq{\nameref{ax:ket-1}}\scalebox{\scaling}{
}
\end{align*}
\end{proof}

\begin{lemma}
\label{lem:ket-k-dim-qf}
If $b\geq a\geq k \geq b_i$:
\def\fig{ket-k-larger-dim}
\begin{align*}
\zwf\vdash~\scalebox{\scaling}{
}
\eq{}\scalebox{\scaling}{
}
\qquad\text{and}\qquad
\def\fig{ket-k-smaller-dim-gen}
\zwf\vdash~\scalebox{\scaling}{
}
\eq{}\scalebox{\scaling}{
}
\end{align*}
\end{lemma}

\begin{proof}
\def\fig{ket-k-larger-dim}
\begin{align*}
\zwf\vdash~\scalebox{\scaling}{
}
\eq{\ref{lem:ket-k-forms}}\scalebox{\scaling}{
}
\eq{\nameref{ax:W-assoc-qf}}\scalebox{\scaling}{
}
\eq{\ref{lem:ket-k-forms}}\scalebox{\scaling}{
}
\end{align*}
\def\fig{ket-k-smaller-dim-gen}
\begin{align*}
\zwf\vdash~\scalebox{\scaling}{
}
&\eq{}\scalebox{\scaling}{
}
\eq{\nameref{ax:bialgebra-W-qf}}\scalebox{\scaling}{
}
\eq{\nameref{ax:id-qf}}\scalebox{\scaling}{
}
\qedhere
\end{align*}
\end{proof}

\begin{lemma}
\label{lem:zcp-1}
\def\fig{zcp-1}
\begin{align*}
\zwf\vdash~~\scalebox{\scaling}{
}
\eq{}\scalebox{\scaling}{
	\input{./figures/\fig/\fig_08.tikz}%
}
\end{align*}
\end{lemma}

\begin{proof}
First, we have:
\def\fig{zcp-1}
\begin{align*}
\zwf\vdash~~\scalebox{\scaling}{
}
\eq{\nameref{ax:bialgebra-Z-W-qf}}\scalebox{\scaling}{
}
\eq{\ref{lem:embedding}}\scalebox{\scaling}{
}
\eq{\nameref{ax:loop-removal-qf}\\\ref{lem:Z-id-qf}}\scalebox{\scaling}{
}
\end{align*}
Then:
\begin{align*}
\zwf\vdash~~\scalebox{\scaling}{
}
&\eq{}\scalebox{\scaling}{
	\input{./figures/\fig/\fig_05.tikz}%
}
\eq{\nameref{ax:bialgebra-Z-W-qf}}\scalebox{\scaling}{
	\input{./figures/\fig/\fig_06.tikz}%
}
\eq{\ref{lem:embedding}}\scalebox{\scaling}{
	\input{./figures/\fig/\fig_07.tikz}%
}
\eq{\nameref{ax:loop-removal-qf}}\scalebox{\scaling}{
	\input{./figures/\fig/\fig_08.tikz}%
}
\qedhere
\end{align*}
\end{proof}

\begin{lemma}
\label{lem:NF-qf-other-dim-qf}
\def\fig{NF-qf-other-dim}
\begin{align*}
\zwf\vdash~\scalebox{\scaling}{
}
\eq{}\scalebox{\scaling}{
}
\end{align*}
\end{lemma}

\begin{proof}
\def\fig{NF-qf-other-dim}
\begin{align*}
\zwf\vdash~\scalebox{\scaling}{
}
&\eq{\ref{lem:ket-k-dim-qf}}\scalebox{\scaling}{
}
\eq{\nameref{ax:bialgebra-W-qf}}\scalebox{\scaling}{
}\\
&\eq{\nameref{ax:bialgebra-Z-W-qf}\\\ref{lem:zcp-1}}\scalebox{\scaling}{
}
\eq{\nameref{ax:W-assoc-qf}}\scalebox{\scaling}{
}
\qedhere
\end{align*}
\end{proof}

\begin{definition}
With $d\geq a$, we define the ``$a$-restricted Z-spider'' as:
\def\fig{Z-spider-restrict-alt}
\begin{align*}
\scalebox{\scaling}{
}
~~:=~~\frac1{a!}\cdot\scalebox{\scaling}{
}
\left(\eq{}\frac1{a!}\cdot\scalebox{\scaling}{
}\right)
\end{align*}
\end{definition}
Notice that the $a$-restricted Z-spider by default only uses capacity $d$ on all its wires.

\begin{lemma}
With $d\geq a$:
\label{lem:Z-restrict-NF-qf}
\def\fig{Z-restrict-qudit}
\begin{align*}
\zwf\vdash~\scalebox{\scaling}{
}
\eq{\ref{lem:ket-k-forms}}\scalebox{\scaling}{
}
\end{align*}
\end{lemma}

\begin{proof}
\def\fig{Z-restrict-qudit}
\begin{align*}
\zwf\vdash~a!\cdot\scalebox{\scaling}{
}
&\eq{\nameref{ax:decomp-Z-qf}}\scalebox{\scaling}{
}
\eq{\ref{lem:embedding}}\scalebox{\scaling}{
}
\eq{\nameref{ax:bialgebra-Z-W-qf}}\scalebox{\scaling}{
}
\eq{\nameref{ax:bialgebra-W-qf}}\scalebox{\scaling}{
	\input{./figures/\fig/\fig_05.tikz}%
}\\
&\eq{\ref{lem:ket-k-dim-qf}}\scalebox{\scaling}{
	\input{./figures/\fig/\fig_06.tikz}%
}
\eq{\ref{lem:ket-k-dim-qf}}\scalebox{\scaling}{
	\input{./figures/\fig/\fig_07.tikz}%
}
\eq{}a!\cdot\scalebox{\scaling}{
}
\qedhere
\end{align*}
\end{proof}

\begin{lemma}
\label{lem:Z-loop-removal-qf}
\def\fig{Z-loop-removal-qf-prf}
\begin{align*}
\zwf\vdash~\scalebox{\scaling}{
}
\eq{}\scalebox{\scaling}{
	\input{./figures/\fig/\fig_06.tikz}%
}
\end{align*}
\end{lemma}

\begin{proof}
\def\fig{Z-loop-removal-qf-prf}
\begin{align*}
\zwf\vdash~\scalebox{\scaling}{
}
&\eq{\ref{lem:ket-k-dim-qf}\\\nameref{ax:id-qf}}\scalebox{\scaling}{
}
\eq{\nameref{ax:bialgebra-W-qf}}\scalebox{\scaling}{
}
\eq{\nameref{ax:bialgebra-Z-W-qf}}\scalebox{\scaling}{
}
\eq{\nameref{ax:Z-spider-qf}\\\ref{lem:embedding}}\scalebox{\scaling}{
}\\
&\eq{\nameref{ax:loop-removal-qf}\\\nameref{ax:Z-spider-qf}\\\ref{lem:Z-id-qf}}\scalebox{\scaling}{
	\input{./figures/\fig/\fig_05.tikz}%
}
\eq{\nameref{ax:bialgebra-W-qf}\\\nameref{ax:id-qf}\\\ref{lem:ket-k-dim-qf}}\scalebox{\scaling}{
	\input{./figures/\fig/\fig_06.tikz}%
}
\qedhere
\end{align*}
\end{proof}

\begin{lemma}
\label{lem:discard-ket-qf}
With $a\geq k$:
\def\fig{braket-k-qf}
\begin{align*}
\zwf\vdash~\scalebox{\scaling}{
}
\eq{}k!
\qquad\text{and}\qquad
\def\fig{discard-ket-qf}
\zwf\vdash~\scalebox{\scaling}{
}
\eq{}\scalebox{\scaling}{
}
\end{align*}
\end{lemma}

\begin{proof}
First:
\def\fig{braket-k-qf}
\begin{align*}
\zwf\vdash~\scalebox{\scaling}{
}
&\eq{\ref{lem:ket-k-forms}}\scalebox{\scaling}{
}
\eq{\nameref{ax:scalar-qf}}\scalebox{\scaling}{
}
\eq{\nameref{ax:bialgebra-W-qf}}\scalebox{\scaling}{
}
\eq{\ref{lem:Pascal-qf}}\scalebox{\scaling}{
}\\
&\eq{\nameref{ax:W-assoc-qf}\\\nameref{ax:hopf-qf}\\\ref{lem:ket-0-qf}}\scalebox{\scaling}{
	\input{./figures/\fig/\fig_05.tikz}%
}
\eq{\nameref{ax:scalar-qf}\\\ref{lem:scalar-1-qf}}k!
\end{align*}
\def\fig{discard-ket-qf}
\begin{align*}
\zwf\vdash~\scalebox{\scaling}{
}
&\eq{}\frac1{k!}\cdot\scalebox{\scaling}{
}
\eq{\ref{lem:ket-k-forms}\\\nameref{ax:scalar-qf}\\\nameref{ax:bialgebra-Z-W-qf}}\frac1{k!}\cdot\scalebox{\scaling}{
}
\eq{\nameref{ax:Z-spider-qf}\\\ref{lem:Z-id-qf}}\frac1{k!}\cdot\scalebox{\scaling}{
}
\eq{}1
\eq{\ref{lem:scalar-1-qf}}\scalebox{\scaling}{
}
\qedhere
\end{align*}
\end{proof}

\begin{proposition}
\label{prop:qufinite-derives-qudit}
All equations of \zw are derivable from \zwf, i.e.
\[\forall d\geq 2,~\zw\vdash D_1=D_2 \implies \zwf\vdash \iota_d(D_1) = \iota_d(D_2)\]
\end{proposition}

\begin{proof}
First, notice that the result is obviously true for all axioms of compact-closed props. Then, it is enough to show the result for the equations in \Cref{fig:equational-theory}, as all equations provable with \zw derive from them. Equations \nameref{ax:Z-spider}, \nameref{ax:W-unit}, \nameref{ax:W-assoc}, \nameref{ax:bialgebra-Z-W}, \nameref{ax:sum}, \nameref{ax:hopf}, \nameref{ax:copy} and \nameref{ax:decomp-Z} are directly translated to an equation of \zwf through $\iota_d$. Then, \nameref{ax:bialgebra-W} is exactly \Cref{lem:W-bialgebra-ctxt}, \nameref{ax:loop} is \Cref{lem:Z-loop-removal-qf}, and \nameref{ax:one} is \Cref{lem:discard-ket-qf}.
\end{proof}

\begin{proposition}
\label{prop:qufinite-to-qudit}
Let $D\in\cat{ZW_f}$ be a diagram (state) with its largest capacity being $d-1$. Then, there exists $D_d\in\cat{ZW}_d$ such that:
\[\def\fig{qufinite-to-qudit}
\zwf\vdash~~\scalebox{\scaling}{
}
\eq{}\scalebox{\scaling}{
}\]
\end{proposition}

\begin{proof}
Let us write $\delta := d-1$ for the maximum capacity in $D$. Using \Cref{lem:embedding}, we can ``force'' capacity $\delta$ on all wires:
\def\fig{embedding-delta}
\begin{align*}
\zwf\vdash~\scalebox{\scaling}{
}
\eq{\ref{lem:embedding}}\scalebox{\scaling}{
}
\end{align*}
Every original generator in $D$ is then surrounded by W-nodes (connected to them by their output). These can be turned into purely qudit systems as follows, first for Z-spiders with degree $\geq1$:
\def\fig{embedding-Z}
\begin{align*}
\zwf\vdash~\scalebox{\scaling}{
}
\eq{\nameref{ax:Z-spider-qf}}\scalebox{\scaling}{
}
\eq{\nameref{ax:bialgebra-Z-W-qf}}\scalebox{\scaling}{
}
\eq{\ref{lem:Z-restrict-NF-qf}}\scalebox{\scaling}{
}
\end{align*}
The case of the Z-spider with no leg has to be dealt with differently:
\def\fig{embedding-Z-0}
\begin{align*}
\zwf\vdash~\scalebox{\scaling}{
}
\eq{\nameref{ax:Z-spider-qf}}\scalebox{\scaling}{
}
\eq{\ref{lem:embedding}}\scalebox{\scaling}{
}
\eq{\ref{lem:Z-restrict-NF-qf}}\scalebox{\scaling}{
}
\end{align*}
The case of the W-node is given by:
\def\fig{embedding-W}
\begin{align*}
\zwf\vdash~\scalebox{\scaling}{
}
\eq{\nameref{ax:W-assoc-qf}}\scalebox{\scaling}{
}
\eq{\nameref{ax:bialgebra-W-qf}}\scalebox{\scaling}{
}
\eq{\nameref{ax:W-assoc-qf}\\\ref{lem:Z-id-qf}}\scalebox{\scaling}{
}
\eq{}\scalebox{\scaling}{
}
\end{align*}
The case of $
\scalebox{\scaling}{
	\begin{tikzpicture}
	\begin{pgfonlayer}{nodelayer}
		\node [style=ket] (0) at (-0.25, 0.25) {$k$};
		\node [style=none] (1) at (-0.25, -0.25) {};
	\end{pgfonlayer}
	\begin{pgfonlayer}{edgelayer}
		\draw (1.center) to (0);
	\end{pgfonlayer}
\end{tikzpicture}
}$ is dealt with Lemmas \ref{lem:ket-k-dim-qf} and \ref{lem:ket-k-forms}. 
Finally, using \nameref{ax:Z-spider-qf} with \ref{lem:Z-restrict-NF-qf}, all ``$a$-restricted'' Z-spiders $\def\fig{embedding-Z}\scalebox{\scaling}{
}$ can be turned into diagrams that only use capacity $\delta$.
Doing so for all generators, we consume all the W-nodes created at the beginning of the proof through \Cref{lem:embedding}, except the ones that are connected to the outputs of the diagram. We hence end up in the form $\def\fig{qufinite-to-qudit}\scalebox{\scaling}{
}$.
\end{proof}

\begin{proposition}
\label{prop:qudit-NF-to-qufinite-NF}
Let $D_d$ be a $\cat{ZW}_d$-state in (qudit) normal form. Then $\def\fig{qufinite-to-qudit}\scalebox{\scaling}{
}$ can be put in (mixed-dimensional) normal form.
\end{proposition}

\begin{proof}
First, we may use \Cref{lem:NF-qf-other-dim-qf} to turn all ``internal'' capacities in $\iota_d(D_d)$ into $1$. $\iota_d(D_d)$ is hence technically in mixed-dimensional normal form. It remains to remove the W-nodes $\obfhilb{\delta} \to \obfhilb{a_i}$ at the bottom of the diagram. This can be done as follows, considering each output individually:
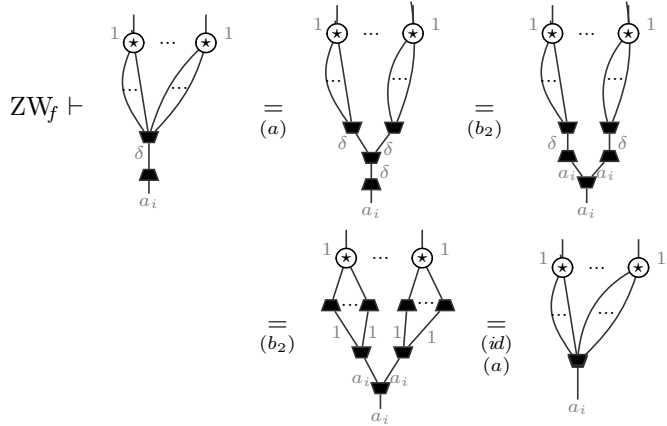
\def\fig{qudit-NF-to-qufinite-NF}
\begin{align*}
\zwf\vdash~\scalebox{\scaling}{
}
&\eq{\nameref{ax:W-assoc-qf}}\scalebox{\scaling}{
}
\eq{\nameref{ax:bialgebra-W-qf}}\scalebox{\scaling}{
}\\
&\eq{\nameref{ax:bialgebra-W-qf}}\scalebox{\scaling}{
}
\eq{\nameref{ax:id-qf}\\\nameref{ax:W-assoc-qf}}\scalebox{\scaling}{
}
\end{align*}
Finally, we may remove Z-spiders connected $k$ times to output $a_i$ when $k>a_i$ using Equation \nameref{ax:hopf-qf} followed by \Cref{lem:copy-0-qf} and Equation \nameref{ax:id-qf}.
\end{proof}

\begin{proof}[Proof of \Cref{thm:completeness-finite}]
We can now show that any $\cat{ZW_f}$-state $D$ can be put in normal form. First, use \Cref{prop:qufinite-to-qudit} to turn all but the outputs of the diagram into a qudit diagram. Using \Cref{prop:qufinite-derives-qudit} and \Cref{thm:completeness-qudit}, we can turn the qudit-part of that diagram into qudit normal form. Finally, using \Cref{prop:qudit-NF-to-qufinite-NF}, the whole diagram can be put in mixed-dimensional normal form.

If two $\cat{ZW_f}$-states $D_1$ and $D_2$ are semantically equivalent ($\interp{D_1}=\interp{D_2}$), then they can both be turned into normal form. By uniqueness of the normal form, the two diagrams are equal, i.e.~$\interp{D_1}=\interp{D_2} \implies \zwf\vdash D_1=D_2$.
\end{proof}

\end{document}